\documentclass[11pt]{article}
\usepackage[T1]{fontenc}
\usepackage{amsmath,amsxtra,amssymb,latexsym,amscd,amsthm,pb-diagram,fancyhdr,euscript}
\usepackage{amsthm}
\usepackage{indentfirst}
\usepackage{epsfig}
\usepackage{mathrsfs}
\usepackage{slashed}
\usepackage{hyperref}
\hypersetup{
   colorlinks,
    citecolor=blue,
    filecolor=black,
    linkcolor=red,
    urlcolor=magenta
}
\hypersetup{linktocpage}

\newcommand{\R}{\mathbb{R}} \newcommand{\N}{\mathbb{N}}

\newcommand{\T}{\mathbb{T}}
\newcommand{\J}{\mathbb{J}}
\newcommand{\E}{\mathbb{E}}

\renewcommand{\d}{\mathrm{d}}

\newcommand{\D}{\mathrm{D}}

\newcommand{\dvol}{\mathrm{dVol}}

\newcommand{\scri}{{\mathscr I}}

\newcommand{\ba}{\mathbf{a}}
\newcommand{\bb}{\mathbf{b}}

\newcommand{\hook}{{\setlength{\unitlength}{11pt}   
                   \begin{picture}(.833,.8)
                   \put(.15,.08){\line(1,0){.35}}
                   \put(.5,.08){\line(0,1){.5}}
                   \end{picture}}}
\newtheorem{definition}{Definition}[section]
\newtheorem{theorem}{Theorem}[section]
\newtheorem{proposition}{Proposition}[section]
\newtheorem{corollary}{Corollary}[section]
\newtheorem{lemma}{Lemma}[section]
\newtheorem{remark}{Remark}[section]
\begin{document}
\mbox{} \thispagestyle{empty}

\begin{center}
\bf{\Huge Peeling on Kerr spacetime~:\\
linear and non linear scalar fields}

\vspace{0.1in}

{Jean-Philippe NICOLAS\footnote{LMBA, UMR CNRS 6205, Department of Mathematics, University of Brest, 6 avenue Victor Le Gorgeu, 29200 Brest, France.
Email~: Jean-Philippe.Nicolas@univ-brest.fr} \& PHAM Truong Xuan\footnote{Faculty of Information Technology, Department of Mathematics,
Thuyloi university, Khoa Cong nghe Thong tin, Bo mon Toan, Dai hoc Thuy loi
175 Tay Son, Dong Da, Hanoi, Vietnam}}
\end{center}

{\bf Abstract.} We study the peeling on Kerr spacetime for fields satisfying conformally invariant linear and nonlinear scalar wave equations. We follow the approach initiated by L.J. Mason and the first author \cite{MaNi2009,MaNi2012} for the Schwarzschild metric, based on a Penrose compactification and energy estimates. This approach provides a definition of the peeling at all orders in terms of Sobolev regularity near $\scri$ instead of ${\cal C}^k$ regularity at $\scri$, allowing to characterise completely and without loss the classes of initial data ensuring a certain order of peeling at $\scri$. This paper extends the construction to the Kerr metric, confirms the validity and optimality of the flat spacetime model (in the sense that the same regularity and fall-off assumptions on the data guarantee the peeling behaviour in flat spacetime and on the Kerr metric) and does so for the first time for a nonlinear equation. Our results are local near spacelike infinity and are valid for all values of the angular momentum of the spacetime, including for fast Kerr metrics.

\vspace{0.1in}


\vspace{0.1in}


\tableofcontents

\section{Introduction}
The peeling is a type of asymptotic behaviour of zero rest-mass fields discovered by R. Sachs for spin $1$ and spin $2$ fields in the flat case in \cite{Sa61} and then established in the asymptotically flat case in \cite{Sa62}. An outgoing zero rest-mass field can be expanded in powers of $1/r$ along a null geodesic going out to future null infinity~; for a field of spin $s$, the part that falls off like $1/r^k$ has $2s-k$ of its principal null directions aligned along the null geodesic. E.T. Newman and R. Penrose \cite{NP} and then R. Penrose \cite{Pe1963,Pe1964,Pe1965} explored the notion further using the conformal method and the spin-coefficient formalism (now known as the Newman-Penrose formalism). In \cite{Pe1965}, the peeling is shown to be equivalent to an apparently much simpler property~: the boundedness of the conformally rescaled field at null infinity. Moreover, Penrose suggests in this paper that the behaviour in flat spacetime extends generically to asymptotically flat situations and he defines precisely the geometrical framework in which this behaviour is supposed to occur~: asymptotically simple spacetimes whose regularity at null infinity encodes the peeling of the Weyl tensor. This suggestion caused some controversy because of Schwarzschild's spacetime, which is the first non trivial asymptotically flat solution of the Einstein equations and, the model for the leading order fall-off in the asymptotic expansion of general asymptotically flat spacetimes. Its asymptotic structure is known to be quite different from that of Minkowski spacetime, with in particular a singular conformal structure at spacelike infinity. Many felt that this difference should impose more stringent constraints on the regularity and fall-off of initial data in order to ensure the peeling of the solution~; as a consequence the classes of data ensuring peeling would be much smaller in the Schwarzschild case than in the flat case and naturally, this would also be true generically for asymptotically flat spacetimes. In short, it was felt that the peeling was not a generic behaviour. In particular, the peeling of the Weyl tensor was not generic and therefore asymptotically simple spacetimes, in addition to having smaller classes of test fields with a peeling behaviour than flat spacetime, were not a generic model for asymptotically flat spacetimes.

Peeling estimates for test fields on a given background, i.e. pointwise decay estimates for the various components of the fields in a null frame, were obtained using integrated energy estimates and Sobolev embeddings, without conformal compactification~; this approach was initiated by D. Christodoulou and S. Klainerman \cite{ChriKla90} for the wave equation on Minkowski spacetime and then used notably by W. Inglese and F. Nicolò \cite{InglNi} and J. Metcalfe, D. Tataru and M. Tohaneanu \cite{MeTaTo} for the Maxwell field on black hole spacetimes. Also, the peeling of the Weyl tensor and the related question of asymptotic simplicity, have been resolved using different approaches by D. Christodoulou and S. Klainerman \cite{ChriKla}, P. Chru\'sciel and E. Delay \cite{ChruDe2002, ChruDe2003}, J. Corvino \cite{Co2000}, J. Corvino and R. Schoen \cite{CoScho2003}, H. Friedrich (see \cite{HFri2004} for a survey of his work on the subject) and S. Klainerman and F. Nicol\`o \cite{KlaNi,KlaNi2002,KlaNi2003}. But the validity of the second part of Penrose's original suggestion, concerning the genericity of the peeling behaviour for zero-rest-mass fields, is still not entirely clear. The question is whether generic asymptotically flat spacetimes admit large enough classes of initial data giving rise to solutions satisfying the peeling. Recently, L. Mason and one of the authors (JPN) provided a complete answer for scalar, Dirac and Maxwell fields on the Schwarzschild metric in \cite{MaNi2009,MaNi2012}~;  their results confirm Penrose's insights.

The approach developed in \cite{MaNi2009,MaNi2012} is based on three essential principles.
\begin{enumerate}
\item The work is done on the compactified spacetime. The peeling in characterised, in the spirit of Penrose, in terms of local regularity at $\scri$ (null infinity) instead of in terms of decay along outgoing null directions. Energy estimates, performed directly on the conformally rescaled field, give an optimal control of the regularity of the rescaled solution at $\scri$ in terms of the regularity and fall-off of the initial data. The optimality is guaranteed by the fact that the estimates go both ways. This provides a definition of the peeling at any given order in terms of weighted Sobolev norms on $\scri$ of transverse derivatives of the solution. Pointwise decay estimates as well as ${\cal C}^k$ norms at $\scri$ are avoided, because they are obtained via Sobolev embeddings which lose derivatives and prevent an optimal control of the classes of data ensuring a given regularity at $\scri$.
\item The estimates are performed in the neighbourhood of $i^0$ (spacelike infinity). If the fall-off at spacelike infinity of the initial data is not strong enough, singularities can appear at $i^0$ that will creep along null infinity. The essential difficulty is to understand precisely how to avoid such singularities by specifying the fall-off and regularity of the data near spacelike infinity. Once we control the regularity of the rescaled solution at $\scri$ near $i^0$, its regularity on the rest of null infinity can easily be inferred from the regularity of the solution in the bulk of the spacetime.
\item Peeling on flat spacetime is redefined with a weaker compactification. The conformal compactifications available on generic asymptotically flat spacetimes are not complete. They allow to construct $\scri$ but usually $i^0$ and the timelike infinities $i^\pm$ are singularities of the conformal structure and are therefore not included in the compactified spacetime. Typically, on the Schwarzschild metric, the compactification is performed using a rescaling of the metric by $1/r^2$. In order to be able to compare the classes of data ensuring peeling on, say, the Schwarzschild metric and flat spacetime, we must make sure that the definitions of peeling are comparable. This would not be the case if the full compactification of Minkowski spacetime (for which we have a smooth embedding in the Einstein cylinder) were to be used~; the problem is essentially with the weights at $i^0$ appearing in the weighted Sobolev norms, not with the local regularity. In the comparison with the flat case, the peeling on Minkowski spacetime is recast using a partial compactification with a rescaling of the metric by $1/r^2$~; this means that the flat spacetime model and the asymptotically flat spacetime considered are compactified in a similar way and the comparison of the results is therefore meaningful. This choice provides in fact a larger class\footnote{Larger than the one obtained using the embedding in the Einstein cylinder.} of data giving rise to solutions that peel at a given order. It also guarantees that the corresponding classes in the Schwarzschild case are comparable to the flat case ones, in the sense that they are characterised by the same regularity and fall-off at infinity. This is simply because of the local uniformity of the estimates with respect to the mass of the spacetime (for more details, see \cite{MaNi2009,MaNi2012}).
\end{enumerate}
In the present work, we extend the results of \cite{MaNi2009,MaNi2012} to conformally invariant linear and nonlinear scalar wave equations on the Kerr metric. In the slow and extreme cases, the exterior of the black hole is globally hyperbolic and we have a well-posed Cauchy problem for the linear equation. For the nonlinear equation, the well-posedness of the Cauchy problem outside the black hole has been proved by one of the authors in \cite{Ni2002} and although the question is still open in the extreme case, it is expected that well-posedness also holds. A study of the peeling in these cases amounts to controlling the behaviour of the field near $i^0$, as in \cite{MaNi2009,MaNi2012}. We define a neighbourhood of $i^0$ in the compactified exterior of the black hole, in which we perform energy estimates for the conformally rescaled field, both ways between future null infinity and the $\{ t=0 \}$ slice (a time reversal and a change of sign of the angular momentum leave the spacetime unchanged, hence our construction is also valid for past null infinity). For a given regularity at $\scri$, these estimates give us the corresponding weighted Sobolev norms (or nonlinear functionals in the nonlinear case) of the initial data, whose finiteness will ensure the required regularity at $\scri$ of the rescaled field. The energy current used for the estimates is associated with a vector field that is timelike on the whole neighbourhood of $i^0$ including at $\scri$~: the Morawetz vector field, or rather an analogue of it on the Kerr exterior. This vector field was introduced in flat spacetime by C. Morawetz \cite{Mo1961} in the early 1960's to obtain pointwise decay estimates for the wave equation. A natural transposition of it on the Schwarzschild metric was used for the first time by W. Inglese and F. Nicolò \cite{InglNi}, then among others by M. Dafermos and I. Rodnianski, notably in \cite{DaRoLN} and \cite{DaRo2009}, and by L.J. Mason and JPN in \cite{MaNi2009, MaNi2012}. In the flat case the Morawetz vector field is a conformal Killing vector and is exactly Killing for the metric rescaled by $1/r^2$. On the Schwarzschild metric rescaled with the same conformal factor, the Killing form of the analog of the Morawetz vector field decays rapidly, like $1/r^2$, at spacelike and null infinity, making it an approximate conformal Killing vector field of Schwarzschild's spacetime. On the Kerr metric however, the Killing form of our version of the Morawetz vector field does not decay at infinity, but it behaves well enough to allow the estimates to close.

Our estimates are in fact valid for all values of the angular momentum and the mass of the spacetime, including for fast Kerr metrics. In this case however, the Cauchy problem is not known to be well-posed or even to make sense since the spacetime is no longer globally hyperbolic and even contains a time machine~; our results are therefore valid in the neighbourhood of $i^0$ that we consider, which is globally hyperbolic, and they extend to any globally defined and sufficiently regular solution.

The paper is organised as follows. Section \ref{GeomAnalSet} contains a brief description of the Kerr metric and its conformal compactification. Then we construct explicitly a neighbourhood of spacelike infinity on the compactified spacetime~; we define it by its boundary, composed of a part of the $\{ t=0 \}$ slice, a part of $\scri^+$ and a null hypersurface $\cal S$ generated by a family of surface-forming outgoing null geodesics, known as the simple null geodesics (see \cite{FleLu2003} and \cite{Ha2009}). We also define a foliation of this neighbourhood by a family of spacelike hypersurfaces that accumulate on $\scri^+$. Then the Morawetz vector field is defined and its Killing form written in details and finally the linear and nonlinear conformal wave equations are presented with some remarks on the Cauchy problem in our neighbourhood of $i^0$. In Section \ref{Lin} we study the peeling for the linear conformal wave equation. We establish energy estimates both ways between the union of $\scri^+$ and $\cal S$ and the $\{ t=0 \}$ slice, at all orders of regularity~: we prove a fundamental estimate first and then use a set of vector fields generating the tangent bundle to obtain higher order estimates. The energy norms on $\scri^+$ appearing in the estimates define the classes of regularity characterising the different orders of peeling. The energy norms on the initial data slice define, by completion of smooth compactly supported data, the corresponding classes of data that ensure the required order of peeling for the solution. Note that the energy on $\{ t=0 \}$ involves derivatives transverse to the initial data slice, but using the equation we turn them into combinations of tangential derivatives~; this is how we obtain our norms on the initial data. Similarly to the Schwarzschild case, the estimates are uniform in the mass and angular momentum of the spacetime on any compact domain of the form $[0,M] \times [-a,a]$. This guarantees that the classes of data are comparable to the corresponding ones on Minkowski spacetime in the sense defined above. Section \ref{NLCase} is devoted to the nonlinear conformal wave equation. The approach is similar but now we use stress-energy tensors associated to both the linear and the nonlinear equations in our estimates. We establish a Sobolev embedding, uniformly on the slices of the foliation, that allows us to control the non-linear energy by (a function of) the linear energy. We obtain a characterisation of the orders of peeling in terms of the linear energy of successive derivatives. Similarly, the classes of data giving rise to a certain order of peeling are defined by the finiteness of the linear energy of successive derivatives on the initial data slice. However the situation is a little more subtle than in the linear case. When we use the equation to express the transverse derivatives as combinations of tangential ones, we get a nonlinear expression. Consequently, our classes of data are characterised by the finiteness of nonlinear functionals and are not vector spaces. The uniformity with respect to the mass and angular momentum are assured and the classes are therefore comparable to the flat space ones.

\vspace{0.1in}

\noindent {\bf Notations and conventions.}
\begin{enumerate}
\item Concerning differential forms and Hodge duality, following R. Penrose and W. Rindler \cite{PeRi}, we adopt the following convention~: on a spacetime $({\cal M},g)$ (i.e. a $4$-dimensional Lorentzian manifold that is oriented and time-oriented), the Hodge dual of a $1$-form $\alpha$ is given by
\begin{equation} \label{Star}
(*\alpha)_{abc} = e_{abcd} \alpha^d\, ,
\end{equation}
where $e_{abcd}$ is the volume form on $({\cal M} , g)$, which in this paper we simply denote $\dvol$. We shall use two important properties of the Hodge star~:
\begin{itemize}
\item given two $1$-forms $\alpha$ and $\beta$, we have
\begin{equation} \label{HStarP1}
\alpha \wedge * \beta = -\frac{1}{4} \alpha_a  \beta^a\, \dvol  \, ;
\end{equation}
\item for a $1$-form $\alpha$ that is differentiable,
\begin{equation} \label{HStarP2}
\d * \alpha =  -\frac{1}{4} (\nabla_a  \alpha^a ) \dvol  \, .
\end{equation}
\end{itemize}
\item We shall use the notation $\lesssim$ to signify that the left-hand side is bounded above by a positive constant times the right-hand side, the constant being independent of the parameters and functions appearing in the inequality. The notation $\simeq$ means that both $\lesssim$ and $\gtrsim$ are valid.
\end{enumerate}

\noindent{\bf Acknowledgements.} This research was partly supported by the ANR funding ANR-12-BS01-012-01.

\section{Geometrical and analytical setting} \label{GeomAnalSet}

\subsection{The Kerr metric}

The Kerr metric is given in Boyer-Lindquist coordinates $(t,r,\theta , \varphi ) \in \R_t \times \R_r \times [0,\pi ] \times [0,2\pi [$ as
\begin{equation}\label{Kerr-metric0}
 g=\left(1-\frac{2Mr}{\rho^2}\right)\d t^2 +\frac{4aMr\sin^2\theta}{\rho^2}\d t\d \varphi -\frac{\rho^2}{\Delta}\d r^2-\rho^2\d \theta^2-\frac{\sigma^2}{\rho^2}\sin^2\theta\d \varphi^2, 
\end{equation}
$$ \rho^2=r^2+a^2\cos^2\theta,\Delta=r^2-2Mr+a^2, $$
$$ \sigma^2=(r^2+a^2)\rho^2+2Mra^2\sin^2\theta=(r^2+a^2)^2-\Delta a^2\sin^2\theta \; ,$$
where $M>0$ is the mass of the black hole and $a\neq 0$ its angular momentum per unit mass. If $a=0$, $g$ reduces to the Schwarzschild metric, and if $M=a=0$, we obtain the Minkowski metric. Kerr's spacetime is asymptotically flat and possesses two independent Killing vector fields $\partial_t$ and $\partial\varphi$ which span the space of all Killing vector fields for $g$.

The expression \eqref{Kerr-metric0} of the Kerr metric has two types of singularities~: the set of points where $\rho^2=0$ (the equatorial ring $\left\{ r=0\, ,~\theta = \pi/2 \right\}$ of the $\left\{r=0\right\}$ sphere) is a true curvature singularity~; the spheres where $\Delta$ vanishes are mere coordinate singularities and correspond to event horizons.

There are three types of Kerr spacetimes depending on the respective values of $M$ and $a$~:
\begin{itemize}
\item[$\bullet$] Slow Kerr spacetime for $0<|a|<M$. $\Delta$ has two real roots
$$r_\pm = M \pm \sqrt{M^2-a^2} \, ;$$
we have two horizons on either different side of the hypersurface $\left\{ r=M\right\}$. The sphere $\left\{ r=r_-\right\}$ is called the inner horizon and the sphere $\left\{ r=r_+\right\}$ is the outer or black hole horizon. The horizons separate the spacetime in three types of domains referred to as blocks. Block I, the exterior of the black hole, is the region $\{ r>r+\}$~; it is locally but not globally stationary since it admits no global timelike Killing vector field. Block II is the dynamic region between the two horizons $\{ r_- < r < r_+\}$~; in this region $\partial_t$ is spacelike and $\partial_r$ is timelike. Block III is the region beyond the inner horizon $\{ r<r_-\}$~; it contains the singularity and a time machine where $\partial_\varphi$ becomes timelike.
\item[$\bullet$] Extreme Kerr spacetime for $|a|=M$. $M$ is then the double root of $\Delta$ and the sphere $\left\{ r=M\right\}$ is the only horizon. Only blocks I and III exist in this case.
\item[$\bullet$] Fast Kerr spacetime for $|a|>M$. $\Delta$ has no real root and the spacetime has no horizon. There is no black hole in this case and the singularity at $r=0$ is naked.
\end{itemize}
In this paper, we do not specify the respective values of $a$ and $M$~; we work in a neighbourhood of spacelike infinity that is sufficiently far away from the black hole or the singularity. We shall describe this neighbourhood in more details in Subsection \ref{Neighbourhood}.

Kerr spacetime has Petrov type D (see \cite{Petrov})~: its Weyl spinor has two double principle null directions
\[ V^\pm = \frac{r^2+a^2}{\Delta}\partial_t \pm \partial_r + \frac{a}{\Delta}\partial_\varphi \, . \]
Introducing a Regge-Wheeler-type variable $r_*$ such that
\begin{equation} \label{rstar}
\frac{\d r_*}{\d r} = \frac{r^2+a^2}{\Delta} \, ,
\end{equation}
the principle null directions take the form
\begin{equation} \label{PNG}
V^\pm = \frac{r^2+a^2}{\Delta} (\partial_t \pm \partial_{r_*} ) + \frac{a}{\Delta}\partial_\varphi \, .
\end{equation}
\begin{remark}
In the slow case, the variable $r_*$, modulo an arbitrary constant of integration, has the following expression
\[r_* = r +\frac{2Mr_+}{r_+-r_-} \log \vert r-r_+ \vert + \frac{2Mr_-}{r_--r_+} \log \vert r-r_- \vert \, .\]
In the extreme case, we have
\[ r_* = r + 2M\log \vert r-M \vert - \frac{2M^2}{r-M} \]
and in the fast case
\[ r_* = r+M \log (r^2+a^2-2Mr) + \frac{2M^2}{\sqrt{a^2-M^2}} \arctan \left( \frac{r-M}{\sqrt{a^2-M^2}} \right) \, .\]
We see in particular that in all three case, for $r$ large enough, we have $r_* > r$ and
\[ \lim_{r\rightarrow +\infty} \frac{r_*}{r} =1 \, .\]
\end{remark}
The Goldberg-Sachs theorem \cite{GS} states that multiple principle null directions of the Weyl spinor generate shear-free null geodesic congruences. The integral lines of $V^+$ and $V^-$ are therefore null geodesics (referred to as principal null geodesics) and their congruences are shear-free, but they have a non trivial twist~; this means that principle null geodesics are not surface forming, in other words, the principle null vectors $V^+$ and $V^-$ are not proportional to gradient fields. The simple null geodesics (see \cite{FleLu2003} and also \cite{Ha2009}) are another family of null geodesics that forms two congruences, an incoming and an outgoing one, that are non twisting, i.e. they are surface forming. The principle null geodesics are used to construct Eddington-Finkelstein-like coordinate systems, star-Kerr and Kerr-star coordinates, which allow to understand event horizons as smooth, totally geodesic, null hypersurfaces (see for example the remarkable book by B. O'Neill \cite{O95}). In Subsection \ref{StarKerr}, we present the Penrose compactification of Kerr spacetime in terms of star-Kerr and Kerr-star coordinates. We shall use the simple null geodesics in Subsection \ref{Neighbourhood} in order to give an explicit example of a null hypersurface going out to future null infinity.

\subsection{Penrose compactification of Kerr spacetime} \label{StarKerr}

In order to construct future null infinity ($\scri^+$), we introduce the so-called star-Kerr coordinates $({^*}t,r,\theta,{^*\varphi})$ based on the family of outgoing principal null geodesics. The new coordinates $^*t$ and $^*\varphi$ are defined as follows
$${^*t} = t- r_* \; , \; {^*\varphi} = \varphi- \Lambda(r) \, ,$$
where the function $r_*$ is defined by \eqref{rstar} and $\Lambda$ satisfies
\[ \frac{\d \Lambda}{\d r}=\frac{a}{\Delta} \, .\]
Denoting by $(\partial_r)_{BL}$ and $(\partial_r)_{^*K}$ the $r$ coordinate vector fields in Boyer-Lindquist and star-Kerr coordinates respectively, we have
\begin{eqnarray*}
(\partial_r)_{^*K} &=& V^+ \, , \\
\partial_{^*t} &=& \partial_t \, ,\\
\partial_{^*\varphi} &=& \partial_\varphi \, , \\
(\partial_r)_{BL} &=& - \frac{r^2+a^2}{\Delta} \partial_{^*t}+(\partial_r)_{^*K} - \frac{a}{\Delta} \partial_{^*\varphi} \, .
\end{eqnarray*}
The outgoing principal null geodesics are therefore the $r$-coordinate lines in star-Kerr coordinates. The Kerr metric now takes the form
\[ g =\left(1-\frac{2Mr}{\rho^2}\right) \d {^*t}^2 + \frac{4aMr\sin^2\theta}{\rho^2} \d {^*t}\d {^*\varphi} -\frac{\sigma^2}{\rho^2}\sin^2\theta \d {^*\varphi^2}  - \rho^2 \d \theta^2 + 2 \d {^*t} \d r - 2 a \sin^2\theta \d {^*\varphi} \d r \, .\]
\begin{remark}
Note that we can see at this point that the horizons, if they are present, are not singularities of the metric. In the slow case, this form of the metric allows us to glue block II in the future of block I, the junction being the future outer horizon ${\mathbb R}_{^*t}\times \left\{r=r_+\right\}\times S_{\theta,{^*\varphi}}^2$ and to glue block III in the future of block II with junction the future inner horizon ${\mathbb R}_{^*t}\times \left\{r=r_-\right\}\times S_{\theta,{^*\varphi}}^2$, thus constructing a black hole. A similar picture is obtained in the extreme case with two blocks and only one horizon.
\end{remark}
Then, we perform an inversion in the $r$ variable, i.e. we work with the coordinates $^*t$, $R=1/r$, $\theta$, $^*\varphi$, and rescale the Kerr metric with the conformal factor $\Omega^2 = R^2$~; we obtain the metric
\begin{eqnarray}
\hat{g}: = R^2 g&=& R^2\left(1-\frac{2Mr}{\rho^2}\right)\d {^*t}^2+\frac{4MaR\sin^2\theta}{\rho^2}\d {^*t} \d {^*\varphi} \nonumber\\
&& -(1+a^2R^2\cos^2\theta)\d {\theta^2} -\left(1+a^2R^2+\frac{2Ma^2R\sin^2\theta}{\rho^2}\right)\sin^2\theta \d {^*\varphi^2}\nonumber\\
&&-2 \d {^*t} \d R+2a\sin^2\theta \d {^*\varphi} \d R \, . \label{RescMet}
\end{eqnarray}
The rescaled metric extends smoothly to the hypersurface $\{ R=0 \}$ which cas be added as a boundary to our spacetime. The metric $\hat{g}$ does not degenerate at this boundary since its determinant,
\[ \det \hat{g} = -R^4 \rho^4 \sin^2 \theta = -\left( 1 + a^2 R^2 \cos^2 \theta \right)^2 \sin^2 \theta \, ,\]
does not vanish for $R=0$. However, the restriction of $\hat{g}$ to the $\{ R=0 \}$ hypersurface
\[  \hat{g} \vert_{\{ R=0 \}} = -\d \theta^2 - \sin^2 \theta \d \, ^*\varphi^2 \, ,\]
is a $2$-metric (in fact it is minus the euclidean metric on the unit $2$-sphere) and is therefore degenerate. This shows that $\{ R=0 \}$ is a smooth null hypersurface. It is the set of limit points of outgoing principal null geodesics as $r\rightarrow +\infty$~; it is denoted $\scri^+$ and called future null infinity. We give here the form of the inverse rescaled metric $\hat{g}^{-1}$, which will be useful to us~:
\begin{eqnarray}
{\hat g}^{-1}&=&-\frac{1}{\rho^2}\left(r^2a^2\sin^2\theta\partial^2_{^*t}+2(r^2+a^2)\partial_{^*t}\partial_R+2ar^2\partial_{^*t}\partial_{^*\varphi}+2a\partial_R\partial_{^*\varphi} \right) \nonumber\\
&&-\frac{1}{\rho^2} \left(R^2\Delta\partial^2_R+r^2\partial_{\theta}^2+\frac{r^2}{\sin^2\theta}\partial^2_{^*\varphi}\right) \; . \label{inverse-metric}
\end{eqnarray}

Similarly, Kerr-star coordinates $(t^*,r,\theta,\varphi^*)$ are constructed using the incoming principal null geodesics. We have
\[ {t^*} = t + r_* \; , \; {\varphi^*} = \varphi + \Lambda \, ,\]
with the same functions $r_*$ and $\Lambda$ as before. Denoting by $(\partial_r)_{BL}$ and $(\partial_r)_{K^*}$ the $r$ coordinate vector fields in Boyer-Lindquist and Kerr-star coordinate respectively, we have
\begin{eqnarray*}
V^- &=& (\partial_r)_{K^*} \, , \\
\partial_{t^*} &=& \partial_t \, ,\\
\partial_{\varphi^*} &=& \partial_{\varphi} \, ,\\
(\partial_r)_{BL} &=& \frac{r^2+a^2}{\Delta} \partial_{t^*}+(\partial_r)_{K^*} + \frac{a}{\Delta} \partial_{\varphi^*} \, .
\end{eqnarray*}
The incoming principal null geodesics are the $r$-coordinate lines in Kerr-star coordinates. The Kerr metric then takes the form
\[ g =\left(1-\frac{2Mr}{\rho^2}\right) \d {t^*}^2 + \frac{4aMr\sin^2\theta}{\rho^2} \d {t^*}\d {\varphi^*} -\frac{\sigma^2}{\rho^2}\sin^2\theta \d {\varphi^*}^2  - \rho^2 \d \theta^2 - 2 \d {t^*} \d r + 2 a \sin^2\theta \d {\varphi^*} \d r \, . \]
\begin{remark}
This coordinate system also allows us to understand horizons as smooth null hypersurfaces at the boundary of blocks, but the hypersurfaces ${\mathbb R}_{t^*}\times \left\{r=r_+\right\}\times S_{\theta,{\varphi^*}}^2$ and ${\mathbb R}_{t^*}\times \left\{r=r_-\right\}\times S_{\theta,{\varphi^*}}^2$ are now the past outer and inner horizons and the extended spacetime describes a white hole. By combining the black hole and the white hole construction and using Kruskal-Szekeres-type coordinates to describe the spheres where past and future horizons meet, we can construct the maximally extended Kerr spacetime in the slow and extreme cases. For the complete construction, see \cite{O95}.
\end{remark}
Now performing a similar construction as in star-Kerr coordinates, we obtain the following expression of the rescaled metric
\begin{eqnarray}
\hat{g}&=& R^2\left(1-\frac{2Mr}{\rho^2}\right)\d {t^*}^2+\frac{4MaR\sin^2\theta}{\rho^2}\d {t^*} \d {\varphi^*} \nonumber\\
&&-\left(1+a^2R^2+\frac{2Ma^2R\sin^2\theta}{\rho^2}\right)\sin^2\theta \d {{\varphi^*}^2} -(1+a^2R^2\cos^2\theta)\d {\theta^2} \nonumber\\
&&+2 \d {t^*} \d R - 2a\sin^2\theta \d {\varphi^*} \d R \, .
\end{eqnarray}
This reveals that $\hat{g}$ extends smoothly to the hypersurface $\scri^-:= \mathbb{R}_{t^*}\times \left\{R=0\right\}\times S^2_{\theta,{\varphi^*}}$, which can be added as a boundary and is the set of past limit points of incoming principle null geodesics~; $\scri^-$ (pass null infinity) is a smooth null hypersurface in this extended spacetime.

We denote by $\bar{\cal M}$ the compactified spacetime made of the union of the exterior of the black hole and $\scri^+$ and $\scri^-$ in the slow and extreme cases, of the whole spacetime with $\scri^+$ and $\scri^-$ in the fast case.

Spacelike infinity, denoted $i^0$, can be described as the set of limit points of uniformly spacelike curves as $r\rightarrow +\infty$ or as the collection of infinities for $t=$constant slices. It is reached from $\scri^+$ in the limit $^*t \rightarrow -\infty$ and from $\scri^-$ in the limit $t^*\rightarrow +\infty$. It is not part of the compactified spacetime.

We shall work exclusively with the rescaled metric $\hat{g}$~; we denote by $\nabla$ the Levi-Civita connection for $\hat{g}$.

\subsection{A neighbourhood of spacelike infinity} \label{Neighbourhood}

We work in a neighbourhood of spacelike infinity where we shall establish energy estimates between the $t=0$ slice and null infinities. We do the calculations explicitely for $\scri^+$, the ones for past null infinity are identical. Therefore, we define a neighbourhood of $i^0$ only in the part of $\bar{\cal M}$ where $t$ is non negative. What matters is that the whole neighbourhood be sufficiently far out towards $i^0$, its exact definition is not so important, in particular we could be vague with the details of its boundary where it differs from $\scri^+$ or the $t=0$ slice. However, for the sake of clarity and in order that all energy fluxes in the estimates be precisely defined, we chose to work with an explicit example of such a neighbourhood.

Our neighbourhood of $i^0$ will be defined as the part of $\bar{\cal M}$ lying in the future of the $t=0$ slice, denoted $\Sigma_0$ and in the past of a certain null hypersurface. This null hypersurface is foliated by outgoing simple null geodesics, the definition of which we shall now recall.
\begin{remark}
One might think that a natural choice of hypersurface to close the boundary of our neighbourhood would be a level hypersurface of $^* t$. However, as the expression \eqref{inverse-metric} of the inverse metric in $(^*t , R, \theta , ^*\varphi )$ coordinates shows, such hypersurfaces are timelike since
\[ \hat{g}(\nabla ^*t , \nabla ^*t ) = \hat{g}^{-1} (\d t , \d t ) =  -\frac{r^2a^2\sin^2\theta}{\rho^2} <0 \, .\]
It is much better for our estimates to choose a spacelike or null hypersurface so that the energy flux across it will be non negative.
\end{remark}

The simple null geodesics are the null geodesics of Kerr spacetime with zero Carter constant and zero total angular momentum around the axis of symmetry~; they are defined as the $r$ coordinate lines in a generalised Bondi-Sachs coordinate system introduced in \cite{FleLu2003} and described also in details in \cite{Ha2009}. We denote this coordinate system by $(\hat{t},r,\hat{\theta} , \varphi )$. The variable $\hat{t}$ is defined by
\[ \hat{t} := t- \hat{r} \, ,\]
where
\begin{equation} \label{rhat}
\hat r := r_* + \int_{0}^{r_*} \left( \sqrt{1-\frac{a^2\Delta(s)}{(r(s)^2+a^2)^2}} - 1 \right) \d s + a\sin\theta \, ,
\end{equation}
$r(r_*)$ being the reciprocal function of $r\mapsto r_*$ and $\Delta (r_* ) = (r(r_*))^2 -2M r(r_*) + a^2$. The function $\hat{\varphi}$ is defined modulo a choice of constant of integration by
\[ \hat{\varphi} := \varphi - \int \frac{2aMr}{\Delta \sqrt{(r^2+a^2)^2-\Delta a^2}} \d r \, .\]
As for the variable $\hat{\theta} \in [0,\pi]$, it is defined in an implicit manner by the following equation
\[ \frac{1+\tanh \alpha \sin \hat{\theta}}{\tanh \alpha + \sin \hat{\theta}} = \sin \theta\, ,\]
where $\alpha$ is a primitive of $a / \sqrt{(r^2+a^2)^2-\Delta a^2}$.

One can perform the Penrose compactification of Kerr spacetime by inverting the variable $r$ in these generalized Bondi-Sachs coordinates (for details, see D. Häfner \cite{Ha2009}). The level hypersurfaces of $\hat{t}$ are null hypersurfaces with topology $\R \times S^2$ foliated by outgoing simple null geodesics~; they intersect the $t=0$ slice and $\scri^+$ at $2$-spheres. Note that the cuts of $\scri^+$ obtained by intersection with the level hypersurfaces of $\hat{t}$ as not the same as the ones obtained with $^*t$. Indeed, we have
\[ \hat{t} = \, ^* t  - \int_{0}^{r_*} \left( \sqrt{1-\frac{a^2\Delta(s)}{(r(s)^2+a^2)^2}} - 1 \right) \d s - a\sin\theta \, ; \]
the integral has a fixed finite value at infinity but the last term gives an extra $\theta$ dependence to $\hat{t}$.

For $^*t_0 <<-1$, we define the neighbourhood $\Omega^+_{^*t_0}$ of $i^0$ as the intersection in $\bar{\cal M}$ of the past of the null hypersurface
\[ \mathcal{S}_{^*t_0} = \left\{ \hat{t} = {^*t_0} \, , ~ t\geq 0 \right\} \, ,\]
with the future of the $t=0$ slice, i.e.
\[ \Omega^+_{^*t_0} = {\cal I}^- ( \mathcal{S}_{^*t_0} ) \cap \{ t \geq 0\} \, .\]
In $\Omega^+_{^*t_0} \setminus \scri^+$, we have
\begin{equation} \label{epsilonestimates}
1 < \frac{r_*}{r} < 1+\varepsilon \, \mbox{ and } ~ 0 < -^*t R \leq \frac{r_*}{r} < 1+\varepsilon \, ,
\end{equation}
where $\varepsilon$ can be made arbitrarily small by choosing $\vert ^*t_0\vert$ arbitrarily large.

We foliate $\Omega_{{^*t}_0}^+$ by the hypersurfaces
$${\cal H}_s \, = \, \left\{^*t = -sr_*; \, \hat{t} \leq {^*t_0} \right\} , \;  0 \leq s \leq 1 \, .$$
\begin{itemize}
\item[$\bullet$] For $s=1$ we have $t = 0$, from which we get that ${\cal H}_1$ is the part of the hypersurface $\Sigma_0 = \left\{t=0\right\}$ inside $\Omega_{{^*t}_0}^+$.
\item[$\bullet$] In order to understand the hypersurface ${\cal H}_0$, let us fix $^*t = ^*t_1$, $\theta = \theta_1$ and $^*\varphi = ^*\varphi_1$ and take $s$ to zero while imposing $^*t = -sr_*$. Then $r_*$ must tend towards infinity and we see that we are going out to infinity on the outgoing principal null geodesic corresponding to $^*t_1$, $\theta_1$ and $^*\varphi_1$. Hence, the corresponding point of ${\cal H}_0$ is the point $(^*t= ^*t_1,\, R=0, \, \theta = \theta_1, \, ^*\varphi = ^*\varphi_1)$ of $\scri^+$. It follows that ${\cal H}_0 = \scri^+ \cap \Omega_{{^*t}_0}^+$, which we also denote by $\scri_{^*t_0}^+$.
\end{itemize}
\begin{lemma}
Provided $-^*t_0$ is large enough, the hypersurfaces ${\cal H}_s$ are spacelike for $0<s\leq 1$.
\end{lemma}
\begin{proof} They are level hypersurfaces of the function
\[ f({^*t}, \, R,\, \theta,\, {^*\varphi}) = {^*t} + sr_* \, ,\]
so the co-normal $1-$form to ${\cal H}_s$ is
\[ \d f = \d ^*t + s \d r_* = \d {^*t} + s \frac{a^2+r^2}{\Delta}\d r = \d ^* t - \frac{s(a^2+r^2)}{\Delta R^2}\d R \]
and we have
\begin{eqnarray*}
-\rho^2{\hat g}^{-1}( \d f, \d f) &=& r^2a^2\sin^2\theta - 2\frac{(a^2+r^2)^2r^2}{\Delta}s + \frac{(a^2+r^2)^2r^2}{\Delta}s^2 \\
& =& r^2 \left( a^2\sin^2\theta -  2\frac{(a^2+r^2)^2}{\Delta}s + \frac{(a^2+r^2)^2}{\Delta}s^2   \right) \, .
\end{eqnarray*}
We consider the equation
$$a^2\sin^2\theta -  2\frac{(a^2+r^2)^2}{\Delta}s + \frac{(a^2+r^2)^2}{\Delta}s^2 = 0 \, ;$$
for $r$ such that $(a^2+r^2)^2 > \Delta$ ($r>1$ is enough), this equation has two solutions
$$0 < s_1 = 1 - \sqrt{1 - \frac{a^2\Delta\sin^2\theta}{(a^2+r^2)^2}} < 1 < s_2 = 1 + \sqrt{1 - \frac{a^2\Delta\sin^2\theta}{(a^2+r^2)^2}} \, .$$
In $\Omega_{{^*t}_0}^+$, provided $-^*t_0$ is large enough, we can assume that
\[ s_1 \leq \frac{a^2\Delta \sin^2 \theta}{(r^2+a^2)^2}  \]
and using \eqref{epsilonestimates} we obtain
\begin{eqnarray*}
\frac{s_1}{s} &=& \frac{r_* \left( 1 - \sqrt{1 - \frac{a^2\Delta\sin^2\theta}{(a^2+r^2)^2}}\right)}{-^*t} \\
&\leq & \frac{(1+\varepsilon)r}{-^*t_0} \frac{a^2\Delta\sin^2\theta}{(a^2+r^2)^2} \\
&\leq & r \frac{a^2}{r^2+a^2} < 1 \, .
\end{eqnarray*}
It follows that in $\Omega_{{^*t}_0}^+$, provided $-^*t_0$ is large enough, for $0<s\leq 1$ we have $s_1<s<s_2$. Therefore $-\rho^2 {\hat g}^{-1}( \d f, \d f) <0$ and ${\hat g}^{-1}( \d f, \d f) >0$. This concludes the proof.
\end{proof}
With this foliation, we choose an identifying vector field $\nu$ that satisfies $\nu(s)=1$~:
\[ \nu = r_*^2R^2 \frac{\Delta}{r^2+a^2}\vert ^*t \vert^{-1}\partial_R \, . \]
This will allow the decomposition of the $4-$volume measure $\mathrm{dVol}^4 = -(1+a^2R^2\cos^2\theta)\d{^*t}\d R\d^2\omega$ into the product of $\d s$ along the integral lines of $\nu^a$ and the $3$-volume measure
\[ \nu \hook \mathrm{dVol}^4 |_{{\cal H}_s}= -r_*^2 R^2 \frac{\Delta}{r^2+a^2} \frac{1}{|^*t|} (1+a^2R^2\cos^2\theta) \d ^*t \d^2\omega |_{{\cal H}_s} \simeq -\frac{1}{|{^*t}|} \d ^*t \d^2\omega |_{{\cal H}_s}\]
on each slice ${\cal H}_s$.

\subsection{The Morawetz vector field}
In order to obtain a timelike energy current (which will ensure that the energy flux across spacelike slices is a positive definite quantity), we need to contract the stress-energy tensor for the wave equation with a vector field that is timelike everywhere in $\Omega_{{^*t}_0}^+$. Following the same approach as in \cite{MaNi2009}, we adapt the Morawetz vector field to the Kerr background. Recall that the Morawetz vector field (see \cite{Mo1961}) is a conformal Killing vector for Minkowski spacetime that is Killing for the Minkowski metric rescaled using the conformal factor $\Omega^2 = r^{-2}$. In the coordinates $(r=t-r,R=1/r,\theta,\varphi)$, it has the form
\begin{equation} \label{MorawetzMink}
u^2 \partial_u -2(1+uR) \partial_R \, .
\end{equation}
Its more usual expression in terms of $(t,r,\theta,\varphi)$ coordinates is
\[ (r^2+t^2)\partial_t +2tr \partial_r \, .\]
We brutally translate the expression \eqref{MorawetzMink} in terms of star-Kerr coordinates~:
\begin{equation} \label{Morawetz}
T^{a}:= {^*t}^2 \partial_{^*t} - 2(1+{^*t}R)\partial_R \, .
\end{equation}
This vector field is not a conformal symmetry of Kerr spacetime and in fact, contrary to the Schwarzschild case, its Killing form does not even tend to zero at infinity as we can see in the next lemma, however, the part that does not tend to zero will turn out to be harmless for our estimates.
\begin{lemma} \label{Killing-form}
The vector field $T^a$ is timelike and future-oriented in $\Omega_{{^*t}_0}^+$. Moreover its Killing form is given by
\begin{align*}
&\nabla_{(a}T_{b)} \d x^a \d x^b \\
&= 4\left\{ (1+{^*t}R)M \frac{\partial}{\partial R}\left(\frac{R}{\rho^2}\right)-\frac{2M{^*t}R}{\rho^2}\right \}\d {^*t}^2-4a\sin^2 \theta {^*t} \d R \d {^*\varphi}\\
&- 4a\sin^2\theta \left\{ 2(1+{^*t}R)M \frac{\partial}{\partial R}\left(\frac{R}{\rho^2}\right) - 2 M{^*t} \frac{R}{\rho^2}+ R \right\} \d ^*t \d {^*\varphi} \\
&+ 4a^2\cos^2\theta (1+{^*t}R)R \d \theta^2 + 4a^2\sin^2\theta (1+{^*t}R)\left\{ R + \sin^2\theta M \frac{\partial}{\partial R}\left( \frac{R}{\rho^2}\right) \right\}\d {^*\varphi}^2 \, .
\end{align*}
\end{lemma}

\begin{proof}
First, we have 
\begin{align*}
\hat g_{ab}T^aT^b &= R^2\left(1-\frac{2Mr}{\rho^2}\right){^*t}^4 + 4{^*t}^2(1+{^*t}R)\\
&= {^*t}^2 \left\{ 4(1+{^*t}R) + {^*t}^2 R^2 \left( 1- \frac{2Mr}{\rho^2} \right)  \right\} \, ,
\end{align*}
which vanishes for the two values of ${^*t}R$~:
$$({^*t}R)_{\pm} = -2\frac{1 \mp \sqrt {\frac{2Mr}{\rho^2}}}{1- \frac{2Mr}{\rho^2}} \, .$$
In $\Omega_{{^*t}_0}^+$, these two values are arbitrarily close to $-2$ and ${^*t}R \in \left[-1-\varepsilon, 0\right]$. Consequently, provided we have chosen $\varepsilon $ small enough, $T^a$ is uniformly timelike in $\Omega_{{^*t}_0}^+$. We also have that $T^a$ is future-oriented since
\begin{eqnarray*}
T^a(t) &=& {^*t}^2\partial_{^*t} t - 2(1+{^*t}R)\partial_R t \, , \\
&=& {^*t}^2\partial_{^*t} ({^*}t+r_*) - 2(1+{^*t}R)\partial_R ({^*}t+r_*) \, , \\
&=& {^*t}^2 + 2(1+{^*t}R)\frac{a^2+r^2}{\Delta R^2} >0  \mbox{ in } \Omega_{^*t_0}^+ \, .
\end{eqnarray*}

Let us now calculate the Killing form of $T^a$. We have
$$\nabla_{(a}T_{b)} \d x^a \d x^b= {\cal L}_T \left(\hat g_{ab}\d x^a \d x^b \right) = T\hat g_{ab}\d x^a \d x^b + \hat g_{ab} \left( {\cal L}_T\d x^a \otimes \d x^b + {\cal L}_T\d x^b \otimes \d x^a \right)$$
and
$${\cal L}_T \d {^*t} = 2 {^*t} \d {^*t}, \; {\cal L}_T \d R = - 2R \d {^*t} - 2 {^*t} \d R, \; {\cal L}_T \d \theta = {\cal L}_T \d {^*\varphi} = 0 \, .$$
So that
\begin{align*}
&\nabla_{(a}T_{b)} \d x^a \d x^b\\
&= 4\left\{ (1+{^*t}R)M \frac{\partial}{\partial R}\left(\frac{R}{\rho^2}\right) - \frac{2M{^*t}R}{\rho^2}\right\} \d {^*t}^2-4a\sin^2 \theta {^*t} \d R \d {^*\varphi}\\
&- 4a\sin^2\theta \left\{ 2(1+{^*t}R)M \frac{\partial}{\partial R}\left(\frac{R}{\rho^2}\right) - 2 M{^*t} \frac{R}{\rho^2}+ R \right\} \d ^*t \d {^*\varphi} \\
&+ 4a^2\cos^2\theta (1+{^*t}R)R \d \theta^2 + 4a^2\sin^2\theta (1+{^*t}R)\left\{ R + \sin^2\theta M \frac{\partial}{\partial R}\left(\frac{R}{\rho^2}\right) \right\} \d {^*\varphi}^2 \, ,
\end{align*}
which completes the proof.
\end{proof}

\subsection{The wave equation}

We study the wave equation on Kerr spacetime
\begin{equation} \label{WaveEq}
\square_g \psi + \lambda \psi^3 = 0 \, ,
\end{equation}
in the cases $\lambda =0$, where the equation is linear, and $\lambda=1$. The operator $\square_g$ is the d'Alembertian for the metric $g$ whose expression in local coordinates is given by
\begin{equation} \label{DAlembertian}
\square_g = \frac{1}{\sqrt{|\det g|}} \frac{\partial}{\partial x^{\ba}} \left(\sqrt{|\det g|}g^{\ba\bb}\frac{\partial}{\partial x^{\bb}}\right) \, .
\end{equation}
Equation \eqref{WaveEq} is conformally invariant in the following sense~: a smooth function $\psi$ on Kerr spacetime satisfies \eqref{WaveEq} if any only if the rescaled function
\begin{equation} \label{RescField}
\hat\psi=\Omega^{-1}\psi=r \psi
\end{equation}
satisfies the wave equation
\begin{equation} \label{RescWaveEq}
\left(\square_{\hat g}+\frac{1}{6}\mathrm{Scal}_{\hat g}\right)\hat\psi + \lambda \hat\psi^3 = 0 \, .
\end{equation}
We calculate the expression of \eqref{RescWaveEq} in star-Kerr coordinates. The scalar curvature for $\hat{g}$ is given by
\[  {\mathrm {Scal}}_{\hat g} = \Omega^{-3} \square_g \Omega =\frac{r^3}{\sqrt {\vert g \vert}}\partial_r \sqrt {\vert g \vert}g^{rr}\partial_r\frac{1}{r} = 12\frac{Mr-a^2}{\rho^2} \]
and the d'Alembertian for $\hat g$ reads
\begin{align*}
\square_{\hat g}& = \frac{1}{\sqrt{\vert \hat g \vert}}\frac{\partial}{\partial x^\ba}\sqrt{\vert \hat g \vert} \hat g^{\ba\bb}\frac{\partial}{\partial x^\bb}  \\
&= -\frac{r^2a^2\sin^2\theta}{\rho^2}\partial_{^*t}^2 - \frac{r^2+a^2}{\rho^2}\partial_{^*t}\partial_R - \frac{2a r^2}{\rho^2}\partial_{^*t}\partial_{^*\varphi}-\frac{a}{\rho^2}\partial_R\partial_{^*\varphi} - \frac{r^2}{\rho^2\sin^2\theta} \partial_{^*\varphi}^2\\
&-\frac{r^2}{\rho^2\sin\theta}\left(\partial_R R^2(r^2+a^2)\sin\theta \partial_{^*t} + \partial_R R^2 \sin\theta R^2 \Delta \partial_R + \partial_R aR^2 \sin\theta \partial_{^*\varphi} + \partial_\theta \sin\theta \partial_\theta \right)\, ,
\end{align*}
All calculations done and after some simplifications, we obtain
\begin{align*}
\square_{\hat g}
&=-\frac{r^2a^2\sin^2\theta}{\rho^2}\partial_{^*t}^2 - \frac{r^2+a^2}{\rho^2}\partial_{^*t}\partial_R - \frac{2a r^2}{\rho^2}\partial_{^*t}\partial_{^*\varphi}-\frac{a}{\rho^2}\partial_R\partial_{^*\varphi}\\
&-\frac{r^2}{\rho^2} \left( (1+a^2R^2) \partial_{^*t}\partial_R + (R^2-2MR^3+a^2R^4) \partial_R^2 + aR^2 \partial_R\partial_{^*\varphi} \right)\\
&-\frac{r^2}{\rho^2}\left( 2a^2R \partial_{^*t}+(2R-6MR^2+4aR^3) \partial_R + 2aR \partial_{^*\varphi}\right)-\frac{r^2}{\rho^2} \Delta_{S^2} \, ,
\end{align*}
where $\Delta_{S^2}$ is the Laplacian on the euclidean $2$-sphere
$$ \Delta_{S^2} = \frac{1} {\sin\theta} \partial_\theta \sin\theta \partial_\theta +\frac{1}{\sin^2\theta}\partial_{^*\varphi}^2 = \partial_\theta^2 + \cot\theta \partial_\theta + \frac{1}{\sin^2\theta}\partial_{^*\varphi}^2 \; .$$
So the wave equation \eqref{RescWaveEq} has the expression
\begin{gather}
-\frac{r^2a^2\sin^2\theta}{\rho^2}\partial_{^*t}^2\hat\psi - \frac{r^2+a^2}{\rho^2}\partial_{^*t}\partial_R\hat\psi - \frac{2a r^2}{\rho^2}\partial_{^*t}\partial_{^*\varphi}\hat\psi-\frac{a}{\rho^2}\partial_R\partial_{^*\varphi}\hat\psi \nonumber\\
-\frac{r^2}{\rho^2} \left( (1+a^2R^2) \partial_{^*t}\partial_R\hat\psi + (R^2-2MR^3+a^2R^4) \partial_R^2\hat\psi + aR^2 \partial_R\partial_{^*\varphi}\hat\psi \right) \nonumber\\
-\frac{r^2}{\rho^2}\left( 2a^2R \partial_{^*t}\hat\psi+(2R-6MR^2+4aR^3) \partial_R\hat\psi + 2aR \partial_{^*\varphi}\hat\psi\right) \nonumber\\
- \frac{r^2}{\rho^2}\Delta_{S^2}\hat\psi + 2\frac{Mr-a^2}{\rho^2}\hat\psi + \lambda \psi^3 = 0 \, . \label{resc-wave}
\end{gather}
The Cauchy problem for \eqref{WaveEq} is not known to be well-posed in general for all values of $M$ and $a$. For $\lambda=0$, provided $M\geq \vert a \vert$, it will be well-posed in block I by Leray's general theory of hyperbolic equations \cite{Le1953}, since block I is globally hyperbolic~: smooth compactly supported initial data will give rise to unique smooth solutions. In the nonlinear case ($\lambda =1$), the well-posedness of the Cauchy problem in block I in the subextremal case has been established in \cite{Ni2002}, for minimum regularity (i.e. solutions that are in $H^1$ on the spacelike slices). The analogous result in the extreme case in not known but is expected to hold. For $M< \vert a \vert$, the singularity is naked and with the time machine surrounding it, the spacetime is no longer globally hyperbolic~; the Cauchy problem does not make sense anymore.

However, we are restricting our study to a neighbourhood of spacelike infinity which is globally hyperbolic. It is easy to solve the Cauchy problem directly for the rescaled equation \eqref{RescWaveEq} on $\Omega_{{^*t}_0}^+$ and using the conformal invariance, this implies the well-posedness of the Cauchy problem for \eqref{WaveEq} in $\Omega_{{^*t}_0}^+$. For $\lambda =0$ we simply use the general theory by Leray, smooth compactly supported data on ${\cal H}_1$ give rise to unique smooth solutions on $\Omega_{{^*t}_0}^+$. For $\lambda =1$, the well-posedness of the Cauchy problem in $\Omega_{{^*t}_0}^+$ for smooth compactly supported data on ${\cal H}_1$ is a consequence of the results by F. Cagnac and Y. Choquet-Bruhat \cite{CaCho} using an extension of $\Omega_{{^*t}_0}^+$, minus a neighbourhood of $i^0$ that is outside the domain of influence of the support of the data, into a regularly sliced cylinder (a similar procedure has been used by J. Joudioux in \cite{Jo2012} and by L.J. Mason and JPN in \cite{MaNi2004}). Note that in both cases, the solutions are smooth right up to the boundary $\scri^+_{^*t_0}$. The solutions in $\Omega_{{^*t}_0}^+$ are the restrictions to $\Omega_{{^*t}_0}^+$ of a large class of solutions in the exterior of the black home in the slow and extremal cases. In the fast case however, the question of extending the solutions in $\Omega_{{^*t}_0}^+$ to a solution on the whole spacetime is much more delicate~; we shall not address this question in this paper.

Throughout the paper, we shall work with smooth solutions of \eqref{RescWaveEq} in $\Omega_{{^*t}_0}^+$ arising from smooth compactly supported data on ${\cal H}_1$, for which we shall prove energy estimates. The estimates in the linear case will automatically extend the validity of the Cauchy problem to the functions spaces obtained by completion of smooth compactly supported functions in the energy norms. In the non linear case, a similar extension requires more care but essentially also follows from the estimates~; this shall be explained in more details in Section \ref{NLCase}.

\section{The linear case} \label{Lin}

We study the case $\lambda =0$ first, \eqref{resc-wave} is then the conformal linear wave equation on the rescaled Kerr metric. We start this section by calculating some commutators between the conformal d'Alembertian $\square_{\hat{g}} + \frac16 \mathrm{Scal}_{\hat{g}}$ with some vector fields, from which we infer some approximate conservation laws. We then calculate and simplify the energy fluxes across the hypersurfaces ${\cal H}_s$ and ${\cal S}_{^*t_0}$. Putting all this together gives us the basic and higher order estimates which then entail our result concerning the peeling of linear scalar fields.

\subsection{Commutation with vector fields} \label{CommutLin}

We use the set of five vector fields $\mathcal{A} = \{ X_i \, ,~ i=0,1,2,3,4 \}$ where the $X_i$'s are defined as follows~:
\begin{gather*}
X_0 = \partial_{^*t} \, , ~ X_1 = \partial_{^*\varphi} \, ,~ X_2 = \sin{^*\varphi}\, \partial_\theta + \cot{\theta}\cos{^*\varphi}\, \partial_{^*\varphi} \, , \\
X_3 = \cos{^*\varphi}\, \partial_\theta - \cot{\theta}\sin{^*\varphi}\, \partial_{^*\varphi} \, ,~ X_4 = \partial_R \, .
\end{gather*}
The first vector field, $X_0$, is timelike and future oriented on $\Omega^+_{^*t_0}$ except at $\scri^+$ where it is nul and future oriented~; it is the null normal to $\scri^+$ and is therefore also tangent to $\scri^+$. The next three, $X_1$, $X_2$ and $X_3$, are the generators of rotations, they are tangent to the $2-$spheres $S^2_{\theta,^*\varphi}$ and generate their tangent planes. The last vector field $X_4$ is null and future oriented in $\Omega^+_{^*t_0}$ and is transverse to $\scri^+$. The vector fields $X_0$ and $X_1$ are Killing but unlike in the Schwarzschild situation, $X_2$ and $X_3$ are not Killing on the Kerr background. Also $X_4$ is not tangent to the hypersurface ${\cal S}^+_{^*t_0}$.
The three generators of rotations commute with $\Delta_{S^2}$ and they satisfy $[ X_1,X_2]=X_3$, $[X_2,X_3]=X_1$, $[X_3,X_1]=X_2$, or equivalently
\[ \left[ {\cal L}_{X_1}, {\cal L}_{X_2}\right] = {\cal L}_{X_3} \, ,~ \left[ {\cal L}_{X_2}, {\cal L}_{X_3}\right] = {\cal L}_{X_1} \, ,~ \left[ {\cal L}_{X_3}, {\cal L}_{X_1}\right] = {\cal L}_{X_2} \, ,\]
and as differential operators acting on scalar fields,
\[ X_1^2 + X_2^2 +X_3^3 = -\Delta_{S^2} \, .\]
Let $\hat{\psi}$ be a smooth solution of \eqref{resc-wave} with $\lambda =0$ in $\Omega^+_{^*t_0}$. First, since $X_0$ and $X_1$ are Killing vector fields, $\mathcal{L}_{X_0}$ and $\mathcal{L}_{X_1}$ commute with $\square_{\hat g}+\frac{1}{6}\mathrm{Scal}_{\hat g}$ and we get
\begin{equation}\label{resc-wave0} 
\left( {\square}_{\hat g} + \frac{1}{6}\mathrm{Scal}_{\hat g}\right) \mathcal{L}_{X_i}\hat\psi = 0 ~ ( i = 0,1) \, .
\end{equation}
Commuting ${\cal L}_{X_2} $ into \eqref{resc-wave} after multiplying it by $\rho^2/r^2$ gives
\begin{eqnarray*}
0={\cal L}_{X_2}\left(\frac{\rho^2}{r^2}\left(\square_{\hat g}+\frac{1}{6}\mathrm{Scal}_{\hat g}\right)\hat\psi \right) & = &\frac{\rho^2}{r^2}\left(\square_{\hat g}+\frac{1}{6}\mathrm{Scal}_{\hat g}\right){\cal L}_{X_2}\hat\psi - a^2\sin 2\theta \sin{^*\varphi}\partial_{{^*t}}^2\hat\psi
  \\
&& + 2a \partial_{^*t} {\cal L}_{X_3}\hat\psi + 2aR^2\partial_R {\cal L}_{X_3} \hat\psi + 2aR {\cal L}_{X_3}\hat\psi \, ,
\end{eqnarray*}
which entails the equation
\begin{equation} \label{resc-wave2}
\left(\square_{\hat g}+\frac{1}{6}\mathrm{Scal}_{\hat g}\right){\cal L}_{X_2}\hat\psi = \frac{r^2}{\rho^2} a^2\sin 2\theta\sin{^*\varphi}\partial_{{^*t}}^2 \hat\psi - \frac{r^2}{\rho^2} \left( 2a \partial_{^*t}  +2 aR^2\partial_R  + 2aR \right) {\cal L}_{X_3}\hat\psi \, .
\end{equation}
Similarly for ${\cal L}_{X_3}$
\begin{equation} \label{resc-wave3}
\left(\square_{\hat g}+\frac{1}{6}\mathrm{Scal}_{\hat g}\right){\cal L}_{X_3}\hat\psi = \frac{r^2}{\rho^2} a^2\sin 2\theta \cos {^*\varphi}\partial_{{^*t}}^2 \hat\psi  +  \frac{r^2}{\rho^2} \left( 2a \partial_{^*t} + 2aR^2\partial_R + 2aR  \right) {\cal L}_{X_2}\hat\psi \, .
\end{equation}
Finally we commute the Lie derivative along $X_4$ into \eqref{resc-wave} multiplied by $\rho^2 / r^2$. Denoting $\mathcal{L}_{X_4}(\hat\psi) = \hat\psi_R$, we have
\begin{eqnarray*}
0= \partial_R\left(\frac{\rho^2}{r^2}\left(\square_{\hat g}+\frac{1}{6}\mathrm{Scal}_{\hat g}\right)\hat\psi \right) &= &\frac{\rho^2}{r^2}\left(\square_{\hat g}+\frac{1}{6}\mathrm{Scal}_{\hat g}\right)\partial_R\hat\psi - 4a^2R \partial_{^*t}\partial_R\hat\psi - 4 aR \partial_R\partial_{^*\varphi}\hat\psi\\
&&- 2a^2\partial_{^*t}\hat\psi - 2a \partial_{^*\varphi}\hat\psi -(4a^2R^3-6MR^2+2R)\partial_R^2\hat\psi \\
&&-(12aR^2-12MR+2)\partial_R\hat\psi + 2(M-2a^2R)\hat\psi \, ,
\end{eqnarray*}
which gives the equation
\begin{eqnarray}
\left(\square_{\hat g}+\frac{1}{6}\mathrm{Scal}_{\hat g}\right)\partial_R\hat\psi &=& \frac{r^2}{\rho^2}\left( 4a^2R \partial_{^*t}\partial_R\hat\psi + 4 aR \partial_{^*\varphi}\partial_R\hat\psi + 2a^2\partial_{^*t}\hat\psi + 2a \partial_{^*\varphi}\hat\psi \right) \nonumber\\
&&+\frac{r^2}{\rho^2} (4a^2R^3-6MR^2+2R)\partial_R^2\hat\psi \nonumber \\
&& + \frac{r^2}{\rho^2} \left( (12aR^2-12MR+2)\partial_R\hat\psi - 2(M-2a^2R)\hat\psi \right) \, . \label{resc-wave1}
\end{eqnarray}

\subsection{Approximate conservation laws}

We use the stress-energy tensor for the wave equation on $\hat g$
\[ T_{ab}(\hat\psi) = T_{(ab)}(\hat\psi) = \partial_a \hat\psi \partial_b \hat\psi - \frac{1}{2} \langle \nabla \hat\psi \, ,\nabla \hat\psi \rangle_{\hat{g}}  \hat g_{ab} \]
and the Morawetz vector field $T^a$ to define the energy current
\[ J_b (\hat{\psi}) := T^a T_{ab} (\hat{\psi}) \, .\]
For any scalar field $\hat{\psi}$, we have
\begin{equation}\label{DivSET}
\nabla^a T_{ab}(\hat\psi) = (\square_{\hat g} \hat\psi ) \partial_b \hat\psi \, ,
\end{equation}
which, for solutions of \eqref{RescWaveEq} with $\lambda=0$, becomes
\begin{equation}\label{AppDivFree}
\nabla^a T_{ab}(\hat\psi) = -2\frac{Mr-a^2}{r^2} \hat\psi \partial_b \hat\psi \, .
\end{equation}
This gives the approximate conservation law for solutions of \eqref{RescWaveEq}~:
\begin{align}\label{ACL0}
\nabla^a J_a (\hat{\psi}) &= (\square_{\hat g} \hat\psi ) \nabla_T \hat\psi + ( \nabla^a T^b ) T_{ab}(\hat\psi) \nonumber\\
&= -2\frac{Mr-a^2}{r^2} \hat\psi \nabla_T \hat\psi + \left(\nabla_{(a} T_{b)}  \right) T^{ab} (\hat\psi) \, ,
\end{align}
where $\nabla_{(a}T_{b)}$ is the Killing form for $T^a$ and is given in Lemma \ref{Killing-form}.

Using the equations satisfied by the Lie derivatives along the vectors $X_i$ of a solution $\hat{\psi}$ of \eqref{RescWaveEq} with $\lambda=0$, we obtain approximate conservation laws for these derivatives. For $i=0,1$, we have equation \eqref{resc-wave0}, which says exactly that $\mathcal{L}_{X_i} \hat{\psi}$ is a solution of \eqref{RescWaveEq}. We therefore have the same approximate conservation laws for $\mathcal{L}_{X_i} \hat{\psi}$ as for $\hat{\psi}$~:
\begin{equation} \label{ACL1}
\nabla^a  J_{a} (\mathcal{L}_{X_i}\hat\psi) = -2\frac{Mr-a^2}{r^2} (\mathcal{L}_{X_i}\hat\psi) \nabla_T (\mathcal{L}_{X_i}\hat\psi) + \left(\nabla_{(a} T_{b)}  \right) T^{ab} (\mathcal{L}_{X_i}\hat\psi) \, , ~ i=0,1 \, .
\end{equation}
For $i=2$, putting equality \eqref{DivSET} for $\mathcal{L}_{X_2}\hat\psi$ and equation \eqref{resc-wave2} together gives the approximate conservation law for ${\cal L}_{X_2} \hat{\psi}$~:
\begin{eqnarray}
\nabla^a J_{a} ({\cal L}_{X_2}\hat\psi) &=& (\square_{\hat g}({\cal L}_{X_2} \hat\psi)) \nabla_T ({\cal L}_{X_2}\hat\psi) + \nabla_{(a} T_{b)} T^{ab} ({\cal L}_{X_2} \hat\psi)\nonumber\\
&=&\frac{r^2}{\rho^2} \left( a^2\sin 2\theta \sin{^*\varphi} \partial_{{^*t}}^2 \hat\psi - 2a \partial_{^*t} {\cal L}_{X_3}\hat\psi  - 2aR^2\partial_R {\cal L}_{X_3} \hat\psi \right.\nonumber\\
&& \left.  - 2aR {\cal L}_{X_3}\hat\psi - 2 (MR-a^2R^2 ){\cal L}_{X_2}\hat\psi \right) \nabla_T ({\cal L}_{X_2} \hat\psi) \nonumber \\
&&+ (\nabla_{(a} T_{b)}) T^{ab} ({\cal L}_{X_2} \hat\psi) \, . \label{ACL2}
\end{eqnarray}
Similarly for ${\cal L}_{X_3}\hat\psi$~:
\begin{eqnarray}
\nabla^a J_{a} ({\cal L}_{X_3}\hat\psi) &=& \square_{\hat g}({\cal L}_{X_3} \hat\psi) \nabla_T ({\cal L}_{X_3}\hat\psi) + \nabla
_{(a} T_{b)} T^{ab} ({\cal L}_{X_3} \hat\psi)\nonumber\\
&=&\frac{r^2}{\rho^2} \left( a^2\sin 2\theta \cos{^*\varphi} \partial_{{^*t}}^2 \hat\psi + 2a \partial_{^*t} {\cal L}_{X_2}\hat\psi + 2aR^2\partial_R {\cal L}_{X_2} \hat\psi \right.\nonumber\\
&&\left.  + 2aR {\cal L}_{X_2}\hat\psi - 2 (MR-a^2R^2 ){\cal L}_{X_3}\hat\psi\right) \nabla_T({\cal L}_{X_3} \hat\psi) \nonumber \\
&& + (\nabla_{(a} T_{b)}) T^{ab} ({\cal L}_{X_3} \hat\psi) \, . \label{ACL3}
\end{eqnarray}
Finally, associating equality \eqref{AppDivFree} for $\mathcal{L}_{X_4}\hat\psi = \partial_R\hat\psi$ and equation \eqref{resc-wave1}, we obtain the approximate conservation law for  $\mathcal{L}_{X_4} \hat\psi= \partial_R \hat\psi$~:
\begin{eqnarray}
\nabla^a J_{a} (\partial_R \hat\psi) & = & \frac{r^2}{\rho^2}\left( 4a^2R \partial_{^*t}\partial_R\hat\psi + 4 aR \partial_{^*\varphi}\partial_R\hat\psi + 2a^2\partial_{^*t}\hat\psi + 2a \partial_{^*\varphi}\hat\psi \right) \nabla_T (\partial_R\hat\psi) \nonumber\\
&& +\frac{r^2}{\rho^2}\left((4a^2R^3-6MR^2+2R)\partial_R^2\hat\psi+(12aR^2-12MR+2)\partial_R\hat\psi \right) \nabla_T (\partial_R\hat\psi) \nonumber\\
&&- \left(2(M-2a^2R)\hat\psi + 2\frac{Mr-a^2}{r^2} \partial_R \hat\psi \right)  \nabla_T (\partial_R\hat\psi) +\nabla_{(a} T_{b)} T^{ab} (\partial_R \hat\psi)  \, . \label{ACL4}
\end{eqnarray}
The Killing form $\nabla_{(a} T_{b)} T^{ab}$ has a rather long and cumbersome expression, but the details of its coefficients is unimportant for our estimates. The following lemma gives a more synthetic expression of this Killing form that focuses on its useful properties (the proof can be found in Appendix \ref{AppMorawetz}).
\begin{lemma}\label{short-morawetz}
For a given scalar field $\hat\psi$, denoting by $\hat\psi_{^*t}$, $\hat\psi_{^*\varphi}$, $\hat\psi_R$ and $\hat\psi_\theta$ the derivatives of $\hat\psi$ with respect to the variables $^*t$, $^*\varphi$, $R$ and $\theta$, we can express $\nabla_{(a} T_{b)}  T^{ab} (\hat\psi)$ under the following form
\begin{align*}
\left(\nabla_a T_b  \right) T^{ab} (\hat\psi) &= A_1 \hat\psi_{^*t}^2 + A_2 \hat\psi_{^*t} \hat\psi_R + A_3 \hat\psi_{^*t} \hat\psi_{^*\varphi} +A_4 R \hat\psi_R^2 + A_5 \hat\psi_R \hat\psi_{^*\varphi}\\
&+ A_6 \sin^2\theta \hat\psi_\theta^2 + A_7 \hat\psi_{^*\varphi}^2 + A_8 \vert \nabla_{S^2} \hat\psi \vert^2 \, ,
\end{align*}
where the functions $A_i \; (i=1,2...8)$ are bounded on $\Omega^+_{*t_0}$.
\end{lemma}
\begin{proof}
First, recall that
\begin{eqnarray*}
\nabla_{(a}T_{b)} \d x^a \d x^b 
&=& 4\left( (1+{^*t}R)M \frac{\partial}{\partial R}\left(\frac{R}{\rho^2}\right)-\frac{2M{^*t}R}{\rho^2} \right) \d {^*t}^2  -4a\sin^2 \theta {^*t} \d R \d {^*\varphi}\\
&&- 4a\sin^2\theta \left( 2(1+{^*t}R)M \frac{\partial}{\partial R}\left(\frac{R}{\rho^2}\right) - 2 M{^*t} \frac{R}{\rho^2}+ R \right) \d ^*t \d {^*\varphi} \\
&&+ 4a^2\cos^2\theta (1+{^*t}R)R \d \theta^2 \\
&&+ 4a^2\sin^2\theta (1+{^*t}R)\left( R + \sin^2\theta M \frac{\partial}{\partial R}\left(\frac{R}{\rho^2}\right) \right) \d {^*\varphi}^2 \,.
\end{eqnarray*} 
Hence, when calculating $T^{ab} \nabla_{(a}T_{b)}$, we are only interested in the components of $T^{ab}$ along $\partial_{^*t}^2$, $\partial_{^*t}\partial_{^*\varphi}$, $\partial_R\partial_{^*\varphi}$, $\partial_{\theta}^2$ and $\partial_{^*\varphi}^2$. These can be written as
\[ A^{ab} + B^{ab} \langle \hat\psi , \hat\psi \rangle_{\hat{g}} \, ,\]
where, denoting $\hat\psi_a = \frac{\partial}{\partial x^a} \hat\psi$,
\begin{eqnarray*}
A^{ab} \partial_{x^a} \partial_{x^b} &=& (\hat g^{0c}\hat g^{0d}\hat\psi_c \hat\psi_d)\partial_{^*t}^2 + 2(\hat g^{0c}\hat g^{3d}\hat\psi_c \hat\psi_d)\partial_{^*t}\partial_{^*\varphi} + 2(\hat g^{1c}\hat g^{3d}\hat\psi_c \hat\psi_d)\partial_R\partial_{^*\varphi} \\
&&+ (\hat g^{22}\hat g^{22}\hat\psi_\theta^2)\partial_\theta^2 + (\hat g^{3c}\hat g^{3d}\hat\psi_c \hat\psi_d)\partial_{^*\varphi}^2 \, ,\\
B^{ab} \partial_{x^a} \partial_{x^b} &=& \frac{1}{2}\left(\frac{r^2a^2\sin^2\theta}{\rho^2}\partial_{^*t}^2 + 2\frac{ar^2}{\rho^2}\partial_{^*t}\partial_{^*\varphi} + 2 \frac{a}{\rho^2}\partial_R\partial_{^*\varphi} + \frac{r^2}{\rho^2} \partial_{\theta}^2 + \frac{r^2}{\rho^2\sin^2\theta}\partial_{^*\varphi}^2 \right) \, .
\end{eqnarray*}
%
%
The first term of $\left(\nabla_a T_b  \right) T^{ab}$ is~:
\begin{eqnarray*}
\nabla_{(a}T_{b)}A^{ab} &=& 4\left( (1+{^*t}R)M \frac{\partial}{\partial R}\left(\frac{R}{\rho^2}\right)-\frac{2M{^*t}R}{\rho^2}\right) (\hat g^{0c}\hat g^{0d}\hat\psi_c \hat\psi_d)\\
&&- 8a\sin^2\theta \left( 2(1+{^*t}R)M \frac{\partial}{\partial R}\left(\frac{R}{\rho^2}\right) - \frac{2 M{^*t} R}{\rho^2}+ R \right) (\hat g^{0c}\hat g^{3d}\hat\psi_c \hat\psi_d) \\
&&-8a\sin^2 \theta {^*t} (\hat g^{1c}\hat g^{3d}\hat\psi_c \hat\psi_d) + 4a^2\cos^2\theta (1+{^*t}R)R (\hat g^{22}\hat g^{22}\hat\psi_\theta^2)\\
&&+ 4\sin^2\theta (1+{^*t}R)\left( a^2R + a^2\sin^2\theta M \frac{\partial}{\partial R}\left(\frac{R}{\rho^2}\right) \right) (\hat g^{3c}\hat g^{3d}\hat\psi_c \hat\psi_d) \, ,
\end{eqnarray*}
in which all the coefficients of $\hat\psi_c \hat\psi_d$ are bounded. The terms involving $\hat\psi_R^2$ are
\begin{gather*}
4\left( (1+{^*t}R)M \frac{\partial}{\partial R}\left(\frac{R}{\rho^2}\right)-\frac{2M{^*t}R}{\rho^2}\right) (\hat g^{01})^2\hat\psi_R^2 \\
= 4\left( (1+{^*t}R)M \frac{\partial}{\partial R}\left(\frac{R}{\rho^2}\right)-\frac{2M{^*t}R}{\rho^2}\right) \left( \frac{r^2+a^2}{\rho^2}\right)^2\hat\psi_R^2 = O (R^2) \hat\psi_R^2 \, ,
\end{gather*}
\begin{gather*}
- 8a\sin^2\theta \left( 2(1+{^*t}R)M \frac{\partial}{\partial R}\left(\frac{R}{\rho^2}\right) - \frac{2 M{^*t} R}{\rho^2}+ R \right) \hat g^{01}\hat g^{31}\hat\psi_R^2 \\
= - 8a\sin^2\theta \left( 2(1+{^*t}R)M \frac{\partial}{\partial R}\left(\frac{R}{\rho^2}\right) - \frac{2 M{^*t} R}{\rho^2}+ R \right) \frac{a(r^2+a^2)}{\rho^4} \hat\psi_R^2 = O (R^3) \hat\psi_R^2 \, ,
\end{gather*}
\[ -8a\sin^2 \theta {^*t} \hat g^{11}\hat g^{31}\hat\psi_R^2 = -8a\sin^2 \theta {^*t}  \frac{a(r^2+a^2)}{\rho^4}\hat\psi_R^2 = O (\, ^*tR^2) \hat\psi_R^2 = O(R) \hat\psi_R^2 \, , \]
\begin{gather*}
4\sin^2\theta (1+{^*t}R)\left( a^2R + a^2\sin^2\theta M \frac{\partial}{\partial R}\left(\frac{R}{\rho^2}\right) \right) (\hat g^{31})^2 \hat\psi_R^2 \\
= 4\sin^2\theta (1+{^*t}R)\left( a^2R + a^2\sin^2\theta M \frac{\partial}{\partial R}\left(\frac{R}{\rho^2}\right) \right) \left( \frac{a}{\rho^2} \right)^2 \hat\psi_R^2 = O(R^5)\hat\psi_R^2 \, .
\end{gather*}
The second term is $\nabla_{(a}T_{b)}B^{ab} \langle \hat\psi , \hat\psi \rangle_{\hat{g}}$. We have
\begin{eqnarray*}
\nabla_{(a}T_{b)}B^{ab} & = &\frac{1}{2} \left(-\frac{r^2a^2\sin^2\theta}{\rho^2}4(1+ {^*t}R)\left(R - M\frac{\partial}{\partial R}\left(\frac{R}{\rho^2}\right)\right)  + 4{^*t}R^2\left(1- \frac{2Mr}{\rho^2}\right) + 4R \right) \\
&&+\frac{ar^2}{\rho^2} \left( -8(1+{^*t}R)Ma\sin^2\theta \frac{\partial}{\partial R}\left(\frac{R}{\rho^2}\right) + 8{^*t}Ma\sin^2\theta\frac{R}{\rho^2}- 4aR\sin^2\theta \right) \\
&&-\frac{4a^2{^*t}\sin^2\theta}{\rho^2} + \frac{2r^2(1+{^*t}R) }{\rho^2}a^2R\cos^2\theta \\
&&+ 2a^2\frac{r^2}{\rho^2}(1+{^*t}R) \left( R + M\sin^2\theta \left( R - M \frac{\partial}{\partial R}\left(\frac{R}{\rho^2}\right) \right) \right) = O(R)
\end{eqnarray*}
and
\begin{eqnarray*}
\langle \hat\psi , \hat\psi \rangle_{\hat{g}} &=&-\frac{1}{\rho^2}\left(r^2a^2\sin^2\theta\hat\psi^2_{^*t}+2(r^2+a^2)\hat\psi_{^*t}\hat\psi_R+2ar^2\hat\psi_{^*t}\hat\psi_{^*\varphi}+2a\hat\psi_R\hat\psi_{^*\varphi} \right)\\
&& -\frac{1}{\rho^2} \left(R^2\Delta\hat\psi^2_R+r^2\hat\psi_{\theta}^2+\frac{r^2}{\sin^2\theta}\hat\psi^2_{^*\varphi}\right)\, .
 \end{eqnarray*}
All the coefficients in $\nabla_{(a}T_{b)}B^{ab} \langle \hat\psi , \hat\psi \rangle_{\hat{g}}$ are bounded and the one in front $\hat{\psi}_R^2$ is or the order of $R^3$. The lemma is therefore proved.
\end{proof}

\subsection{Energy fluxes}

Given a $1$-form $\alpha$ on $\bar{\cal M}$, its Hodge dual is given as (see \eqref{Star})
\[ * \alpha =  \mathrm{dVol}\, \lrcorner \, \alpha \]
where $\dvol$ is the $4$-volume form associated with the metric $\hat{g}$~:
\begin{eqnarray*}
\mathrm{dVol} &=& \sqrt{\vert \det (\hat{g})\vert } \, \d \,^*t \wedge \d R \wedge \d \theta \wedge \d \,^*\varphi \\
&=& R^2\rho^2 \sin \theta \, \d \,^*t \wedge \d R \wedge \d \theta \wedge \d \,^*\varphi \\
&=& R^2\rho^2 \d \,^*t \wedge \d R \wedge \d^2 \omega
\end{eqnarray*}
with $\d^2 \omega = \sin \theta \, \d \theta \wedge \d \,^*\varphi$. If $\beta$ is another $1$-form, we have the following equality
\[ \alpha \wedge *\beta = -\frac14 \alpha_a \beta^a \dvol \, .\]
Let $S$ be an oriented hypersurface in $\bar{\cal M}$, $l^a$ a vector field transverse to $S$ with the same orientation as $S$ and $n^a$ a normal vector field to $S$ such that $\hat{g} (l,n)=1$, then, denoting by $n^\flat$ the $1$-form associated with $n$, we have
\begin{eqnarray*}
\int_S * \alpha &=& \int_S \langle l , n \rangle_g *\alpha \\
&=& \int_S l \lrcorner ( n^\flat \wedge *\alpha ) - \underset{=0 \mathrm{~since~} n \perp S}{\underbrace{\int_S n^\flat \wedge (l\lrcorner *\alpha )}} \\
&=& \int_S l \lrcorner (-\frac14) \alpha_a n^a \dvol \, . 
\end{eqnarray*}
Because of the choice of orientation of the vector field $l$, the quantity $l \lrcorner \dvol$ is a positive measure on $S$ (i.e. it has the same orientation as $S$) and
\[ \int_S \alpha_a n^a l \, \lrcorner \, \dvol \]
defines the flux of $\alpha$ across $S$. We apply this to our energy current $1$-form $J$ in order to calculate energy fluxes. We define the energy 3-form $E(\hat\psi)$ associated with $T^a$ as minus four times the Hodge dual of the energy current $J$~:
\[ E(\hat\psi) = -4 *J (\hat\psi )= 4 T^a T_a^b (\hat\psi) \partial_b \hook \mathrm{ dVol} \, .\]
Then, for an oriented hypersurface $S$, the energy flux of the scalar field $\hat\psi$ across $S$ will be
\[ {\cal E}_{S} (\hat\psi) := \int_{S} E(\hat \psi) \, .\]
\begin{remark}
The exterior derivative of the Hodge dual of $\alpha$ is related to the divergence of $\alpha$ as follows
\[ \d *\alpha = -\frac14 \nabla_a \alpha^a \mathrm{dVol} \, .\]
Therefore, is $S$ is the boundary of a bounded open subset $\Omega$ of $\bar{\cal M}$ and has outgoing orientation, using Stokes theorem, we see that
\[ -4 \int_S * \alpha = \int_{S} \alpha_a n^a ( l \lrcorner \dvol ) = \int_\Omega \nabla_a \alpha^a \dvol \, .\]
\end{remark}
The following lemma gives us simplified equivalent expressions of the energy fluxes across the leaves of the foliation ${\cal H}_s$ of $\Omega^+_{^*t_0}$~:
\begin{lemma}\label{short-energy}
For $\vert ^*t_0\vert$ large enough, the energy fluxes of $\hat\psi$ across the hypersurfaces ${\cal H}_s$, $0<s\leq 1$ and ${\cal H}_0 = \scri^+_{^*t_0} $ have the following simpler equivalent expressions~:
\begin{eqnarray}
{\cal E}_{{\cal H}_s} (\hat\psi) &\simeq& \int_{{\cal H}_s} \left({^*t}^2 \hat\psi_{^*t}^2 +  \frac{R}{\vert {^*t} \vert}\hat\psi_R^2 + \vert \nabla_{S^2}\hat\psi \vert^2 \right) \d {^*t} \wedge \d^2 \omega \, , \label{equivEnHs} \\
{\cal E}_{\scri^+_{^*t_0}} (\hat\psi) &\simeq & \int_{\scri^+_{^*t_0}} \left({^*t}^2 \hat\psi_{^*t}^2 + \vert \nabla_{S^2}\hat\psi \vert^2 \right) \d {^*t} \wedge \d^2 \omega \, , \label{equivEnH0}
\end{eqnarray}
where
\[ \vert \nabla_{S^2} \hat\psi \vert^2 = \hat\psi_\theta^2 + \frac{1}{\sin^2\theta} \hat\psi_{^*\varphi}^2 \, .\]
We do not need the details of the energy flux ${\cal E}_{{\cal S}_{^*t}} (\hat\psi)$ so we do not give a simplified form for it. All that matters here is that, thanks to the dominant energy condition, it is non negative.
\end{lemma}
\begin{proof}
The energy $3$-form reads
\begin{eqnarray*}
E &=& 4 T^a T_a^b (\hat\psi) \partial_b \hook \mathrm{ dVol} \\
&=& 4 R^2\rho^2 \left( T^a T_a^0 (\hat{\psi}) \d \, ^*t - T^a T_a^1 (\hat{\psi}) \d R \right) \wedge \d^2 \omega \\
&&+ 4 R^2\rho^2 \d \, ^*t \wedge \d R \wedge \left( T^a T_a^2 (\hat{\psi}) \d \theta - T^a T_a^3 (\hat{\psi}) \d \, ^*\varphi \right) \, .
\end{eqnarray*}
Only the first two terms are tangent to ${\cal H}_s$ and therefore
\[{\cal E}_{{\cal H}_s} (\hat\psi) = 4\int_{{\cal H}_s}  - T^a T^1_a R^2 \rho^2\, \d {^*t} \wedge \d ^2 \omega + T^a T_a^0 R^2 \rho^2\, \d R \wedge \d^2 \omega \, . \]
Since on the hypersurface ${ \cal H}_s$, for $0<s\leq 1$, we have
\[ \d R = \frac{\Delta R^2}{s(r^2 + a^2)} \d {^*t} \, ,\]
it follows that
\[{\cal E}_{{\cal H}_s} (\hat\psi) = 4 \int_{{\cal H}_s} \left(  - T^a T^1_a + \frac{\Delta R^2}{s(r^2 + a^2)}T^a T_a^0  \right) R^2 \rho^2\, \d {^*t} \wedge \d ^2 \omega \, .\]
We have
\[ T^aT_a^0 = {^*t}^2 T_0^0 - 2(1+ {^*t}R) T_1^0 \, ,\]
with
\[ T_0^0 = \hat g^{0\ba}T_{0\ba}= -\frac{1}{\rho^2} \left(\frac{r^2a^2\sin^2\theta}{2}\hat\psi_{^*t}^2 -\frac{R^2\Delta}{2}\hat\psi_{R}^2 - \frac{r^2}{2} \hat\psi_{\theta}^2 - \frac{r^2}{2\sin^2\theta}\hat\psi_{^*\varphi}^2 - a\hat\psi_R \hat\psi_{^*\varphi}\right) \]
and 
\[ T_1^0 = \hat g^{0\ba}T_{1\ba} = -\frac{1}{\rho^2} \left((r^2+a^2)\hat\psi_R^2 + r^2a^2\sin^2\theta \hat\psi_{^*t}\hat\psi_R + ar^2 \hat\psi_R \hat\psi_{^*\varphi}  \right) \, . \]
Also,
\[ T^aT_a^1 = {^*t}^2 T_0^1 - 2(1+ {^*t}R) T_1^1 \, , \]
with
\begin{eqnarray*}
T_0^1 &=& -\frac{1}{\rho^2} \left((r^2 + a^2)\hat\psi_{^*t}^2 + R^2\Delta \hat \psi_{^*t} \hat\psi_R + a\hat\psi_{^*t}\hat\psi_{^*\varphi} \right)\\
&&+ MRa^2\sin^2\theta\frac{1}{\rho^4} \left(1 - \frac{1}{\rho^2}\right) \left(r^2a^2\sin^2\theta \hat\psi_{^*t}^2 + R^2\Delta \hat\psi_R^2 + r^2 \hat\psi_\theta^2 + \frac{r^2}{\sin^2\theta} \hat\psi_{^*\varphi}^2 \right)\\
&&+ MRa^2\sin^2\theta\frac{1}{\rho^4} \left(1 - \frac{1}{\rho^2}\right) \left(2(r^2+a^2)\hat\psi_{^*t}\hat\psi_R + 2ar^2 \hat\psi_{^*t}\hat\psi_{^*\varphi} + 2a\hat\psi_R  \hat\psi_{^*\varphi}  \right)
\end{eqnarray*}
and
\[ T_1^1 =  -\frac{1}{\rho^2} \left(\frac{R^2\Delta}{2}\hat\psi_R^2 - \frac{r^2a^2\sin^2\theta}{2}\hat\psi_{^*t}^2 - \frac{r^2}{2}\hat\psi_{\theta}^2 - \frac{r^2}{2\sin^2\theta}\hat\psi_{^*\varphi}^2 - ar^2 \hat\psi_{^*t}\hat\psi_{^*\varphi}  \right) \, . \]
Using these expressions and noting that $R. \rho^{-4} = O (R^5)$, we have~:
\begin{gather}
- T^a T^1_a + \frac{\Delta R^2}{s(r^2 + a^2)}T^a T_a^0 \nonumber \\
= \frac{1}{\rho^2} \left({^*t}^2(r^2+a^2) + (1+ {^*t}R)r^2a^2\sin^2\theta - \frac{{^*t}^2r^2a^2\sin^2\theta}{2}\frac{\Delta R^2}{s(r^2 + a^2)} + O(R^3 \sin^4 \theta ) \right)\hat\psi_{^*t}^2 \nonumber \\
+ \frac{1}{\rho^2} \left( -(1+{^*t}R)R^2\Delta + \frac{\Delta R^2}{s(r^2 + a^2)} \left(\frac{{^*t}^2R^2\Delta}{2} + 2(1+ {^*t}R)(r^2 + a^2) \right) + O(R^5 \sin^2\theta) \right)\hat \psi_R^2 \nonumber \\
+ \frac{1}{\rho^2} \left( (1 + {^*t}R)r^2 + \frac{{^*t}^2r^2}{2}\frac{\Delta R^2}{s(r^2 + a^2)} +O(R^3 \sin^2 \theta )\right)\hat\psi_\theta^2 \nonumber \\
+ \frac{1}{\rho^2} \left(\frac{1}{\sin^2\theta} (1 + {^*t}R)r^2 + \frac{1}{\sin^2\theta}\frac{{^*t}^2r^2}{2}\frac{\Delta R^2}{s(r^2 + a^2)} + O(R^3) \right)\hat\psi_{^*\varphi}^2 \nonumber \\
+ \frac{1}{\rho^2} \left( {^*t}^2R^2 \Delta + 2(1 + {^*t}R)r^2a^2\sin^2\theta\frac{\Delta R^2}{s(r^2 + a^2)}  + O(R^3 \sin^2 \theta )\right)\hat\psi_{^*t} \hat\psi_R \nonumber \\
+ \frac{1}{\rho^2} \left( {^*t}^2a + 2(1 + {^*t}R) ar^2 + O(R^3 \sin^2 \theta ) \right)\hat\psi_{^*t}\hat\psi_{^*\varphi} \nonumber \\
+ \frac{1}{\rho^2} \left( \left( {^*t}^2a + 2(1 + {^*t}R)ar^2 \right)\frac{\Delta R^2}{s(r^2 + a^2)} + O(R^5 \sin^2 \theta ) \right) \hat\psi_R \hat\psi_{^*\varphi} \, . \label{EnergyDensityLimExp}
\end{gather}
Recalling from \eqref{epsilonestimates} that in $\Omega^+_{^*t_0}$, we have
\[ -r_*<{^*}t<<-1 \, , \, {^*t}R \in [-1-\varepsilon,0] \, , \, r_*R \in [1, 1+\varepsilon ] \, ,\]
and using the fact that $s = -\,^*t /r_* = \vert ^*t \vert/r_*$, let us obtain equivalents or estimates for the different terms in \eqref{EnergyDensityLimExp}. We shall not distinguish between the various small quantities that appear in the estimates below~; they shall all be denoted $\varepsilon$ since, assuming $\vert ^*t_0 \vert$ large enough they can all be assumed to be the same small $\varepsilon$.
\begin{itemize}
\item The coefficient of $\hat{\psi}_{^*t}^2$ is
\begin{eqnarray*}
A &=& {^*t}^2\left( \frac{r^2+a^2}{\rho^2} + \frac{(1+ {^*t}R)r^2a^2\sin^2\theta}{^*t^2 \rho^2} - \frac{\Delta a^2\sin^2\theta}{2s(r^2 + a^2)\rho^2} + O(R^5 \sin^4 \theta \,^*t^{-2} ) \right) \\
&=& {^*t}^2\left( \frac{r^2+a^2}{\rho^2} + \frac{(1+ {^*t}R)r^2a^2\sin^2\theta}{^*t^2 \rho^2} + \frac{\Delta a^2\sin^2\theta \, r_*}{2\,^*t(r^2 + a^2)\rho^2} + O(R^5 \sin^4 \theta \,^*t^{-2} ) \right)
\end{eqnarray*}
and therefore
\begin{equation} \label{EstimA}
(1-\varepsilon) ^*t^2 \leq A \leq (1+\varepsilon) ^*t^2 \, .
\end{equation}
\item We then turn to the coefficient of $\hat{\psi}_R^2$~:
\begin{eqnarray*}
B&=& -\frac{(1+{^*t}R)R^2\Delta}{\rho^2} + \frac{\Delta R^2}{s(r^2 + a^2)\rho^2} \left(\frac{{^*t}^2R^2\Delta}{2} + 2(1+ {^*t}R)(r^2 + a^2) \right) + O(R^7 \sin^2\theta) \\
&=& \frac{R}{\vert ^*t \vert} \left( -\frac{(1+{^*t}R)R \vert ^*t\vert \Delta}{\rho^2} + \frac{\Delta Rr_*}{(r^2 + a^2)\rho^2} \left(\frac{{^*t}^2R^2\Delta}{2} + 2(1+ {^*t}R)(r^2 + a^2) \right) \right. \\
&& + O(R^8 \sin^2\theta \vert ^*t\vert^{-1}) \bigg) \, .
\end{eqnarray*}
Putting $x = \vert ^*t\vert R$, we have
\begin{eqnarray*}
B&=& \frac{R}{\vert ^*t \vert} \left( \frac{\Delta}{\rho^2} \left( -(1-x)x + \frac{Rr_*}{(r^2 + a^2)} \left(\frac{x^2\Delta}{2} + 2(1-x)(r^2 + a^2) \right) \right) \right. \\
&& + O(R^8 \sin^2\theta \vert ^*t\vert^{-1}) \bigg) \\
&=& \frac{R}{\vert ^*t \vert} \left( -(1-x)x + \frac{x^2}{2} + 2(1-x) + O(R) + O(R^8 \sin^2\theta \vert ^*t\vert^{-1}) \right) \\
&=& \frac{R}{\vert ^*t \vert} \left( \frac{3x^2-6x+4}{2}+ O(R) + O(R^8 \sin^2\theta \vert ^*t\vert^{-1}) \right) \, .
\end{eqnarray*}
Noting that
\[ 3x^2-6x+4 = 3 (x-1)^2 +1 \geq 1 \, ,\]
we conclude that
\begin{equation} \label{EstimB1}
B \geq \frac{(1-\varepsilon)}{2} \frac{R}{\vert ^*t\vert} \, .
\end{equation}
Moreover, since $x \in [0,1+\varepsilon] \subset[0,2]$, we also have
\begin{equation} \label{EstimB2}
B \leq (2+\varepsilon) \frac{R}{\vert ^*t\vert} \, .
\end{equation}
\item The factor in front of
\[ \vert \nabla_{S^2} \hat{\psi}^2 \vert^2 = \hat{\psi}_\theta^2 + \frac{1}{\sin ^2 \theta } \hat{\psi}_{^*\varphi}^2 \]
is
\begin{eqnarray*}
C &=& \frac{1}{\rho^2} \left( (1 + {^*t}R)r^2 + \frac{{^*t}^2r^2}{2}\frac{\Delta R^2}{s(r^2 + a^2)} +O(R^3 \sin^2 \theta )\right) \\
&=& (1 + {^*t}R)\frac{r^2}{\rho^2} + \frac{\vert ^*t \vert r_* r^2}{2}\frac{\Delta R^2}{(r^2 + a^2)\rho^2} +O(R^5 \sin^2 \theta )
\end{eqnarray*}
which entails
\begin{equation} \label{EstimC}
1-\varepsilon \leq C \leq 1+\varepsilon \, .
\end{equation}
\item The term involving $\hat\psi_{^*t} \hat\psi_R$ is
\begin{eqnarray*}
D &=& \frac{1}{\rho^2} \left( {^*t}^2R^2 \Delta + 2(1 + {^*t}R)r^2a^2\sin^2\theta\frac{\Delta R^2}{s(r^2 + a^2)}  + O(R^3 \sin^2 \theta )\right) \hat\psi_{^*t} \hat\psi_R \\
&=& \left( {^*t}^2R^2 \frac{\Delta}{\rho^2} + 2(1 + {^*t}R)r^2a^2\sin^2\theta\frac{\Delta R^2r_*}{\vert^*t\vert (r^2 + a^2)\rho^2}  + O(R^5 \sin^2 \theta )\right) \hat\psi_{^*t}\hat\psi_R \, .
\end{eqnarray*}
Writing
\begin{eqnarray*}
\vert D \vert &=& \vert ^*t \hat\psi_{^*t} \vert \sqrt{\frac{R}{\vert ^* t \vert }} \vert \hat\psi_R \vert \left( \vert ^*t R\vert^{3/2} \frac{\Delta}{\rho^2} + 2(1 + {^*t}R)r^2a^2\sin^2\theta\frac{\Delta R^{3/2}r_*}{\vert^*t\vert^{3/2} (r^2 + a^2)\rho^2}  \right. \\
&& \left. + O(R^{9/2} \vert ^* t\vert^{-1/2} \sin^2 \theta )\right) \, ,
\end{eqnarray*}
we can estimate $\vert D \vert$ as follows
\begin{equation} \label{EstimD}
\vert D \vert \leq (1+ \varepsilon ) \left( \frac{\lambda^2}{2} (^* t)^2 \hat{\psi}_{^*t}^2 + \frac{1}{2\lambda^2}\frac{R}{\vert ^* t\vert} \hat{\psi}_R^2 \right) \, ,
\end{equation}
where $\lambda >0$ can be chosen arbitrarily. In order to control $D$ by the terms involving $\hat{\psi}_{^*t}^2 $ and $\hat{\psi}_R^2$, we need to ensure that
\begin{equation} \label{Parameterlambda}
\frac{\lambda^2}{2}<1 \mbox{ and } \frac{1}{2\lambda^2} < \frac{1}{2}\, ,
\end{equation}
which is equivalent to $1<\lambda<\sqrt 2$.
\item The term involving $\hat{\psi}_{^*t}\hat{\psi}_{^*\varphi}$ is (putting again $x = \vert ^*t\vert R$)
\begin{eqnarray*}
E &=& \frac{1}{\rho^2} \left( {^*t}^2a + 2(1 + {^*t}R) ar^2 + O(R^3 \sin^2 \theta ) \right)\hat\psi_{^*t}\hat\psi_{^*\varphi} \\
&=& \frac{a}{R^2\rho^2} \left( x^2 + 2 - 2x + O(R^5 \sin^2 \theta ) \right)\hat\psi_{^*t}\hat\psi_{^*\varphi} \\
&=& \frac{a}{R^2\rho^2} \left( (x -1)^2+1 + O(R^5 \sin^2 \theta ) \right)\hat\psi_{^*t}\hat\psi_{^*\varphi} \, .
\end{eqnarray*}
Since $x \in [0,1+\varepsilon ] \subset [0,2]$, we have
\begin{eqnarray}
\vert E \vert &\leq & \vert a \vert (2+\varepsilon) \vert \hat\psi_{^*t}\vert \vert \hat\psi_{^*\varphi} \vert \nonumber \\
&\leq & \frac{\vert a \vert \sin \theta (2+\varepsilon)}{\vert ^*t\vert} \vert ^* t \vert \vert \hat\psi_{^*t}\vert \frac{1}{\sin \theta} \vert \hat\psi_{^*\varphi} \vert \nonumber \\
&\leq & \varepsilon \left( ^* t ^2\hat\psi_{^*t}^2 + \frac{1}{\sin^2 \theta} \hat\psi_{^*\varphi}^2 \right) \, . \label{EstimE}
\end{eqnarray}
\item Finally, the term involving $\hat{\psi}_{R}\hat{\psi}_{^*\varphi}$ is
\begin{eqnarray*}
F &=& \frac{1}{\rho^2} \left( \left( {^*t}^2a + 2(1 + {^*t}R)ar^2 \right)\frac{\Delta R^2}{s(r^2 + a^2)} + O(R^5 \sin^2 \theta ) \right) \hat\psi_R \hat\psi_{^*\varphi} \\
&=& \frac{a}{R^2 \rho^2} \left( \left( x^2+2 -2x \right)\frac{\Delta R^4r_*}{\vert ^*t \vert (r^2 + a^2)} + O(R^5 \sin^2 \theta ) \right) \hat\psi_R \hat\psi_{^*\varphi} \, .
\end{eqnarray*}
We can estimate this term as follows
\begin{eqnarray}
\vert F \vert &\leq & \frac{\vert a\vert R^3}{\vert ^* t \vert} (2+\varepsilon)  \vert \hat\psi_R \vert \vert \hat\psi_{^*\varphi} \vert \nonumber \\
&\leq & \frac{\vert a\vert R^{5/2}\sin\theta}{\vert ^* t \vert^{1/2}} (2+\varepsilon) \sqrt{\frac{R}{\vert ^*t\vert}} \vert \hat\psi_R \vert \frac{1}{\sin \theta} \vert \hat\psi_{^*\varphi} \vert \nonumber \\
&\leq & \varepsilon \left( \frac{R}{\vert ^*t \vert} \hat\psi_R^2 + \frac{1}{\sin \theta} \hat\psi_{^*\varphi}^2 \right) \, . \label{EstimF}
\end{eqnarray}
\end{itemize}
Putting together the estimates \eqref{EstimA}-\eqref{EstimF}, the equivalence \eqref{equivEnHs} follows.

On ${\cal H}_0$, we have $R=0$ so the term in ${\cal E}_{{\cal H}_0} (\hat\psi)$ involving $\d R \wedge \d^2\omega$ drops out. Moreover, the quantity $R^2\rho^2$ tends to $1$ on ${\cal H}_0$. Hence, we get the simpler expression~:
\[ {\cal E}_{{\cal H}_0} (\hat\psi) = 4\int_{{\cal H}_s}  - T^a T^1_a \d {^*t} \wedge \d ^2 \omega \, .\]
The equivalence \eqref{equivEnH0} then follows by similar but simpler calculations.

Finally, the energy of $\hat\psi$ on $\mathcal{S}_{^*t_0}$ is non negative since~:
\begin{itemize}
\item $\mathcal{S}_{^*t_0}$ is a null hypersurface, oriented by its future-pointing null normal vector field~;
\item the stress-energy tensor $T_{ab}$ satisfies the dominant energy condition and the Morawetz vector field $T^a$ is timelike and oriented towards the future, whence $T^a_bT^b$ is a causal and future-pointing vector field.
\end{itemize}
\end{proof}

\subsection{Energy estimates and  peeling}

\subsubsection{Basic estimate}

We start by a technical lemma giving a Poincaré-type estimate.
\begin{lemma}\label{ine-Poin}
Given $^*t_0 < 0$, there exists a constant $C>0$ such that for any $f\in {\mathcal C}_0^\infty(\mathbb R)$
$$\int_{-\infty}^{^*t_0} (f(^*t))^2 \d {^*t} \leq C\int_{-\infty}^{^*t_0} {^*t}^2 (f'(^*t))^2 \d {^*t} \; .$$
\end{lemma}
This has the following consequence~:
\begin{lemma} \label{L2Control}
For ${^*t}_0 < 0$, $\vert {^*t}_0 \vert$ large enough and for any smooth compactly supported initial data at $t=0$, the associated rescaled solution $\hat\psi$ satisfies
$$ \int_{{\cal H}_s} \hat\psi^2 \d {^*t} \d^2 \omega \lesssim  {\cal E}_{{\cal H}_s} (\hat\psi) \, .$$
\end{lemma}
Both results are already present in \cite{MaNi2009} in the Schwarzschild case~; the proofs in the present framework are identical to those in \cite{MaNi2009}, therefore we do not repeat them here.

Integrating the approximate conservation law \eqref{ACL0} on the domain
\begin{equation} \label{Omegs1s2}
\Omega_{{^*t}_0}^{s_1, s_2} := \Omega^+_{{^*t}_0} \cap \{ s_1 < s < s_2 \} \mbox{ with } 0 \leq s_1 < s_2 \leq 1 \, ,
\end{equation}
we get
\begin{align*}
&\left| {\cal E}_{{\cal H}_{s_1}} (\hat\psi) + {\cal E}_{{\cal S}_{{^*t}_0}^{s_1, s_2}}(\hat\psi) - {\cal E}_{{\cal H}_{s_2}}(\hat\psi) \right| \\
&\simeq \left|\int_{s_1}^{s_2} \int_{{\cal H}_s}\left( -2\frac{Mr-a^2}{r^2} \hat\psi \nabla_T \hat\psi + \left(\nabla_a T_b  \right) T^{ab} (\hat\psi)\right) \frac{1}{|{^*t}|} \d {^*t} \d^2\omega \d s\right| \\
&\lesssim \int_{s_1}^{s_2} \int_{{\cal H}_s}\left| -2\frac{Mr-a^2}{r^2} \hat\psi \nabla_T \hat\psi + \left(\nabla_a T_b  \right) T^{ab} (\hat\psi)\right| \frac{1}{|{^*t}|} \d {^*t} \d^2 \omega \d s\, ,
\end{align*}
where
\begin{align*}
\left(\nabla_a T_b  \right) T^{ab} (\hat\psi) &= A_1 \hat\psi_{^*t}^2 + A_2 \hat\psi_{^*t} \hat\psi_R + A_3 \hat\psi_{^*t} \hat\psi_{^*\varphi} +A_4 R \hat\psi_R^2 + A_5 \hat\psi_R \hat\psi_{^*\varphi}\\
&+ A_6 \sin^2\theta \hat\psi_\theta^2 + A_7 \hat\psi_{^*\varphi}^2 + A_8 \vert \nabla_{S^2} \hat\psi \vert^2 \, , 
\end{align*}
the $A_i$'s being bounded functions on $\Omega^+_{{^*t}_0}$. We note that $s=-^*t / r_* \simeq - ^*t R$ entails
\[ \frac{1}{|^*t|} \simeq \frac{1}{\sqrt s} \sqrt{\frac{R}{|^*t|}} \, .\]
Hence,
\begin{align*}
E &:=  \left|\left( -2\frac{Mr-a^2}{r^2} \hat\psi \nabla_T \hat\psi + \left(\nabla_a T_b  \right) T^{ab} (\hat\psi)\right)\right| \frac{1}{|{^*t}|} \\
& \lesssim \left( 2R\vert{^*t}^2\hat\psi \partial_{^*t}\hat\psi \vert + 4R \vert\hat\psi \partial_R \hat\psi\vert +  A_1\vert \hat\psi_{^*t}^2 \vert+ A_2\vert \hat\psi_{^*t} \hat\psi_R \vert + A_3\vert \hat\psi_{^*t} \hat\psi_{^*\varphi}\vert \right)   \frac{1}{|{^*t}|} \\
&+ \left( A_4\vert R \hat\psi_R^2 \vert + A_5 \vert \hat\psi_R \hat\psi_{^*\varphi} \vert + A_6 \vert \sin^2\theta \hat\psi_\theta^2 \vert + A_7\vert \hat\psi_{^*\varphi}^2 \vert + A_8\vert \nabla_{S^2} \hat\psi \vert^2  \right) \frac{1}{|{^*t}|}\\ 
& \lesssim 2R \left( \hat\psi^2 + {^*t}^2\hat\psi_{^*t}^2 \right)  + \frac{4}{\sqrt s} \left( \hat\psi^2 + \frac{R}{|{^*t}|} \hat\psi_R^2 \right) +  A_1 {^*t}^2\hat\psi_{^*t}^2 + \frac{A_2}{\sqrt s} \left( \hat\psi_{^*t}^2 + \frac{R}{|{^*t}|} \hat\psi_R^2 \right) \\
& + A_3 \left( {^*t}^2\hat\psi_{^*t}^2 + \hat\psi_{^*\varphi}^2 \right) + A_4 \frac{R}{|{^*t}|} \hat\psi_R^2 + A_5 \frac{R}{\sqrt s} \left( \hat\psi_{^*\varphi}^2 + \frac{R}{|{^*t}|}\hat\psi_R^2 \right) \\
& + \left( \max\left\{ A_6,A_7 \right\}\sin^2\theta + A_8 \right)\vert \nabla_{S^2} \hat\psi \vert^2  \\
&\lesssim \frac{1}{\sqrt s} \left( {^*t}^2 \hat\psi_{^*t}^2 +  \frac{R}{\vert {^*t} \vert}\hat\psi_R^2 + \vert \nabla_{S^2}\hat\psi \vert^2 + \hat\psi^2 \right) \, .
\end{align*}
Using Lemma \ref{L2Control}, we obtain
\begin{align*}
&\left| {\cal E}_{{\cal H}_{s_1}} (\hat\psi) + {\cal E}_{{\cal S}_{{^*t}_0}^{s_1, s_2}}(\hat\psi) - {\cal E}_{{\cal H}_{s_2}}(\hat\psi) \right| \\
& \lesssim \int_{s_1}^{s_2} \frac{1}{\sqrt s} \int_{{\cal H}_s}\left( {^*t}^2 \hat\psi_{^*t}^2 +  \frac{R}{\vert {^*t} \vert}\hat\psi_R^2 + \vert \nabla_{S^2}\hat\psi \vert^2 + \hat\psi^2 \right)\d {^*t} \d^2 \omega \d s\\
&\lesssim \int_{s_1}^{s_2} \frac{1}{\sqrt s}{\cal E}_{{\cal H}_s}(\hat\psi) \d s \, .
\end{align*}
Since the function $1/\sqrt s$ is integrable on $[0,1]$, Gronwall's inequality entails the following result~:
\begin{theorem}\label{theorem_0_linear}
For ${^*t}_0 < 0, \, \vert {^*t}_0 \vert$ large enough and for any smooth compactly supported initial data at $t = 0$, the associated rescaled solution $\hat\psi$ satisfies
$${\cal E}_{\scri_{{^*t}_0}^+}(\hat\psi) \lesssim {\cal E}_{{\cal H}_1} (\hat\psi),$$
$${\cal E}_{{\cal H}_1} (\hat\psi) \lesssim {\cal E}_{\scri_{{^*t}_0}^+}(\hat\psi) + {\cal E}_{{\cal S}_{{^*t}_0}}(\hat\psi) \, .$$
\end{theorem}

\subsubsection{Higher order estimates and peeling}

In order to obtain estimates on successive derivatives of the solution, we use the approximate conservation laws \eqref{ACL1}, \eqref{ACL2}, \eqref{ACL3}, \eqref{ACL4}. First, since the vector fields $X_0 = \partial_t$ and $X_1 = \partial_{^*\varphi}$ are Killing, the approximate conservation laws \eqref{ACL1} for $\partial_{X_i} \hat\psi$, $i=0,1$, and \eqref{ACL0} for $\hat\psi$ are identical and we have an immediate corollary of Theorem \ref{theorem_0_linear}~:
\begin{corollary}\label{peeling_killing}
For ${^*t}_0 < 0, \, \vert {^*t}_0 \vert$ large enough and for any smooth compactly supported initial data at $t = 0$, the associated rescaled solution $\hat\psi$ satisfies for $i=0,1$~:
\begin{eqnarray*}
{\cal E}_{\scri_{{^*t}_0}^+}(\mathcal{L}_{X_i}\hat\psi) &\lesssim& {\cal E}_{{\cal H}_1} (\mathcal{L}_{X_i}\hat\psi)\, ,\\
{\cal E}_{{\cal H}_1} (\mathcal{L}_{X_i}\hat\psi) &\lesssim& {\cal E}_{\scri_{{^*t}_0}^+}(\mathcal{L}_{X_i}\hat\psi) + {\cal E}_{{\cal S}_{{^*t}_0}}(\mathcal{L}_{X_i}\hat\psi) \, .
\end{eqnarray*}
\end{corollary}
For $\mathcal{L}_{X_4} = \partial_R \hat\psi$, by integrating \eqref{ACL4} on $\Omega^{s_1,s_2}_{^*t_0}$, we get
\begin{align}\label{control}
&\left| {\cal E}_{{\cal H}_{s_1}} (\hat\psi_R) + {\cal E}_{{\cal S}_{{^*t}_0}^{s_1, s_2}}(\hat\psi_R) - {\cal E}_{{\cal H}_{s_2}}(\hat\psi_R) \right| \nonumber\\
& \lesssim \int_{s_1}^{s_2} \int_{{\cal H}_s} \left| \frac{r^2}{\rho^2}\left( 4a^2R \partial_{^*t}\partial_R\hat\psi + 4 aR \partial_{^*\varphi}\partial_R\hat\psi + 2a^2\partial_{^*t}\hat\psi + 2a \partial_{^*\varphi}\hat\psi \right) \nabla_T \partial_R\hat\psi \right| \frac{\d {^*t}}{|{^*t}|} \d^2\omega \d s  \nonumber\\
&+ \int_{s_1}^{s_2} \int_{{\cal H}_s} \left| \frac{r^2}{\rho^2}\left((4a^2R^3-6MR^2+2R)\partial_R^2\hat\psi +(12aR^2-12MR+2)\partial_R\hat\psi \right) \nabla_T \partial_R\hat\psi \right|  \frac{\d {^*t}}{|{^*t}|}  \d^2\omega \d s\nonumber\\
&+ \int_{s_1}^{s_2} \int_{{\cal H}_s} \left| \left(2(M-2a^2R)\hat\psi + 2\frac{Mr-a^2}{r^2} \partial_R \hat\psi \right)  \nabla_T \partial_R\hat\psi +\left(\nabla_a T_b  \right) T^{ab} \partial_R \hat\psi\right|  \frac{\d {^*t}}{|{^*t}|} \d^2\omega \d s
\end{align}
where
\begin{align*}
\left(\nabla_a T_b  \right) T^{ab} (\partial_R\hat\psi) &= A_1 (\partial_R \hat\psi)_{^*t}^2 + A_2 (\partial_R\hat\psi)_{^*t} (\partial_R\hat\psi)_R + A_3 (\partial_R \hat\psi)_{^*t} (\partial_R\hat\psi)_{^*\varphi} + A_4 R (\partial_R\hat\psi)_R^2 \\
& + A_5 (\partial_R\hat\psi)_R (\partial_R\hat\psi)_{^*\varphi} + A_6 \sin^2\theta (\partial_R\hat\psi)_\theta^2 + A_7 (\partial_R\hat\psi)_{^*\varphi}^2 + A_8 \vert \nabla_{S^2} (\partial_R\hat\psi) \vert^2 
\end{align*}
and 
$$ \nabla_T (\partial_R\hat\psi) = {^*t}^2 \partial_{^*t} (\partial_R \hat\psi) - 2(1+ {^*t}R) \partial_R (\partial_R \hat\psi) \, .$$
We can control all the terms on the right-hand side of inequality \eqref{control} in the same way as we have done in the basic estimate except for the term involving $\partial_R^2\hat\psi$. To control it, we need to use the fact that its coefficient $4a^2R^3-6MR^2+2R$ is of order one in $R$~:
\begin{align*}
& \int_{{\cal H}_s} \left| (4a^2R^3-6MR^2+2R)\partial_R^2\hat\psi \nabla_T (\partial_R\hat\psi) \right| \frac{1}{|{^*t}|} \d {^*t} \d^2\omega \\
&= \int_{{\cal H}_s} \left| (4a^2R^3-6MR^2+2R)\partial_R^2\hat\psi \left( {^*t}^2 \partial_{^*t} (\partial_R \hat\psi) - 2(1+ {^*t}R)  \partial_R (\partial_R \hat\psi)  \right) \right| \frac{1}{|{^*t}|} \d {^*t}\d^2\omega \\
&\lesssim \int_{{\cal H}_s} \left( 2R \left( (\partial_R\hat\psi_R)^2 + {^*t}^2 (\partial_{^*t}\hat\psi_R)^2 \right) + 4\frac{R}{|{^*t}|} (\partial_R \hat\psi_R)^2 \right) \d {^*t}\d^2\omega \\
& \lesssim {\cal E}_{{\cal H}_s} (\hat\psi_R) \; .
\end{align*}
Therefore the right-hand side of \eqref{control} can be controlled uniformly in $s_1,s_2$ by the quantity
\[ \int_{s_1}^{s_2}\frac{1}{\sqrt s} \left( \mathcal{E}_{\mathcal{H}_s}(\hat{\psi}_R) + \mathcal{E}_{\mathcal{H}_s}(\hat\psi) \right) \d s \, .\]
Using the Gronwall's inequality, we then get the result 
\begin{eqnarray*}
{\cal E}_{\scri_{{^*t}_0}^+}(\partial_R\hat\psi) + {\cal E}_{\scri_{{^*t}_0}^+}(\hat\psi) &\lesssim& {\cal E}_{{\cal H}_1} (\partial_R\hat\psi) + {\cal E}_{{\cal H}_1} (\hat\psi) ,\\
{\cal E}_{{\cal H}_1} (\partial_R\hat\psi)+ {\cal E}_{{\cal H}_1} (\hat\psi) &\lesssim& {\cal E}_{\scri_{{^*t}_0}^+}(\partial_R\hat\psi)+ {\cal E}_{\scri_{{^*t}_0}^+}(\hat\psi)  + {\cal E}_{{\cal S}_{{^*t}_0}}(\partial_R\hat\psi) + {\cal E}_{{\cal S}_{{^*t}_0}}(\hat\psi) \, .
\end{eqnarray*}
Iterating the process, we obtain similar estimates for $\mathcal{L}_{X_4}^k(\hat\psi)$, which we can summarize as follows~:
\begin{theorem}\label{theorem_1_linear}
For ${^*t}_0 < 0, \, \vert {^*t}_0 \vert$ large enough and for any smooth compactly supported initial data at $t = 0$, the associated rescaled solution $\hat\psi$ satisfies
\begin{eqnarray*}
{\cal E}_{\scri_{{^*t}_0}^+}(\mathcal{L}_{X_4}^k\hat\psi) &\lesssim& \sum_{p=0}^k{\cal E}_{{\cal H}_1} (\mathcal{L}_{X_4}^p\hat\psi) \, ,\\
{\cal E}_{{\cal H}_1} (\mathcal{L}_{X_4}^k\hat\psi) &\lesssim& \sum_{p=0}^k \left( {\cal E}_{\scri_{{^*t}_0}^+}(\mathcal{L}_{X_4}^p\hat\psi) + {\cal E}_{{\cal S}_{{^*t}_0}}(\mathcal{L}_{X_4}^p\hat\psi) \right) \, .
\end{eqnarray*}
\end{theorem}
This is enough for controlling transverse regularity at $\scri^+$. But if we also wish to have information on tangential regularity, we need to control the energy of an arbitrary number of Lie derivatives along $X_1$, $X_2$ and $X_3$. The case of $X_1$ has been treated in Corollary \ref{peeling_killing} using the fact that $X_1$ is Killing. For $X_2$ and $X_3$, we integrate \eqref{ACL2} and \eqref{ACL3} on $\Omega^{s_1,s_2}_{^*t_0}$. We obtain
\begin{gather*}
\left| {\cal E}_{{\cal H}_{s_1}} ({\cal L}_{X_2}\hat\psi) + {\cal E}_{{\cal S}_{{^*t}_0}^{s_1, s_2}}({\cal L}_{X_2}\hat\psi) - {\cal E}_{{\cal H}_{s_2}}({\cal L}_{X_2}\hat\psi) \right| \\
\lesssim \int_{s_1}^{s_2} \int_{{\cal H}_s} \left|  \frac{r^2}{\rho^2} \left( a^2\sin 2\theta \sin{^*\varphi} \partial_{{^*t}}^2 \hat\psi - \left( 2a \partial_{^*t} + 2a R^2 \partial_R +2aR \right) {\cal L}_{X_3} \hat\psi \right) \nabla_T{\cal L}_{X_2} \hat\psi  \right| \frac{\d {^*t}}{|{^*t}|}  \d^2\omega \d s  \\
+ \int_{s_1}^{s_2} \int_{{\cal H}_s} \left| \frac{r^2}{\rho^2} \left( -2(MR-a^2R^2) {\cal L}_{X_2}\hat\psi\right) \nabla_T{\cal L}_{X_2} \hat\psi  \right|  \frac{ \d {^*t}}{|{^*t}|} \d^2\omega \d s \\
+ \int_{s_1}^{s_2} \int_{{\cal H}_s} \left| \left(\nabla_a T_b  \right) T^{ab} ({\cal L}_{X_2}\hat\psi)  \right|  \frac{\d {^*t}}{|{^*t}|}  \d^2\omega \d s \, ,
\end{gather*}
and
\begin{gather*}
\left| {\cal E}_{{\cal H}_{s_1}} ({\cal L}_{X_3}\hat\psi) + {\cal E}_{{\cal S}_{{^*t}_0}^{s_1, s_2}}({\cal L}_{X_3}\hat\psi) - {\cal E}_{{\cal H}_{s_2}}({\cal L}_{X_3}\hat\psi) \right| \\
\lesssim \int_{s_1}^{s_2} \int_{{\cal H}_s} \left|  \frac{r^2}{\rho^2} \left( a^2\sin 2\theta \cos{^*\varphi} \partial_{{^*t}}^2 \hat\psi + \left( 2a \partial_{^*t} + 2a R^2 \partial_R + 2aR \right) {\cal L}_{X_2} \hat\psi \right)  \nabla_T({\cal L}_{X_3} \hat\psi)  \right| \frac{\d {^*t} }{|{^*t}|} \d^2\omega \d s \\
+ \int_{s_1}^{s_2} \int_{{\cal H}_s} \left| \frac{r^2}{\rho^2} \left( -2(MR-a^2R^2) {\cal L}_{X_3}\hat\psi\right) \nabla_T ({\cal L}_{X_3} \hat\psi)   \right|  \frac{\d {^*t}}{|{^*t}|}  \d^2\omega \d s \\
+ \int_{s_1}^{s_2} \int_{{\cal H}_s} \left| \left(\nabla_a T_b  \right) T^{ab} ({\cal L}_{X_3}\hat\psi)  \right|  \frac{\d {^*t}}{|{^*t}|}  \d^2\omega \d s \, .
\end{gather*}
The terms involving $\partial_{^*t}^2\hat\psi$ can be estimated using the energy of $\mathcal{L}_{X_0}(\hat\psi) = \hat\psi_{^*t}$ since
\[ \int_{{\cal H}_s}\vert \partial_{^*t}^2\hat\psi \vert^2 \d {^*t} \d^2 \omega \leq {\cal E}_{{\cal H}_s} (\hat\psi_{^*t}) \, .\]
Therefore, we can control the energy of ${\cal L}_{X_2}(\hat\psi)$ by the energies of ${\cal L}_{X_3}(\hat\psi)$ and $\mathcal{L}_{X_0}\hat\psi$. Similarly, the energy of ${\cal L}_{X_3}(\hat\psi)$ can be controlled by the energies of ${\cal L}_{X_2}(\hat\psi)$ and $\mathcal{L}_{X_0}(\hat\psi)$. 
Let us denote
\[ {\cal B} = \{ X_0 , X_2 , X_3 \} \]
and for $p\in \N$ and $S$ an oriented hypersurface,
\[ {\cal E}_{S} ({\cal L}_{\cal B}^p) := \sum_{i_1,i_2,...,i_p \in \{ 0,2,3\}}  {\cal E}_{S}  \left( {\cal L}_{X_{i_1}}...{\cal L}_{X_{i_{p}}} \right) \, .\]
We have the following result~:
\begin{proposition}\label{theorem_2_linear}
For ${^*t}_0 < 0, \, \vert {^*t}_0 \vert$ large enough and for any smooth compactly supported initial data at $t = 0$, the associated rescaled solution $\hat\psi$ satisfies for any $k \in \N$
\begin{eqnarray*}
\sum_{p=0}^k {\cal E}_{\scri_{{^*t}_0}^+} ({\cal L}_{\cal B}^{p}\hat\psi) &\lesssim & \sum_{p=0}^k {\cal E}_{{\cal H}_1} ({\cal L}_{\cal B}^{p} \hat\psi) \, , \\
\sum_{p=0}^k{\cal E}_{{\cal H}_1} ({\cal L}_{\cal B}^{p} \hat\psi) &\lesssim & \sum_{p=0}^k \left( {\cal E}_{\scri_{{^*t}_0}^+} ({\cal L}_{\cal B}^{p} \hat\psi) + {\cal E}_{{\cal S}_{{^*t}_0}} ({\cal L}_{\cal B}^{p} \hat\psi) \right) \, ,
\end{eqnarray*}
\end{proposition}
Putting together Corollary \ref{theorem_0_linear}, Theorem \ref{theorem_1_linear} and Proposition \ref{theorem_2_linear}, we obtain combined estimates involving all derivatives along the elements of ${\cal A} = \{ X_0,X_1,X_2,X_3,X_4\}$. Denoting for $p\in \N$ and for an oriented hypersurface $S$~:
\[ {\cal E}_{S} ({\cal L}_{\cal A}^p ) := \sum_{i_1,i_2,...,i_p \in \{ 0,12,3,4 \}}  {\cal E}_{S} \left( {\cal L}_{X_{i_1}}...{\cal L}_{X_{i_{p}}} \right) \,  ,\]
we have~:
\begin{theorem}\label{theorem-main}
For ${^*t}_0 < 0, \, \vert {^*t}_0 \vert$ large enough and for any smooth compactly supported initial data at $t = 0$, the associated rescaled solution $\hat\psi$ satisfies for any $k\in \N$
\begin{eqnarray*}
\sum_{p=0}^k {\cal E}_{\scri_{{^*t}_0}^+}(\mathcal{L}_{\cal A}^{p} \hat\psi) &\lesssim &\sum_{p=0}^k {\cal E}_{{\cal H}_1}(\mathcal{L}_{\cal A}^{p} \hat\psi) \, ,\\
\sum_{p=0}^k{\cal E}_{{\cal H}_1} (\mathcal{L}_{\cal A}^{p} \hat\psi) &\lesssim &\sum_{p=0}^k \left( {\cal E}_{\scri_{{^*t}_0}^+}(\mathcal{L}_{\cal A}^{p} \hat\psi) + {\cal E}_{{\cal S}_{{^*t}_0}}(\mathcal{L}_{\cal A}^{p} \hat\psi) \right) \, .
\end{eqnarray*}
\end{theorem}
Now we give the definition of the peeling at order $k\in \mathbb{N}^*$:
\begin{definition}
A solution $\psi$ of the wave equation in $\Omega^+_{^*t_0}$ peels at order $k \in \mathbb N$ if  the rescaled solution $\hat\psi$ satisfies
\[ \sum_{p=0}^k  {\cal E}_{\scri_{{^*t}_0}^+}(\mathcal{L}_{\cal A}^{p} \hat\psi) < +\infty \, , \]
\end{definition}
Theorem \ref{theorem-main} gives us a characterization of the class of initial data on ${\cal H}_1$ that guarantees that the corresponding solution peels at a given order $k$~; it is the completion of smooth compactly supported data on ${\cal H}_1$ in the norm
\begin{equation} \label{NormDatakPeel}
\left( \sum_{p=0}^k{\cal E}_{{\cal H}_1} (\mathcal{L}_{\cal A}^{p} \hat\psi) \right)^{1/2} \, .
\end{equation}
This norm however involves derivatives transverse to ${\cal H}_1$ (along $\partial_t$ and $\partial_R$) and appears therefore not to be well defined on a space of initial data. In order to replace these derivatives by other differential operators acting tangentially on ${\cal H}_1$, we use the wave equation satisfied by $\hat\psi$.
First, recall that
\[ \left( \partial_r \right)_{^*K} = V^+ = \frac{r^2+a^2}{\Delta} (\partial_t + \partial_{r_*} ) + \frac{a}{\Delta}\partial_\varphi \, ,\]
whence
\begin{equation} \label{dR}
\partial_R = -\frac{1}{R^2} \left( \partial_r \right)_{^*K} = -r^2 \left( \frac{r^2+a^2}{\Delta} (\partial_t + \partial_{r_*} ) + \frac{a}{\Delta}\partial_\varphi \right) \, .
\end{equation}
In this expression in terms of Boyer-Lindquist coordinates, the only derivative that is not tangent to ${\cal H}_1$ is $\partial_t$. Recall also that $\partial_t$ in Boyer-Lindquist coordinates and $\partial_{^*t}$ in $^*t,\,R,\,\theta,\, ^*\varphi$ coordinates are the same vector field. So all we need to do is to use the equation to express the action of $\partial_t = \partial_{^*t}$ on a solution $\hat\psi$ as a differential operator involving only derivatives acting tangentially on ${\cal H}_1$. It is better to work in Boyer-Lindquist coordinates (or in $t,\, r_*,\, \theta ,\, \varphi$ coordinates) because only $\partial_t$ is transverse to ${\cal H}_1$ whereas in $^*t,\,R,\,\theta,\, ^*\varphi$ coordinates, both $\partial_{^*t}$ and $\partial_R$ are transverse to ${\cal H}_1$. So we first need to express in $t,\, r_*,\, \theta ,\, \varphi$ coordinates the equation satisfied by $\hat\psi$. Rather than re-expressing \eqref{resc-wave} in terms of these coordinates, we give the expression of equation \eqref{WaveEq} with $\lambda =0$ in terms of $t,\, r_*,\, \theta ,\, \varphi$ and then use the fact that $\psi = \frac1r \hat\psi$. Equation \eqref{WaveEq} takes the form (see for example \cite{Ni2002})
\begin{eqnarray}
\frac{\partial^2 \psi}{\partial t^2} + \frac{4aMr}{\sigma^2}
\frac{\partial^2 \psi}{\partial t \partial \varphi} - \frac{\Delta}{\sigma^2}
\frac{\partial}{\partial r} \left( \Delta \frac{\partial \psi}{\partial
r} \right) - \frac{\Delta}{\sigma^2 \sin \theta}
\frac{\partial}{\partial \theta} \left( \sin \theta \frac{\partial
    \psi}{\partial \theta} \right) \nonumber \\
- \frac{\rho^2 -2Mr}{\sigma^2 \sin^2 \theta}\, \frac{\partial^2
  \psi}{\partial \varphi^2} + \lambda \frac{\Delta
  \rho^2}{\sigma^2} |\psi|^2 \psi = 0 \, . \label{WaveBL}
\end{eqnarray}
For $\lambda =0$ and in terms of $\hat\psi$, we obtain
\begin{eqnarray}
\frac{\partial^2 \psi}{\partial t^2} + \frac{4aMr}{\sigma^2}
\frac{\partial^2 \hat\psi}{\partial t \partial \varphi} - \frac{\Delta r}{\sigma^2}
\frac{\partial}{\partial r} \left( \Delta \frac{\partial}{\partial
r} \left( \frac1r \hat\psi \right) \right) - \frac{\Delta}{\sigma^2 \sin \theta}
\frac{\partial}{\partial \theta} \left( \sin \theta \frac{\partial
    \hat\psi}{\partial \theta} \right) \nonumber \\
- \frac{\rho^2 -2Mr}{\sigma^2 \sin^2 \theta}\, \frac{\partial^2
  \hat\psi}{\partial \varphi^2} = 0 \, . \label{WaveBLhatpsi}
\end{eqnarray}
Written as a system whose unknown is the pair $(\hat\psi \, ,~ \partial_ {t} \hat\psi )$, this becomes~:
\begin{eqnarray}
\partial_{t}\left(\begin{matrix} \hat\psi \\ \partial_{t} \hat\psi \end{matrix}\right) &=& \left( \begin{matrix} {0} & {1} \\ {\frac{\Delta r}{\sigma^2}
\partial_r \Delta \partial_r \frac1r + \frac{\Delta}{\sigma^2 \sin \theta}
\partial_\theta \sin \theta \partial_\theta + \frac{\rho^2 -2Mr}{\sigma^2 \sin^2 \theta}\, \partial^2_\varphi} & {-\frac{4aMr}{\sigma^2}
\partial_\varphi} \end{matrix} \right) \left(\begin{matrix} \hat\psi \\ \partial_{t} \hat\psi \end{matrix}\right) \label{WaveBLhatpsiHamil}
\end{eqnarray}
and we denote by $L$ the operator on the right-hand side
\begin{equation} \label{OpL}
L:=  \left( \begin{matrix} {0} & {1} \\ {\frac{\Delta r}{\sigma^2}
\partial_r \Delta \partial_r \frac1r + \frac{\Delta}{\sigma^2 \sin \theta}
\partial_\theta \sin \theta \partial_\theta + \frac{\rho^2 -2Mr}{\sigma^2 \sin^2 \theta}\, \partial^2_\varphi} & {-\frac{4aMr}{\sigma^2}
\partial_\varphi} \end{matrix} \right)  \, .
\end{equation}
Using \eqref{dR}, we also have
\begin{eqnarray}
\partial_R\left(\begin{matrix} \hat\psi \\ \partial_t\hat\psi \end{matrix}\right) &=& -r^2 \left( \frac{r^2+a^2}{\Delta} (\partial_t + \partial_{r_*} ) + \frac{a}{\Delta}\partial_\varphi \right) \left(\begin{matrix} \hat\psi \\ \partial_t\hat\psi \end{matrix}\right) \nonumber \\
&=& -r^2 \left( \frac{r^2+a^2}{\Delta} (L + \partial_{r_*} ) + \frac{a}{\Delta}\partial_\varphi \right) \left(\begin{matrix} \hat\psi \\ \partial_t\hat\psi \end{matrix}\right) =: N \left(\begin{matrix} \hat\psi \\ \partial_t\hat\psi \end{matrix}\right) \, . \label{OpN}
\end{eqnarray}
Let us consider the set of differential operators
\[ {\mathbb{A}} := \{ N \, ,~ {\cal L}_{X_1} \, ,~ {\cal L}_{X_2} \, ,~ {\cal L}_{X_3} \, ,~ L \} \, ,\]
all of them involving only derivatives tangent to ${\cal H}_1$. We can re-express the characterization of the class of data on ${\cal H}_1$ ensuring peeling at order $k$ from Theorem \ref{theorem-main} as follows~:
\begin{theorem} \label{LinearPeeling}
The space $\mathfrak{h}^k ({\cal H}_1 )$ of initial data whose associated solution peels at order $k$ is the completion of ${\cal C}_0^\infty( {\cal H}_1 ) $ in the norm~:
\[ \left|\left| \left(\begin{matrix}\hat {\psi}_0 \\ \hat {\psi}_1 \end{matrix}\right)\right|\right|^2_{\mathfrak{h}^k ({\cal H}_1 )} := \sum_{p=0}^k \sum_{(i_1,...,i_p) \, ;~ A_{i_j} \in \mathbb{A}} {\cal E}_{{\cal H}_1}\left( A_{i_1}...A_{i_p} \left(\begin{matrix}\hat {\psi}_0 \\ \hat{\psi}_1 \end{matrix}\right) \right) \]
where we have denoted by $\mathcal{E}_{\mathcal{H}_1}\left(\begin{matrix} \hat{\psi}_0\\ \hat{\psi}_1 \end{matrix}\right)$ the energy $\mathcal{E}_{\mathcal{H}_1}(\hat\psi)$ where $\hat\psi$ is replaced by $\hat\psi_0$ and $\partial_{t}\hat\psi = \hat\psi_{^*t}$ is replaced by $\hat\psi_1$.
\end{theorem}

\section{The nonlinear case} \label{NLCase}

We now turn to the study in $\Omega^+_{^*t_0}$ of equation \eqref{WaveEq} with $\lambda=1$, i.e.
\begin{equation}\label{nonlinear_wave}
\square_g \psi + \psi^3 = 0 \, ,
\end{equation}
and its conformally rescaled version
\begin{equation} \label{RescNLW}
{\square}_{\hat g} \hat\psi + \frac{1}{6}\mathrm{Scal}_{\hat g}\hat\psi + \hat{\psi}^3 = 0
\end{equation}
In star-Kerr coordinates, equation \eqref{RescNLW} has the expression~:
\begin{eqnarray}
-\frac{r^2a^2\sin^2\theta}{\rho^2}\partial_{^*t}^2\hat\psi - \frac{r^2+a^2}{\rho^2}\partial_{^*t}\partial_R\hat\psi - \frac{2a r^2}{\rho^2}\partial_{^*t}\partial_{^*\varphi}\hat\psi-\frac{a}{\rho^2}\partial_R\partial_{^*\varphi}\hat\psi \nonumber\\
-\frac{r^2}{\rho^2} \left( (1+a^2R^2) \partial_{^*t}\partial_R\hat\psi + (R^2-2MR^3+a^2R^4) \partial_R^2\hat\psi  \right) \nonumber\\
-\frac{r^2}{\rho^2}\left(  aR^2 \partial_R\partial_{^*\varphi}\hat\psi + 2a^2R \partial_{^*t}\hat\psi+(2R-6MR^2+4aR^3) \partial_R\hat\psi \right.\nonumber\\
\left.+ 2aR \partial_{^*\varphi}\hat\psi + \Delta_{S^2}\hat\psi - 2\frac{Mr-a^2}{r^2}\hat\psi\right) + \hat{\psi}^3 &=& 0 \,.\label{resc-nonwave}
\end{eqnarray}
which is equation \eqref{resc-wave} for $\lambda =1$. We use the same approach as in the linear case and insist mostly on the changes induced by the presence of the nonlinearity.

\subsection{Commutation with vector fields}

For a smooth solution $\hat\psi$ of \eqref{RescNLW}, straightforward modifications of the calculations of subsection \ref{CommutLin} give us the equations satisfied by ${\cal L}_{X_ i} \hat\psi$, $i\in \{ 0,1,2,3,4\}$~:
\begin{eqnarray}
\left(\square_{\hat g}+\frac{1}{6}\mathrm{Scal}_{\hat g}\right)\hat\psi_{^*t} + 3\hat\psi^2\hat\psi_{^*t} &=& 0 \, , \label{resc-nonwave-t} \\
\left(\square_{\hat g}+\frac{1}{6}\mathrm{Scal}_{\hat g}\right)\hat\psi_{^*\varphi} + 3\hat\psi^2\hat\psi_{^*\varphi} &=& 0 \, , \label{resc-nonwave-varphi}
\end{eqnarray}
\begin{eqnarray}
\left(\square_{\hat g}+\frac{1}{6}\mathrm{Scal}_{\hat g}\right){\cal L}_{X_2}\hat\psi &=& \frac{r^2}{\rho^2} a^2\sin 2\theta\sin{^*\varphi}\partial_{{^*t}}^2 \hat\psi  \nonumber \\
&& + \frac{r^2}{\rho^2} \left( 2a \partial_{^*t} + 2aR^2\partial_R  + 2aR  \right) {\cal L}_{X_3}\hat\psi \nonumber \\
&&- 3\hat\psi^2{\cal L}_{X_2}\hat\psi + \frac{a^2}{\rho^2}\sin 2\theta \sin{^*\varphi} \, \hat\psi^3 \, , \label{resc-nonwave2} \\
\left(\square_{\hat g}+\frac{1}{6}\mathrm{Scal}_{\hat g}\right){\cal L}_{X_3}\hat\psi &=& \frac{r^2}{\rho^2} a^2\sin 2\theta \cos {^*\varphi}\partial_{{^*t}}^2 \hat\psi  \nonumber \\
&& - \frac{r^2}{\rho^2} \left( 2a \partial_{^*t} +2aR^2\partial_R + 2aR \right) {\cal L}_{X_2} \hat\psi \nonumber \\
&&- 3\hat\psi^2{\cal L}_{X_3}\hat\psi + \frac{a^2}{\rho^2}\sin 2\theta \cos{^*\varphi}\, \hat\psi^3 \, , \label{resc-nonwave3} \\
\left(\square_{\hat g}+\frac{1}{6}\mathrm{Scal}_{\hat g}\right)\hat\psi_R &=& \frac{r^2}{\rho^2}\left( 4a^2R \partial_{^*t} \hat\psi_R + 4 aR \partial_{^*\varphi} \hat\psi_R + 2a^2\partial_{^*t}\hat\psi + 2a \partial_{^*\varphi}\hat\psi \right) \nonumber\\
&&+\frac{r^2}{\rho^2}\left((4a^2R^3-6MR^2+2R)\partial_R\hat\psi_R \right. \nonumber\\
&&\hspace{2cm}\left. + (12aR^2-12MR+2)\hat\psi_R  - 2(M-2a^2R)\hat\psi \right) \nonumber\\
&&- 3\hat\psi^2\hat\psi_R - \frac{ra^2\cos^2\theta}{\rho^2}\, \hat\psi^3 \, . \label{resc-nonwave1}
\end{eqnarray}

\subsection{Approximate conservation laws}

In the linear case, we have used the stress-energy tensor $T_{ab}$ for the wave equation and we have recovered a control on the $L^2$-norm of the field via a Poincaré-type estimate. Here we proceed slightly differently and we incorporate the $L^2$-control directly in the energy fluxes. We shall use both the stress-energy tensor for the linear Klein-Gordon equation $\square_{\hat g}\hat\psi + \hat\psi = 0$~:
\begin{eqnarray}
\T_{ab} (\hat\psi) = \T_{(ab)} (\hat\psi) &=& \partial_a\hat\psi\partial_b\hat\psi - \frac{1}{2}\langle \nabla \hat\psi\, ,\nabla \hat\psi\rangle_{\hat{g}} \hat g_{ab} + \frac12 \hat\psi^2 \hat g_{ab} \label{SETKG}\\
&=& T_{ab} (\hat\psi) + \frac12\hat\psi^2 \hat g_{ab} \nonumber
\end{eqnarray}
and that for the nonlinear equation $\square_{\hat g}\hat\psi + \hat\psi + \hat\psi^3 = 0$~:
\begin{equation} \label{SETNLKG}
\tilde \T_{ab} (\hat\psi) = \tilde \T_{(ab)}(\hat\psi) = \T_{ab} (\hat\psi) + \frac{1}{4}\hat\psi^4 \hat{g}_{ab}\, .
\end{equation}
For a solution $\hat{\psi}$ of equation \eqref{RescNLW}, we have
\begin{eqnarray}
T^b \nabla^a \T_{ab} (\hat{\psi}) = \left( \square_{\hat{g}} \hat{\psi} + \hat{\psi}  \right) \nabla_T \hat{\psi} &=& \left( \left( 1 - \frac16 \mathrm{Scal}_{\hat{g}} \right) \hat{\psi} -  \hat{\psi}^3 \right) \nabla_T \hat{\psi} \nonumber \\
&=& \left( \left( 1 - 2\frac{Mr-a^2}{r^2} \right) \hat{\psi} -  \hat{\psi}^3 \right) \nabla_T \hat{\psi}\, , \label{SETCL1}\\
T^b \nabla^a \T_{ab} (\hat{\psi}) = \left( \square_{\hat{g}} \hat{\psi} + \hat{\psi} + \hat{\psi}^3 \right) \nabla_T \hat{\psi} &=& \left( 1 - \frac16 \mathrm{Scal}_{\hat{g}} \right) \hat{\psi} \nabla_T \hat{\psi} \nonumber \\
&=& \left( 1 - 2\frac{Mr-a^2}{r^2} \right) \hat{\psi} \nabla_T \hat{\psi}\,  .\label{SETCL2}
\end{eqnarray}
The contraction of the Killing form of the Morawetz vector field $T^a$ with both tensors \eqref{SETKG} and \eqref{SETNLKG} are simple modifications of its contraction with $T_{ab}$
\begin{eqnarray*}
\nabla_{(a}T_{b)} \T^{ab} &=& \nabla_{(a}T_{b)} T^{ab} + \frac12 \hat\psi^2 \nabla_a T^a \, , \\
\nabla_{(a}T_{b)} \tilde{\T}^{ab} &=& \nabla_{(a}T_{b)} \T^{ab} + \frac{1}{4}\hat\psi^4 \nabla_a T^a
\end{eqnarray*}
and from Lemma \ref{Killing-form} and the expression \eqref{inverse-metric} of the inverse rescaled metric $\hat{g}^{-1}$, the divergence of $T^a$ is given by
\begin{eqnarray*}
\nabla_a T^a &=& 4\left( (1+{^*t}R)M \frac{\partial}{\partial R}\left(\frac{R}{\rho^2}\right)-\frac{2M{^*t}R}{\rho^2}\right) \left( - \frac{r^2a^2\sin^2 \theta}{\rho^2} \right) \\
&& -4a\sin^2 \theta \, {^*t} \left( -\frac{2a}{\rho^2} \right) \\
&& - 4a\sin^2\theta \left\{ 2(1+{^*t}R)M \frac{\partial}{\partial R}\left(\frac{R}{\rho^2}\right) - 2 M{^*t} \frac{R}{\rho^2}+ R \right\} \left( -\frac{2ar^2}{\rho^2} \right) \\
&& + 4a^2\cos^2\theta (1+{^*t}R)R \left( -\frac{r^2}{\rho^2} \right) \\
&& + 4a^2\sin^2\theta (1+{^*t}R)\left\{ R + \sin^2\theta M \frac{\partial}{\partial R}\left( \frac{R}{\rho^2}\right) \right\} \left( -\frac{r^2}{\rho^2\sin^2 \theta} \right) \, .
\end{eqnarray*}
This is of order $1$ in $R$ and using Lemma \ref{short-morawetz}, which was giving a short-hand expression for $\nabla_{(a}T_{b)} T^{ab}$, we can write~:
\begin{eqnarray}
\left(\nabla_a T_b  \right) \T^{ab} (\hat\psi) &=& A_1 \hat\psi_{^*t}^2 + A_2 \hat\psi_{^*t} \hat\psi_R + A_3 \hat\psi_{^*t} \hat\psi_{^*\varphi} +A_4 R \hat\psi_R^2 + A_5 \hat\psi_R \hat\psi_{^*\varphi}\nonumber \\
&&+ A_6 \sin^2\theta \hat\psi_\theta^2 + A_7 \hat\psi_{^*\varphi}^2 + A_8 \vert \nabla_{S^2} \hat\psi \vert^2 + A_9 \hat\psi^2 \, , \label{A1A9}\\
\left(\nabla_a T_b  \right) \tilde{\T}^{ab} (\hat\psi) &=& A_1 \hat\psi_{^*t}^2 + A_2 \hat\psi_{^*t} \hat\psi_R + A_3 \hat\psi_{^*t} \hat\psi_{^*\varphi} +A_4 R \hat\psi_R^2 + A_5 \hat\psi_R \hat\psi_{^*\varphi} \nonumber \\
&&+ A_6 \sin^2\theta \hat\psi_\theta^2 + A_7 \hat\psi_{^*\varphi}^2 + A_8 \vert \nabla_{S^2} \hat\psi \vert^2 + A_9 \hat\psi^2 +A_{10}\hat\psi^4 \, , \label{A1A10}
\end{eqnarray}
where the functions $A_i \, (i=1,2...8)$ are the ones given in Lemma \ref{short-morawetz} and
\begin{equation} \label{A9A10}
A_9 = 2 A_{10}= \frac12 \nabla_a T^a \, .
\end{equation}
All functions $A_i$, $i=1,...,10$, are bounded in $\Omega^+_{^*t_0}$ and $A_9$ and $A_{10}$ are of order $1$ in $R$.

From \eqref{SETCL2}, the nonlinear energy current
\begin{equation} \label{NLEnCurrent}
\tilde{\J}_a (\hat\psi ) := T^b \tilde{\T}_{ab} (\hat\psi )
\end{equation}
satisfies the approximate conservation law
\begin{eqnarray}
\nabla^a \tilde{\J}_a (\hat\psi)&=& \left( \hat\psi - 2\frac{Mr-a^2}{r^2}\hat\psi \right) \nabla_T \hat\psi + \left(\nabla_a T_b\right) \tilde{\T}^{ab} (\hat\psi) \nonumber\\
& =& \left( \hat\psi - 2\frac{Mr-a^2}{r^2}\hat\psi \right) \left( {^*t}^2\partial_{^*t}\hat\psi - 2(1+{^*t}R)\partial_R \hat\psi \right) + \left(\nabla_a T_b\right) \tilde{\T}^{ab} (\hat\psi) \, .\label{conser-non}
\end{eqnarray}
Similarly, from \eqref{SETCL1}, the linear energy current
\begin{equation} \label{LEnCurrent}
\J_a (\hat\psi ):= T^b \T_{ab} (\hat\psi ) \, .
\end{equation}
satisfies
\begin{eqnarray}
\nabla^a {\J}_a (\hat\psi)&=& \left( \hat\psi - 2\frac{Mr-a^2}{r^2}\hat\psi - \hat{\psi}^3 \right) \nabla_T \hat\psi + \left(\nabla_a T_b\right) {\T}^{ab} (\hat\psi) \nonumber\\
& =& \left( \hat\psi - 2\frac{Mr-a^2}{r^2}\hat\psi - \hat{\psi}^3 \right) \left( {^*t}^2\partial_{^*t}\hat\psi - 2(1+{^*t}R)\partial_R \hat\psi \right) + \left(\nabla_a T_b\right) {\T}^{ab} (\hat\psi) \, .\label{ACLLin}
\end{eqnarray}
For the approximate conservation laws satisfied by the successive derivatives of $\hat\psi$, we shall only use the stress energy tensor $\T_{ab}$ for the linear Klein-Gordon equation, since these derivatives in fact satisfy linear equations with sources. Using equation \eqref{resc-nonwave-t} we have
\begin{eqnarray}
\nabla^a \J_a (\hat\psi_{^*t})  &=& \left( \square_{\hat g}\hat\psi_{^*t} + \hat\psi_{^*t} \right)  \nabla_T \hat\psi_{^*t} + \left(\nabla_a T_b\right) \T^{ab} (\hat\psi_{^*t}) \nonumber\\
& = & \left( \hat\psi_{^*t} - 2\frac{Mr-a^2}{r^2}\hat\psi_{^*t} - 3\hat\psi^2\hat\psi_{^*t} \right) \left( {^*t}^2\partial_{^*t}\hat\psi_{^*t} - 2(1+{^*t}R)\partial_R \hat\psi_{^*t} \right) \nonumber \\&&+ \left(\nabla_a T_b\right) \T^{ab} (\hat\psi_{^*t}) 
\label{conser-non-t}
\end{eqnarray}
and a similar formula for $\hat\psi_{^*\varphi}$ is obtained using \eqref{resc-nonwave-varphi}
\begin{eqnarray}
\nabla^a \J_a (\hat\psi_{^*\varphi})  &=& \left( \square_{\hat g} \hat\psi_{^*\varphi} + \hat\psi_{^*\varphi} \right) \nabla_T \hat\psi_{^*\varphi} + \left(\nabla_a T_b\right) \T^{ab} (\hat\psi_{^*\varphi}) \nonumber\\
& =& \left( \hat\psi_{^*\varphi} - 2\frac{Mr-a^2}{r^2}\hat\psi_{^*\varphi} - 3\hat\psi^2\hat\psi_{^*\varphi} \right) \left( {^*t}^2\partial_{^*t}\hat\psi_{^*\varphi} - 2(1+{^*t}R)\partial_R \hat\psi_{^*\varphi} \right) \nonumber \\
&&+ \left(\nabla_a T_b\right) \T^{ab} (\hat\psi_{^*\varphi}) \, .
\label{conser-non-varphi}
\end{eqnarray}
The conservation laws for ${\cal L}_{X_i} \, (i=2,3,4)$ are more complicated. They are obtained from equations \eqref{resc-nonwave2} and \eqref{resc-nonwave3} and \eqref{resc-nonwave1} as follows~:
\begin{eqnarray}
\nabla^a \J_a ({\cal L}_{X_2}\hat\psi) &=& \left( \square_{\hat g}{\cal L}_{X_2} \hat\psi + {\cal L}_{X_2} \hat\psi \right) \nabla_T {\cal L}_{X_2}\hat\psi + \left(\nabla_a T_b  \right) \T^{ab} ({\cal L}_{X_2}\hat\psi ) \nonumber\\
&=&\frac{r^2}{\rho^2} \left( a^2\sin 2\theta \sin{^*\varphi} \partial_{{^*t}}^2 \hat\psi - 2a \partial_{^*t} {\cal L}_{X_3}\hat\psi - 2a R^2 \partial_R {\cal L}_{X_3} \hat\psi - 2aR {\cal L}_{X_3}\hat\psi \right)  \nabla_T{\cal L}_{X_2} \hat\psi \nonumber\\
&& + \left( {\cal L}_{X_2}\hat\psi - 3\hat\psi^2{\cal L}_{X_2}\hat\psi + \frac{a^2}{\rho^2}\sin 2\theta \sin{^*\varphi}\hat\psi^3 - 2\frac{Mr - a^2}{r^2}{\cal L}_{X_2}\hat\psi  \right)  \nabla_T {\cal L}_{X_2} \hat\psi \nonumber\\
&& + \left(\nabla_a T_b  \right) \T^{ab} ({\cal L}_{X_2}\hat\psi ) \, ,
\label{conser-non2}
\end{eqnarray}
and 
\begin{eqnarray}
\nabla^a \J_a  ({\cal L}_{X_3}\hat\psi) &=& \left( \square_{\hat g}{\cal L}_{X_3} \hat\psi + {\cal L}_{X_3} \hat\psi \right) \nabla_T {\cal L}_{X_3}\hat\psi + \left(\nabla_a T_b  \right) \T^{ab} ({\cal L}_{X_3}\hat\psi) \nonumber\\
&=&\frac{r^2}{\rho^2} \left( a^2\sin 2\theta \cos {^*\varphi} \partial_{{^*t}}^2 \hat\psi + 2a \partial_{^*t} {\cal L}_{X_2}\hat\psi + 2a R^2 \partial_R {\cal L}_{X_2} \hat\psi + 2aR {\cal L}_{X_2}\hat\psi \right) \nabla_T {\cal L}_{X_3} \hat\psi \nonumber\\
&& + \left( {\cal L}_{X_3}\hat\psi - 3\hat\psi^2{\cal L}_{X_3}\hat\psi + \frac{a^2}{\rho^2}\sin 2\theta \cos{^*\varphi}\hat\psi^3 - 2\frac{Mr - a^2}{r^2}{\cal L}_{X_3}\hat\psi  \right)  \nabla_T{\cal L}_{X_3} \hat\psi \nonumber\\
&&+ \left(\nabla_a T_b  \right) \T^{ab} ({\cal L}_{X_3}\hat\psi) \, ,
\label{conser-non3}
\end{eqnarray}
and
\begin{eqnarray}
\nabla^a \J_a  (\hat\psi_R) & =& \left( \square_{\hat g}\hat\psi_R + \hat\psi_R \right) \nabla_T \hat\psi_R + \left(\nabla_a T_b  \right) \T^{ab} (\hat\psi_R) \nonumber\\
& = & \frac{r^2}{\rho^2}\left( 4a^2R \partial_{^*t}\hat\psi_R + 4 aR \partial_{^*\varphi}\hat\psi_R + 2a^2\partial_{^*t}\hat\psi + 2a \partial_{^*\varphi}\hat\psi \right) \nabla_T \hat\psi_R \nonumber\\
&&+\frac{r^2}{\rho^2}\left((4a^2R^3-6MR^2+2R)\partial_R\hat\psi_R + (12aR^2-12MR+2)\hat\psi_R \right) \nabla_T \hat\psi_R \nonumber\\
&&- \left( 2(M-2a^2R)\hat\psi + 2\frac{Mr-a^2}{r^2}\hat\psi_R \right)  \nabla_T \hat\psi_R  \nonumber\\
&& + \left( \hat\psi_R - 3\hat\psi^2\hat\psi_R - \frac{ra^2\cos^2\theta}{\rho^2}\hat\psi^3  \right) \nabla_T \hat\psi_R +\left(\nabla_a T_b  \right) \T^{ab} (\hat\psi_R) \, .\label{conser-non1}
\end{eqnarray}

\subsection{Energy fluxes}

We shall denote by $\E_{S}$ and $\tilde{\E}_{S}$ the energy fluxes across an oriented hypersurface $S$ associated with the currents $\J$ (defined in \eqref{LEnCurrent}) and $\tilde{\J}$ (defined in \eqref{NLEnCurrent}). We give their simplified equivalent forms for the hypersurfaces ${\cal H}_s$, $0\leq s\leq 1$ in the following lemma~:
\begin{lemma}
The energy fluxes of a solution $\hat\psi$ of \eqref{resc-nonwave} associated with the non-linear energy current $\tilde{\J}$ have the equivalent simplified form~:
\begin{eqnarray*}
\tilde{\E}_{{\cal H}_s}(\hat\psi) &\simeq& \int_{{\cal H}_s} \left( {^*t}^2\hat\psi_{^*t}^2 + \frac{R}{|{^*t}|}\hat\psi_R^2 + \vert \nabla_{S^2}\hat\psi \vert^2 + \hat\psi^2 + \hat\psi^4 \right) \d {^*t} \d^2\omega \, , ~0<s\leq 1 \, , \\
\tilde{\E}_{\scri^+_{{^*t}_0}} (\hat\psi) &\simeq& \int_{\scri^+_{{^*t}_0}} \left( {^*t}^2 \hat\psi_{^*t}^2 + \vert \nabla_{S^2}\hat\psi \vert^2 + \hat\psi^2 + \hat\psi^4\right) \d {^*t} \d^2 \omega 
\end{eqnarray*}
and for the energy fluxes associated with the linear current $\J$, we have~:
\begin{eqnarray*}
{\cal E}_{{\cal H}_s}(\hat\psi) &\simeq & \int_{{\cal H}_s} \left( {^*t}^2\hat\psi_{^*t}^2 + \frac{R}{|{^*t}|}\hat\psi_R^2 + \vert \nabla_{S^2} \hat\psi \vert^2 + \hat\psi^2 \right) \d {^*t} \d^2\omega \, ,~0<s\leq 1 \, ,\\
{\cal E}_{\scri^+_{{^*t}_0}} (\hat\psi) &\simeq & \int_{\scri^+_{{^*t}_0}} \left( {^*t}^2 \hat\psi_{^*t}^2 + \vert \nabla_{S^2}\hat\psi \vert^2 + \hat\psi^2 \right) \d {^*t} \d^2 \omega \, .
\end{eqnarray*}
Both stress energy tensors $\T_{ab}$ and $\tilde{\T}_{ab}$ satisfy the dominant energy condition, hence
\[ {\E}_{{\cal S}_{^*t}}(\hat\psi) \geq 0 \, ,~ \tilde{\E}_{{\cal S}_{^*t}}(\hat\psi) \geq 0 \, .\]
\end{lemma}
\begin{proof}
The proof is similar to the linear case.
\end{proof}
\begin{remark}
Note that $( \tilde{\E}_{{\cal H}_s} )^{1/2}$, $0<s\leq 1$, does not define a norm due to the presence of the nonlinear term $\hat\psi^4$.
\end{remark}

\subsection{Energy estimates and peeling}

\subsubsection{Basic estimate}

Integrating the conservation law \eqref{conser-non} for $\hat\psi$ on $\Omega_{^*t_0}^{s_1,s_2}$, we obtain~:
\begin{gather}
\left| \tilde{\E}_{{\cal H}_{s_1}}(\hat\psi) + \tilde{\E}_{{\cal S}_{u_0}^{s_1,s_2}}(\hat\psi) - \tilde{\E}_{{\cal H}_{s_2}}(\hat\psi) \right| \nonumber\\
\lesssim \int_{s_1}^{s_2} \int_{{\cal H}_s} \left(  \left| \left( 1 - 2\frac{Mr-a^2}{r^2} \right)\hat\psi  \left( {^*t}^2\partial_{^*t}\hat\psi - 2(1+{^*t}R)\partial_R \hat\psi \right) \right| \right.\nonumber\\
\hspace{1in}+ \left| \left(\nabla_a T_b\right) \tilde{\T}^{ab} (\hat\psi)  \right| \left) \frac{1}{\vert {^*t} \vert}\d {^*t} \d^2 \omega \d s \right.\nonumber\\
\lesssim \int_{s_1}^{s_2} \int_{{\cal H}_s} \left( \left|{^*t}\hat\psi \partial_{^*t}\hat\psi\right| + \left|\hat\psi \partial_R\hat\psi\right|\frac{1}{|^*t|} + \left| \left(\nabla_a T_b\right) \tilde{\T}^{ab} (\hat\psi)  \right| \frac{1}{|{^*t}|} \right)\d {^*t} \d^2\omega \d s \nonumber\\
\lesssim \int_{s_1}^{s_2} \int_{{\cal H}_s} \left( \left|{^*t}\hat\psi \partial_{^*t}\hat\psi\right| + \left|\hat\psi \partial_R\hat\psi\right|\frac{1}{\sqrt s}\sqrt{\frac{R}{|^*t|}} + \left| \left(\nabla_a T_b\right) \tilde{\T}^{ab} (\hat\psi)  \right| \frac{1}{|{^*t}|} \right)\d {^*t} \d^2\omega \d s \nonumber\\
\lesssim \int_{s_1}^{s_2} \int_{{\cal H}_s} \left\{ \left ( \hat\psi^2 + {^*t}^2\hat\psi_{^*t}^2 \right) + \frac{1}{\sqrt s}\left( \hat\psi^2 + \frac{R}{\vert {^*t} \vert }\hat\psi_R^2  \right) + \left| \left(\nabla_a T_b\right) \tilde{\T}^{ab} (\hat\psi)  \right| \frac{1}{|{^*t}|} \right\}\d {^*t} \d^2\omega \d s \, . \label{inter0}
\end{gather}
Here we have used the equivalence
$$\frac{1}{|^*t|} \simeq \frac{1}{\sqrt s}\sqrt{\frac{R}{|^*t|}} \, ,$$ 
which will also allow us to control the last term in \eqref{inter0} as follows~:
\begin{eqnarray}
 \left| \left(\nabla_a T_b  \right) \tilde{\T}^{ab} (\hat\psi) \right| \frac{1}{\vert {^*t} \vert} & \leq & \left| A_1 \hat\psi_{^*t}^2 + A_2 \hat\psi_{^*t} \hat\psi_R + A_3 \hat\psi_{^*t} \hat\psi_{^*\varphi} +A_4 R \hat\psi_R^2 + A_5 \hat\psi_R \hat\psi_{^*\varphi} \right| \frac{1}{\vert {^*t} \vert} \nonumber \\
&&+ \left| A_6 \sin^2\theta \hat\psi_\theta^2 + A_7 \hat\psi_{^*\varphi}^2 + A_8 \vert \nabla_{S^2} \hat\psi \vert^2 + A_9 \hat\psi^2 + A_{10} \hat{\psi}^4 \right| \frac{1}{\vert {^*t} \vert} \nonumber \\
& \lesssim &\left( \vert A_1\vert + \vert A_3\vert \right) \hat\psi_{^*t}^2 + \frac{\vert A_2\vert}{\sqrt s} \left( \hat{\psi}_{^*t}^2 +  \frac{R}{|{^*t}|} \hat\psi_R^2 \right) + \vert A_4\vert \frac{R}{|{^*t}|} \hat\psi_R^2 \nonumber \\
&&+(\vert A_3\vert +\vert A_5\vert + \vert A_7\vert )\hat\psi_{^*\varphi}^2 + \vert A_6\vert \sin^2\theta\hat\psi_\theta^2 + \vert A_8\vert |\nabla_{S^2}\hat\psi|^2 \nonumber \\
&&+ \vert A_9\vert \hat\psi^2 + \vert A_{10} \vert \hat{\psi}^4 \nonumber  \\
& \lesssim & \left(\vert A_1\vert+ \frac{1}{\sqrt s}\vert A_2\vert +\vert A_3\vert  \right) {^*t}^2 \hat\psi_{^*t}^2 + \frac{\vert A_2\vert }{\sqrt s} \frac{R}{|{^*t}|} \hat\psi_R^2 + \vert A_4\vert \frac{R}{|{^*t}|} \hat\psi_R^2 \nonumber \\
&&+ \left( (\vert A_3\vert + \vert A_5\vert + \vert A_6\vert + \vert A_7\vert )\sin^2\theta + \vert A_8 \vert \right) |\nabla_{S^2}\hat\psi|^2 \nonumber \\
&& + \vert A_9 \vert \hat\psi^2 + \vert A_{10} \vert \hat{\psi}^4 \nonumber \\
& \lesssim& \frac{1}{\sqrt s} \left( {^*t}^2 \hat\psi_{^*t}^2 + \frac{R}{|{^*t}|} \hat\psi_R^2 +  |\nabla_{S^2}\hat\psi|^2 + \hat\psi^2 + \hat{\psi}^4 \right) \, . \label{KFNL}
\end{eqnarray}
This together with \eqref{inter0} implies
\begin{equation} \label{interNLLevel0}
\left| \tilde{\E}_{{\cal H}_{s_1}}(\hat\psi) + \tilde{\E}_{{\cal S}_{u_0}^{s_1,s_2}}(\hat\psi) - \tilde{\E}_{{\cal H}_{s_2}}(\hat\psi) \right| \lesssim \int_{s_1}^{s_2} \frac{1}{\sqrt s}  \tilde{\E}_{{\cal H}_s}(\hat\psi) \d s \, .
\end{equation}
Using Gronwall's inequality, we then get the following result~:
\begin{proposition}\label{orderKerr0}
For ${^*t}_0 < 0$, $\vert {^*t}_0 \vert$ large enough, and for any smooth compactly supported initial data at $t = 0$, the associated rescaled solution $\hat\psi$ satisfies for all $0 \leq s_1 < s_2 \leq 1$,
\begin{eqnarray*}
\tilde{\E}_{{\cal H}_{s_1}}(\hat\psi) &\lesssim& \tilde{\E}_{{\cal H}_{s_2}}(\hat\psi) \, ,\\
\tilde{\E}_{{\cal H}_{s_2}}(\hat\psi) &\lesssim& \tilde{\E}_{\mathcal{H}_{s_1}}(\hat\psi)  + \tilde{\E}_{{\cal S}_{{^*t}_0}^{s_1,s_2}}(\hat\psi) \, .
\end{eqnarray*}
In particular for $s_1 = 0, \, s_2 =1$ we get
\begin{eqnarray*}
\tilde{\E}_{\scri_{{^*t}_0}^+}(\hat\psi) &\lesssim& \tilde{\E}_{{\cal H}_1}(\hat\psi) \, ,\\
\tilde{\E}_{{\cal H}_1}(\hat\psi) &\lesssim& \tilde{\E}_{\scri_{{^*t}_0}^+}(\hat\psi)  + \tilde{\E}_{{\cal S}_{{^*t}_0}}(\hat\psi) \, .
\end{eqnarray*}
\end{proposition}
We now wish to estimate the linear energies which do define a norm and can be used to construct a function space by completion of smooth compactly supported functions. We can either prove the estimates from the approximate conservation law \eqref{ACLLin} for the linear energy current or deduce them from the estimates on the non linear energies. Either way, our main tool will be a Sobolev embedding $H^1 \hookrightarrow L^6$ on ${\cal H}_s$. We start by establishing the following result.
\begin{proposition} \label{EstimNLLinEn}
There exists a constant $C>0$ such that, for all $s\in [0,1]$, for all $\hat\psi \in {\cal C}^\infty_0 (\Omega^+_{^*t_0})$,
\[ \Vert \hat\psi \Vert_{L^6 ({\cal H}_s)} \leq \left( \E_{{\cal H}_s} (\hat\psi) \right)^{1/2} \, ,\]
where the $L^6$ norm on ${\cal H}_s$ is taken for the measure $\d ^*t \d^2 \omega$, i.e.
\[ \Vert \hat\psi \Vert_{L^6 ({\cal H}_s)}^6 = \int_{{\cal H}_s} \hat\psi^6  \d {^*t} \d^2 \omega \, .\]
\end{proposition}
\begin{proof}
Let us consider the infinite half-cylinder $M=]0,+\infty[_x \times S^2_{\omega}$ endowed with the metric
\[ h = \d x^2 + \d \omega^2 \, ,\]
$\d \omega^2$ being the euclidean metric on $S^2$. We start by establishing the following result
\begin{lemma} \label{SobHalfCyl}
The Sobolev embedding $H^1 (M) \hookrightarrow L^6 (M)$ is valid.
\end{lemma}
\begin{proof}
We use a Babbitch extension and a classic result for complete manifolds. We know from T. Aubin \cite{Au} that for any complete Riemannian Manifold for which the sectional curvatures are uniformly bounded and the injectivity radius is strictly positive, the Sobolev embedding of $H^1$ into $L^6$ is valid. In particular this is true for the whole cylinder ${\cal C} = \R \times S^2$ equipped with the metric $h$, i.e. there exists $C>0$ such that, for all $f \in H^1({\cal C})$,
\[ \Vert f \Vert_{L^6 ({\cal C})} \leq C \Vert f \Vert_{H^1 ({\cal C})} \, . \]
Now we consider the Babbitch extension from $M$ to $\cal C$ which is an extension by reflection~: to a function $\phi$ on $M$ we associate the function $T\phi:=\tilde{\phi}$ defined on $\cal C$ by~:
\[ \tilde{\phi} (x,\omega) = \left\{ \begin{array}{lll} \phi (x,\omega) & if & x\geq 0 \, ,\\\phi ( -x,\omega) & if & x\leq 0 \, . \end{array} \right. \]
It is straightforward to check that $T$ is a bounded linear map from $H^1 (M)$ to $H^1({\cal C})$ and in addition forall $\phi \in {\cal C}^\infty_0 (M)$ (the support of $\phi$ is allowed to contain the boundary $\{ x=0 \}$),
\[ \Vert \tilde{\phi} \Vert_{H^1({\cal C})} = \sqrt{2} \Vert \phi \Vert_{H^1(M)} \, ,~ \Vert \tilde{\phi} \Vert_{L^6 ({\cal C})} = 2^{1/6}\Vert \phi \Vert_{L^6 (M)} \, .\]
Putting the two together, we have for all $\hat\psi \in {\cal C}^\infty_0 (M)$~:
\[ \Vert \phi \Vert_{L^6 (M)} = \frac{1}{2^{1/6}} \Vert \tilde{\phi} \Vert_{L^6 ({\cal C})} \leq \frac{C}{2^{1/6}} \Vert \tilde{\phi} \Vert_{H^1 ({\cal C})} = \frac{C\sqrt{2}}{2^{1/6}} \Vert \phi \Vert_{H^1 (M)} \, .\]
The result follows by density of ${\cal C}^\infty_0 (M)$ in $H^1(M)$.
\end{proof}
Lemma \ref{SobHalfCyl} gives us for $\hat\psi \in {\cal C}^\infty_0 (\Omega^+_{^*t_0})$
\[ \left( \int_{\mathcal{H}_s}\hat{\psi}^6 \d {^*t} \d^2 \omega\right)^{1/3} \lesssim \int_{{\cal H}_s} \left( (\partial_{^*t}(\hat{\psi}|_{\mathcal {H}_s}))^2 + |\nabla_{S^2}\hat\psi|^2 + \hat{\psi}^2 \right) \d {^*t} \d^2 \omega \, \]
and, using the fact that on ${\cal H}_s$, ${^*t} = -sr_*$,
\[ \partial_{^*t}(\hat\psi|_{\mathcal{H}_s}) = \hat\psi_{^*t} + \frac{r_*R^2\Delta}{|{^*t}|(r^2+a^2)}\hat\psi_R \simeq \hat\psi_{^*t} + \frac{R}{|{^*t}|}\hat\psi_R \, .\]
Hence,
\begin{eqnarray*}
\left( \int_{{\cal H}_s}\hat\psi^6 \d {^*t} \d \omega \right)^{1/3} &\lesssim&  \int_{{\cal H}_s} \left( \left( \hat\psi_{^*t} \right)^2 + \frac{R^2}{|{^*t}|^2} \left( \hat\psi_R\right)^2 + \vert \nabla_{S^2}\hat\psi \vert^2 + \hat\psi^2\right) \d {^*t} \d^2\omega \\
&\lesssim &\int_{{\cal H}_s} \left( \left( \hat\psi_{^*t} \right)^2 + \frac{R}{|{^*t}|} \left( \hat\psi_R\right)^2 + \vert \nabla_{S^2}\hat\psi \vert^2 + \hat\psi^2\right) \d {^*t} \d^2\omega \\
&\lesssim &{\E}_{{\cal H}_s} (\hat\psi) \, .
\end{eqnarray*} 
This concludes the proof of the proposition.
\end{proof}
In particular, this allows to control the non linear energy $\tilde{\E}_{{\cal H}_s}(\hat\psi)$ in terms of the linear energy $\E_{{\cal H}_s}(\hat\psi)$. We have
\[ \int_{{\cal H}_s} \hat\psi^4 \d {^*t}\d^2 \omega \leq \frac{1}{2} \int_{{\cal H}_s} (\hat\psi^6 + \hat\psi^2)  \d {^*t}\d^2 \omega \lesssim \left( {\E}_{{\cal H}_s}(\hat\psi)  \right)^3 + {\E}_{{\cal H}_s}(\hat\psi) \, ,\]
whence
\begin{equation}\label{basic-ine}
{\E}_{{\cal H}_s}(\hat\psi) \leq \tilde{\E}_{{\cal H}_s}(\hat\psi) \lesssim \left( {\E}_{{\cal H}_s}(\hat\psi) \right)^3 + {\E}_{{\cal H}_s}(\hat\psi) \, .
\end{equation}
We now turn to controlling the linear energy on ${\cal H}_s$. Integrating \eqref{ACLLin} on $\Omega^{s_1,s_2}_{^*t_0}$, we obtain
\begin{gather*}
\left| {\E}_{{\cal H}_{s_1}}(\hat\psi) + {\E}_{{\cal S}_{u_0}^{s_1,s_2}}(\hat\psi) - {\E}_{{\cal H}_{s_2}}(\hat\psi) \right| \nonumber\\
\lesssim \int_{s_1}^{s_2} \int_{{\cal H}_s} \left(  \left| \left( \left( 1 - 2\frac{Mr-a^2}{r^2} \right)\hat\psi -\hat{\psi}^3 \right) \left( {^*t}^2\partial_{^*t}\hat\psi - 2(1+{^*t}R)\partial_R \hat\psi \right) \right| \right.\nonumber\\
\hspace{1in}+ \left| \left(\nabla_a T_b\right) {\T}^{ab} (\hat\psi)  \right| \left) \frac{1}{\vert {^*t} \vert}\d {^*t} \d^2 \omega \d s \right.\nonumber\\
\lesssim \int_{s_1}^{s_2} \frac{1}{\sqrt s} \int_{{\cal H}_s} \left\{ \hat\psi^2 + \hat{\psi}^6 + {^*t}^2\hat\psi_{^*t}^2  + \frac{R}{\vert {^*t} \vert }\hat\psi_R^2  + \left| \left(\nabla_a T_b\right) {\T}^{ab} (\hat\psi)  \right| \frac{1}{|{^*t}|} \right\}\d {^*t} \d^2\omega \d s \, .
\end{gather*}
Using Proposition \ref{EstimNLLinEn} and the fact that
\[ \left\vert \left(\nabla_a T_b\right) {\T}^{ab} (\hat\psi) \right\vert \lesssim \frac{1}{\sqrt s} \left( {^*t}^2 \hat\psi_{^*t}^2 + \frac{R}{|{^*t}|} \hat\psi_R^2 +  |\nabla_{S^2}\hat\psi|^2 + \hat\psi^2 \right) \, ,\]
whose proof is similar to that of \eqref{KFNL}, we get
\[ \left| {\E}_{{\cal H}_{s_1}}(\hat\psi) + {\E}_{{\cal S}_{u_0}^{s_1,s_2}}(\hat\psi) - {\E}_{{\cal H}_{s_2}}(\hat\psi) \right| \lesssim \int_{s_1}^{s_2} \left( {\E}_{{\cal H}_{s}}(\hat\psi) + {\E}_{{\cal H}_{s}}(\hat\psi)^3 \right) \frac{\d s}{\sqrt{s}} \, . \]
Then estimating the linear energy by the nonlinear energy and applying Proposition \ref{orderKerr0} entails
\begin{eqnarray}
{\E}_{{\cal H}_{s_1}}(\hat\psi) &\lesssim & {\E}_{{\cal H}_{s_2}}(\hat\psi) + \left( 1 + \tilde{\E}_{{\cal H}_{s_2}} (\hat\psi)^2 \right) \int_{s_1}^{s_2} {\E}_{{\cal H}_{s}}(\hat\psi) \frac{\d s}{\sqrt{s}} \, , \label{IntermEstZero1}\\
{\E}_{{\cal H}_{s_2}}(\hat\psi) &\lesssim & \left( {\E}_{{\cal H}_{s_1}}(\hat\psi) + {\E}_{{\cal S}_{u_0}^{s_1,s_2}}(\hat\psi) \right) \nonumber \\
&& + \left( 1 + \left( \tilde{\E}_{{\cal H}_{s_1}} (\hat\psi)+ \tilde{\E}_{{\cal S}^{s_1,s_2}_{^* t_0}} (\hat{\psi} ) \right)^2 \right) \int_{s_1}^{s_2} {\E}_{{\cal H}_{s}}(\hat\psi) \frac{\d s}{\sqrt{s}} \, . \label{IntermEstZero2}
\end{eqnarray}
Finally, Gronwall's lemma and \eqref{basic-ine} give the following result~:
\begin{proposition} \label{BasicNLlin}
There exists a function $\Phi \, :~ \R^+ \rightarrow \R^+$ such that, for ${^*t}_0 < 0, \, \vert {^*t}_0 \vert$ large enough and for any smooth compactly supported initial data at $t = 0$, the associated rescaled solution $\hat\psi$ satisfies for all $0 \leq s_1 < s_2 \leq 1$,
\begin{eqnarray*}
{\E}_{{\cal H}_{s_1}}(\hat\psi) &\leq & \Phi \left( \E_{{\cal H}_{s_2}}(\hat\psi) \right) {\E}_{{\cal H}_{s_2}}(\hat\psi )  \, , \\
{\E}_{{\cal H}_{s_2}}(\hat\psi) &\leq& \Phi \left( {\E}_{{\cal H}_{s_1}}(\hat\psi) + {\E}_{{\cal S}_{u_0}^{s_1,s_2}}(\hat\psi) \right)\left( {\E}_{{\cal H}_{s_1}}(\hat\psi) + {\E}_{{\cal S}_{u_0}^{s_1,s_2}}(\hat\psi) \right) \, .
\end{eqnarray*}
In particular for $s_1 = 0, \, s_2 = 1$,
\begin{eqnarray*}
{\E}_{\scri_{{^*t}_0}^+}(\hat\psi) &\leq & \Phi \left( \E_{{\cal H}_1}(\hat\psi) \right)  {\E}_{{\cal H}_{1}}(\hat\psi) \, , \\
{\E}_{{\cal H}_{1}}(\hat\psi) &\leq& \Phi \left( \tilde{\E}_{{\cal S}_{{^*t}_0}}(\hat\psi)+ {\E}_{\scri_{{^*t}_0}^+}(\hat\psi)\right) \left( \tilde{\E}_{{\cal S}_{{^*t}_0}}(\hat\psi)+ {\E}_{\scri_{{^*t}_0}^+}(\hat\psi)\right) \, .
\end{eqnarray*}
\end{proposition}
Note that Propositon \ref{orderKerr0} together with \eqref{basic-ine} yield sharper estimates~:
\begin{corollary}\label{inequality-basicnon2}
For ${^*t}_0 < 0, \, \vert {^*t}_0 \vert$ large enough and for any smooth compactly supported initial data at $t = 0$, the associated rescaled solution $\hat\psi$ satisfies for all $0 \leq s_1 < s_2 \leq 1$,
\begin{eqnarray*}
{\E}_{{\cal H}_{s_1}}(\hat\psi) &\lesssim& \left( {\E}_{{\cal H}_{s_2}}(\hat\psi) \right)^3 + {\E}_{{\cal H}_{s_2}}(\hat\psi) \, ,\\
{\E}_{{\cal H}_{s_2}}(\hat\psi) &\lesssim& \left( {\E}_{\mathcal{H}_{s_1}}(\hat\psi) \right)^3 + {\E}_{\mathcal{H}_{s_1}}(\hat\psi)  + \tilde{\E}_{{\cal S}_{{^*t}_0}^{s_1,s_2}}(\hat\psi) \, .
\end{eqnarray*}
In particular for $s_1 = 0, \, s_2 =1$ we get
\begin{eqnarray*}
{\E}_{\scri_{{^*t}_0}^+}(\hat\psi) &\lesssim & \left( {\E}_{{\cal H}_1}(\hat\psi) \right)^3 + {\E}_{{\cal H}_1}(\hat\psi) \, ,\\
{\E}_{{\cal H}_1}(\hat\psi) &\lesssim & \left( {\E}_{\scri_{{^*t}_0}^+}(\hat\psi) \right)^3 + {\E}_{\scri_{{^*t}_0}^+}(\hat\psi)  + \tilde{\E}_{{\cal S}_{{^*t}_0}}(\hat\psi) \, .
\end{eqnarray*}
\end{corollary}
For derivatives of $\hat\psi$ of arbitrarily high order, we cannot obtain estimates as explicit as in Corollary \ref{inequality-basicnon2} but we shall obtain results analogous to those of Proposition \ref{BasicNLlin}.

\subsubsection{Higher order estimates}

First, we control the energy of $\hat\psi_R$. Integrating \eqref{conser-non1} over the domain $\Omega^{s_1,s_2}_{^*t_0}$, $0 \leq s_1 < s_2 \leq 1$, defined in \eqref{Omegs1s2}, we obtain
\begin{gather}
\left| {\E}_{{\cal H}_{s_1}} (\hat\psi_R) + {\E}_{{\cal S}_{{^*t}_0}^{s_1, s_2}}(\hat\psi_R) - {\E}_{{\cal H}_{s_2}}(\hat\psi_R) \right| \nonumber\\
\lesssim \int_{s_1}^{s_2} \int_{{\cal H}_s} \left| \frac{r^2}{\rho^2}\left( 4a^2R \partial_{^*t}\partial_R\hat\psi + 4 aR \partial_{^*\varphi}\partial_R\hat\psi + 2a^2\partial_{^*t}\hat\psi \right) \nabla_T (\partial_R\hat\psi) \right| \frac{1}{|{^*t}|}  \d{^*t} \d^2\omega \d s \nonumber\\
+ \int_{s_1}^{s_2} \int_{{\cal H}_s} \left\vert \frac{r^2}{\rho^2}\left( 2a \partial_{^*\varphi}\hat\psi + (4a^2R^3-6MR^2+2R)\partial_R^2\hat\psi \right.\right. \hspace{1in} \nonumber\\
\hspace{1in}\left. + (12aR^2-12MR+2)\partial_R\hat\psi \right) \nabla_T (\partial_R\hat\psi) \left\vert  \frac{1}{|{^*t}|} \right. \d{^*t} \d^2\omega \d s \nonumber\\
+ \int_{s_1}^{s_2} \int_{{\cal H}_s} \left| \left( 2(M-2a^2R)\hat\psi + 2\frac{Mr-a^2}{r^2} \partial_R \hat\psi \right)  \nabla_T (\partial_R\hat\psi)\right|  \frac{1}{|{^*t}|} \d{^*t}\d^2\omega \d s \nonumber\\
+  \int_{s_1}^{s_2} \int_{{\cal H}_s} \left| \left(\hat\psi_R - 3\hat\psi^2\hat\psi_R - \frac{ra^2\cos^2\theta}{\rho^2}\hat\psi^3 \right) \nabla_T (\partial_R\hat\psi) \right|\frac{1}{|{^*t}|} \d{^*t}\d^2\omega \d s \nonumber\\
+ \int_{s_1}^{s_2} \int_{{\cal H}_s} \left|  \left(\nabla_a T_b  \right) T^{ab} (\partial_R \hat\psi)\right| \frac{1}{|{^*t}|} \d{^*t}\d^2\omega \d s \, . \label{control-non}
\end{gather}
The right-hand side can be controlled by the linear energy in the same manner as in the linear case, except for the following term which arises from the nonlinear part of the equation~:   
\[ E := \int_{s_1}^{s_2} \int_{{\cal H}_s} \left| \left( - 3\hat\psi^2\hat\psi_R - \frac{ra^2\cos^2\theta}{\rho^2}\hat\psi^3 \right) \nabla_T (\hat\psi_R) \right|  \frac{1}{|{^*t}|} \d{^*t}\d^2\omega \d s \, .\]
We now use the equivalence
\[ \frac{1}{|^*t|} \simeq \frac{1}{\sqrt s}\sqrt{\frac{R}{|^*t|}} \, , \]
as well as the classic inequality $ab\leq \frac12 (a^2+b^2)$ and H\"older estimates~:
\begin{eqnarray*}
E &= &\int_{s_1}^{s_2} \int_{{\cal H}_s} \left| \left( - 3\hat\psi^2\hat\psi_R - \frac{ra^2\cos^2\theta}{\rho^2}\hat\psi^3 \right) \left( {^*t}^2\partial_{^*t}\hat\psi_R - 2(1 + {^*t}R)\partial_R\hat\psi_R \right) \right|  \frac{1}{|{^*t}|} \d{^*t}\d^2\omega \d s\\
& \lesssim& \int_{s_1}^{s_2} \int_{{\cal H}_s}  \frac{3}{2} \left(\hat\psi^4\hat\psi_R^2 + {^*t}^2(\partial_{^*t}\hat\psi_R)^2  \right) + \frac{ra^2\cos^2\theta}{2\rho^2} \left( \hat\psi^6 + {^*t}^2(\partial_{^*t}\hat\psi_R)^2  \right) \d{^*t}\d^2\omega \d s \\
&&+ \int_{s_1}^{s_2} \int_{{\cal H}_s}  3(1 + {^*t}R)\frac{1}{\sqrt s} \left(\hat\psi^4\hat\psi_R^2 + \frac{R}{|{^*t}|}(\partial_R\hat\psi_R)^2  \right) \d{^*t}\d^2\omega \d s \\
&&+ \int_{s_1}^{s_2} \int_{{\cal H}_s} (1 + {^*t}R)\frac{ra^2\cos^2\theta}{\rho^2}\frac{1}{\sqrt s} \left( \hat\psi^6 + \frac{R}{|{^*t}|}(\partial_R\hat\psi_R)^2 \right) \d{^*t}\d^2\omega \d s \\
& \lesssim& \int_{s_1}^{s_2} \frac{1}{\sqrt s}{\E}_{{\cal H}_s}(\hat\psi_R) \d s + \int_{s_1}^{s_2} \frac{1}{\sqrt s} \int_{{\cal H}_s}\left(\hat\psi^4\hat\psi_R^2 + \hat\psi^6  \right) \d{^*t}\d^2\omega \d s \\
& \lesssim& \int_{s_1}^{s_2} \frac{1}{\sqrt s}{\E}_{{\cal H}_s}(\hat\psi_R) \d s + \int_{s_1}^{s_2} \frac{1}{\sqrt s} \left( \int_{{\cal H}_s} \hat\psi^6 \d{^*t}\d^2\omega \right)^{2/3} \left( \int_{{\cal H}_s} \hat\psi_R^6 \d{^*t}\d^2\omega \right)^{1/3} \d s \\
& &+ \int_{s_1}^{s_2} \frac{1}{\sqrt s}\int_{{\cal H}_s} \hat\psi^6 \d{^*t}\d^2\omega \d s \, .
\end{eqnarray*}
From Proposition \ref{EstimNLLinEn}, we infer
\[ E \lesssim \int_{s_1}^{s_2} \frac{1}{\sqrt s}{\E}_{{\cal H}_s}(\hat\psi_R) \d s + \int_{s_1}^{s_2} \frac{1}{\sqrt s} \left( {\E}_{{\cal H}_s}(\hat\psi) \right)^2 {\E}_{{\cal H}_s}(\hat\psi_R) \d s + \int_{s_1}^{s_2} \frac{1}{\sqrt s} \left( {\E}_{{\cal H}_s}(\hat\psi) \right)^3 \d s\, . \]
Using Proposition \ref{inequality-basicnon2}, we can estimate ${\E}_{{\cal H}_s}(\hat\psi)$ as follows
\[ {\E}_{{\cal H}_s}(\hat\psi) \lesssim \tilde{\E}_{{\cal H}_s}(\hat\psi) \lesssim \tilde{\E}_{{\cal H}_{s_2}}(\hat\psi) \]
and
\[ {\E}_{{\cal H}_s}(\hat\psi) \lesssim \tilde{\E}_{{\cal H}_s}(\hat\psi)  \lesssim \tilde{\E}_{\mathcal{H}_{s_1}}(\hat\psi)  + \tilde{\E}_{{\cal S}_{{^*t}_0}^{s_1,s_2}}(\hat\psi) \, , \]
from which we obtain
\begin{equation} \label{IntermA}
E \lesssim \left(1+\left( \tilde{\E}_{{\cal H}_{s_2}}(\hat\psi) \right)^2 \right) \int_{s_1}^{s_2} \frac{1}{\sqrt s} \left( {\E}_{{\cal H}_s}(\hat\psi_R) + {\E}_{{\cal H}_s}(\hat\psi) \right) \d s 
\end{equation}
and
\begin{equation} \label{IntermB}
E \lesssim \left( 1+ \left( \tilde{\E}_{\mathcal{H}_{s_1}}(\hat\psi)  + \tilde{\E}_{{\cal S}_{{^*t}_0}^{s_1,s_2}}(\hat\psi) \right)^2 \right) \int_{s_1}^{s_2} \frac{1}{\sqrt s} \left( {\E}_{{\cal H}_s}(\hat\psi_R) + {\E}_{{\cal H}_s}(\hat\psi) \right) \d s \, .
\end{equation}
Inequalities \eqref{control-non}, \eqref{IntermA} and \eqref{IntermB} give the two estimates~:
\begin{gather*}
\left| {\E}_{{\cal H}_{s_1}}(\hat\psi_R) + {\E}_{{\cal S}_{{^*t}_0}^{s_1,s_2}}(\hat\psi_R) - {\E}_{{\cal H}_{s_2}}(\hat\psi_R) \right| \hspace{3in} \\
 \hspace{1in}\lesssim \left( 1+\left( \tilde{\E}_{{\cal H}_{s_2}}(\hat\psi) \right)^2 \right) \int_{s_1}^{s_2} \left( {\E}_{{\cal H}_s}(\hat\psi_R) + {\E}_{{\cal H}_s}(\hat\psi) \right) \frac{\d s}{\sqrt s}
\end{gather*}
and
\begin{gather*}
 \left\vert {\E}_{{\cal H}_{s_1}}(\hat\psi_R) + {\E}_{{\cal S}_{{^*t}_0}^{s_1,s_2}}(\hat\psi_R) - {\E}_{{\cal H}_{s_2}}(\hat\psi_R) \right\vert \hspace{3in} \\
 \hspace{1in} \lesssim \left( 1+ \left( \tilde{\E}_{\mathcal{H}_{s_1}}(\hat\psi)  + \tilde{\E}_{{\cal S}_{{^*t}_0}^{s_1,s_2}}(\hat\psi) \right)^2 \right) \int_{s_1}^{s_2} \frac{1}{\sqrt s}\left( {\E}_{{\cal H}_s}(\hat\psi_R) + {\E}_{{\cal H}_s}(\hat\psi) \right) \d s \, ,
\end{gather*}
which, combined with \eqref{IntermEstZero1} and \eqref{IntermEstZero2}, in turn entail
\begin{eqnarray*}
{\E}_{{\cal H}_{s_1}}(\hat\psi) + {\E}_{{\cal H}_{s_1}}(\hat\psi_R) &\lesssim & {\E}_{{\cal H}_{s_2}}(\hat\psi) + {\E}_{{\cal H}_{s_2}}(\hat\psi_R) \\
&& + \left( 1+\left( \tilde{\E}_{{\cal H}_{s_2}}(\hat\psi) \right)^2 \right) \int_{s_1}^{s_2} \left( {\E}_{{\cal H}_s}(\hat\psi) + {\E}_{{\cal H}_s}(\hat\psi_R) \right) \frac{\d s}{\sqrt s}
\end{eqnarray*}
and
\begin{eqnarray*}
{\E}_{{\cal H}_{s_2}}(\hat\psi) + {\E}_{{\cal H}_{s_2}}(\hat\psi_R) &\lesssim & {\E}_{{\cal H}_{s_1}}(\hat\psi_R) + {\E}_{{\cal S}_{{^*t}_0}^{s_1,s_2}}(\hat\psi_R) \\
&& + \left( 1+\left( \tilde{\E}_{\scri_{{^*t}_0}^+}(\hat\psi)  + \tilde{\E}_{{\cal S}_{{^*t}_0}}(\hat\psi) \right)^2 \right) \int_{s_1}^{s_2} \left( {\E}_{{\cal H}_s}(\hat\psi) + {\E}_{{\cal H}_s}(\hat\psi_R) \right) \frac{\d s}{\sqrt s} \, .
\end{eqnarray*}
Using Gronwall's inequality, we obtain~:
\begin{proposition}\label{ordernon1}  
For ${^*t}_0 < 0, \, \vert {^*t}_0 \vert$ large enough and for any smooth compactly supported initial data at $t = 0$, the associated rescaled solution $\hat\psi$ satisfies for all $0 \leq s_1 < s_2 \leq 1$,
\begin{eqnarray*}
{\E}_{{\cal H}_{s_1}}(\hat\psi)+ {\E}_{{\cal H}_{s_1}}(\hat\psi_R) &\lesssim & \left( {\E}_{{\cal H}_{s_2}}(\hat\psi) + {\E}_{{\cal H}_{s_2}}(\hat\psi_R) \right) \times \mathrm{exp} \left( \left( 1+\left( \tilde{\E}_{{\cal H}_{s_2}}(\hat\psi) \right)^2 \right) \int_{s_1}^{s_2}  \frac{\d s}{\sqrt s} \right) \\
&\lesssim & \left( {\E}_{{\cal H}_{s_2}}(\hat\psi) + {\E}_{{\cal H}_{s_2}}(\hat\psi_R) \right) \times \mathrm{exp} \left(2 \left( 1+\left( \tilde{\E}_{{\cal H}_{s_2}}(\hat\psi) \right)^2 \right) \right) \, ,\\
{\E}_{{\cal H}_{s_2}}(\hat\psi) + {\E}_{{\cal H}_{s_2}}(\hat\psi_R) & \lesssim & \left( {\E}_{{\cal H}_{s_1}}(\hat\psi) + {\E}_{{\cal S}_{{^*t}_0}^{s_1,s_2}}(\hat\psi) + {\E}_{{\cal H}_{s_1}}(\hat\psi_R) + {\E}_{{\cal S}_{{^*t}_0}^{s_1,s_2}}(\hat\psi_R) \right) \\
&&\times \mathrm{exp} \left( 2 \left( 1+\left( {\E}_{{\cal H}_{s_1}}(\hat\psi) + {\E}_{{\cal S}_{{^*t}_0}^{s_1,s_2}}(\hat\psi) \right)^2 \right) \right) \, .
\end{eqnarray*}
In particular for $s_1 = 0, \, s_2 = 1$,
\begin{eqnarray*}
{\E}_{\scri_{{^*t}_0}^+}(\hat\psi)+ {\E}_{\scri_{{^*t}_0}^+}(\hat\psi_R) &\lesssim & \left( {\E}_{{\cal H}_{1}}(\hat\psi) + {\E}_{{\cal H}_{1}}(\hat\psi_R) \right) \times \mathrm{exp} \left(2 \left( 1+\left( \tilde{\E}_{{\cal H}_{1}}(\hat\psi) \right)^2 \right) \right) \, ,\\
{\E}_{{\cal H}_{1}}(\hat\psi) + {\E}_{{\cal H}_{1}}(\hat\psi_R) & \lesssim & \left( {\E}_{\scri_{{^*t}_0}^+}(\hat\psi) + {\E}_{{\cal S}_{{^*t}_0}}(\hat\psi) + {\E}_{\scri_{{^*t}_0}^+}(\hat\psi_R) + {\E}_{{\cal S}_{{^*t}_0}}(\hat\psi_R) \right) \\
&&\times \mathrm{exp} \left( 2 \left( 1+\left( \tilde{\E}_{\scri_{{^*t}_0}^+}(\hat\psi)  + \tilde{\E}_{{\cal S}_{{^*t}_0}}(\hat\psi) \right)^2 \right) \right) \, .
\end{eqnarray*}
\end{proposition}
We obtain similar estimates for $\mathcal{L}_{X_0}(\hat\psi) =\hat\psi_{^*t}$ and $\mathcal{L}_{X_1}(\hat\psi) = \hat\psi_{^*\varphi}$ from \eqref{conser-non-t} and \eqref{conser-non-varphi}.  As for the approximate conservation law \eqref{conser-non2} for ${\cal L}_{X_2}\hat\psi$ (resp. \eqref{conser-non3} for ${\cal L}_{X_3}\hat\psi$), it involves the derivatives of $\hat{\psi}$ along $X_0$ and $X_3$ (resp. $X_2$) in its error terms, but only at the linear level. We can therefore etablish combined estimates for the derivatives of $\hat{\psi}$ along the elements of ${\cal B} = \{ X_0,X_2,X_3 \}$ and of ${\cal A} = \{ X_0,X_1,X_2,X_3,X_4 \}$. Let us denote for an oriented hypersurface $S$
\[ \E_{S} ({\cal L}_{\cal B} \hat{\psi} ) := \sum_{i\in\{0,2,3\}} \E_{\cal S} ({\cal L}_{X_i} \hat{\psi} ) \, ,~~\E_{S} ({\cal L}_{\cal A} \hat{\psi} ) := \sum_{i=0}^4 \E_{\cal S} ({\cal L}_{X_i} \hat{\psi} ) \,  ,\]
we have~:
\begin{theorem}\label{ordernon3}
For ${^*t}_0 < 0, \, \vert {^*t}_0 \vert$ large enough and for any smooth compactly supported initial data at $t = 0$, the associated rescaled solution $\hat\psi$ satisfies~:
\begin{eqnarray*}
{\E}_{\scri_{{^*t}_0}^+}(\hat\psi) + \E_{\scri_{{^*t}_0}^+} ({\cal L}_{\cal B} \hat{\psi} ) & \lesssim & \left( {\E}_{{\cal H}_1}( \hat\psi) + \E_{{\cal H}_1} ({\cal L}_{\cal B} \hat{\psi} ) \right) \times \exp \left( 2 \left( 1+ \left( \tilde{\E}_{{\cal H}_1}(\hat\psi) \right)^2 \right) \right) \, ,\\
{\E}_{{\cal H}_1}(\hat\psi) + \E_{{\cal H}_1} ({\cal L}_{\cal B} \hat{\psi} ) &\lesssim & \left( {\E}_{\scri_{{^*t}_0}^+}(\hat\psi)  + {\E}_{{\cal S}_{{^*t}_0}}(\hat\psi)  + \E_{\scri_{{^*t}_0}^+} ({\cal L}_{\cal B} \hat{\psi} )+ \E_{{\cal S}_{{^*t}_0}} ({\cal L}_{\cal B} \hat{\psi} ) \right) \\
&&\times \exp \left( 2 \left( 1+ \left( \tilde{\cal E}_{\scri_{{^*t}_0}^+}(\hat\psi)  + \tilde{\cal E}_{{\cal S}_{{^*t}_0}}(\hat\psi) \right)^2\right) \right)
\end{eqnarray*}
and
\begin{eqnarray*}
{\E}_{\scri_{{^*t}_0}^+}(\hat\psi) + \E_{\scri_{{^*t}_0}^+} ({\cal L}_{\cal A} \hat{\psi} ) & \lesssim & \left( {\E}_{{\cal H}_1}( \hat\psi) + \E_{{\cal H}_1} ({\cal L}_{\cal A} \hat{\psi} ) \right) \times \exp \left( 2 \left( 1+ \left( \tilde{\E}_{{\cal H}_1}(\hat\psi) \right)^2 \right) \right) \, ,\\
{\E}_{{\cal H}_1}(\hat\psi) + \E_{{\cal H}_1} ({\cal L}_{\cal A} \hat{\psi} ) &\lesssim & \left( {\E}_{\scri_{{^*t}_0}^+}(\hat\psi)  + {\E}_{{\cal S}_{{^*t}_0}}(\hat\psi)  + \E_{\scri_{{^*t}_0}^+} ({\cal L}_{\cal A} \hat{\psi} )+ \E_{{\cal S}_{{^*t}_0}} ({\cal L}_{\cal A} \hat{\psi} ) \right) \\
&&\times \exp \left( 2 \left( 1+ \left( \tilde{\cal E}_{\scri_{{^*t}_0}^+}(\hat\psi)  + \tilde{\cal E}_{{\cal S}_{{^*t}_0}}(\hat\psi) \right)^2\right) \right)
\end{eqnarray*}
\end{theorem}
When commuting an arbitrary number of derivatives along the elements of $\cal A$ (let us denote it $D^\alpha$) into the equation \eqref{RescNLW}, $D^\alpha \hat\psi$ will satisfy an equation of the form
\[ \square_{\hat{g}} \D^\alpha \hat{\psi} = \mbox{RHS} \, .\]
In the right-hand side, the linear terms can be controlled as in the linear case. As for the non-linear terms, once we look at the flux of the energy current $\J (\D^\alpha \hat{\psi} )$ across the boundary of $\Omega^{s_1,s_2}_{^*t_0}$, they induce two types of error terms~:
\[ E_1 = \int_{s_1}^{s_2} \int_{{\cal H}_s} f \hat\psi_1 \hat\psi_2 \hat\psi_3 \, ^*t^2 \partial_{^*t} \left( \D^\alpha \hat{\psi} \right) \frac{1}{\vert \, ^*t\vert}\d ^*t \d \omega \d s \]
and
\[ E_2=\int_{s_1}^{s_2} \int_{{\cal H}_s} f \hat\psi_1 \hat\psi_2 \hat\psi_3 \, \partial_{R} \left( \D^\alpha \hat{\psi} \right) \frac{1}{\vert \, ^*t\vert}\d ^*t \d \omega \d s \, , \]
where $f$ is a bounded function and the $\hat\psi_i$'s are either $D^\alpha\hat\psi$ or lower order derivatives of $\hat\psi$ that involve the same vector fields as $D^\alpha$. The $E_1$-type terms are dealt with by standard Hölder estimates as follows
\begin{eqnarray*}
E_1 &\leq & \int_{s_1}^{s_2} \int_{{\cal H}_s} \vert f\vert \frac12 \left(  (\hat\psi_1 \hat\psi_2 \hat\psi_3)^2 +  (^*t \partial_{^*t} ( \D^\alpha \hat{\psi} ))^2 \right) \d ^*t \d \omega \d s \\
&\lesssim & \int_{s_1}^{s_2} \Vert f \Vert_{L^\infty ({\cal H}_s)} \left( \Pi_{i=1}^3 \Vert \hat\psi_i \Vert_{L^6({\cal H}_s)}^2 + \Vert ^*t \partial_{^*t} D^\alpha \hat\psi_i \Vert_{L^2({\cal H}_s)}^2  \right) \d s \, ,
\end{eqnarray*}
then, using Proposition \ref{EstimNLLinEn}, we obtain
\[ E_1 \lesssim  \int_{s_1}^{s_2} \left( \Pi_{i=1}^3 \E(\hat\psi_i) + \E ( \D^\alpha \hat{\psi} ) \right) \d s \, .\]
The $E_2$-type terms are estimated analogously with the additional ingredient of the familiar equivalence
\[ \frac{1}{|^*t|} \simeq \frac{1}{\sqrt s}\sqrt{\frac{R}{|^*t|}} \, .\]
We obtain
\begin{eqnarray*}
E_2 &\leq & \int_{s_1}^{s_2} \frac{1}{\sqrt{s}} \int_{{\cal H}_s} \vert f\vert \frac12 \left(  (\hat\psi_1 \hat\psi_2 \hat\psi_3)^2 +  \left( \frac{R}{|^*t|} \partial_{R} ( \D^\alpha \hat{\psi} )\right)^2 \right) \d ^*t \d \omega \d s \\
& \lesssim & \int_{s_1}^{s_2} \frac{1}{\sqrt{s}} \left( \Pi_{i=1}^3 \E(\hat\psi_i) + \E ( \D^\alpha \hat{\psi} ) \right) \d s \, .
\end{eqnarray*}
Depending whether we differentiate merely along $\partial_R$, along the elements of $\cal B$ or along all the elements of $\cal A$, we obtain three types of higher order estimates.
\begin{theorem} \label{NLEstimates}
There exists a sequence of continuous functions $F_p \, :~\R^{p+1} \rightarrow ]0,+\infty[$, $p\in \N$, such that, for any smooth compactly supported initial data, the associated solution of \eqref{RescNLW} satisfies for all $p \in \N^*$~:
\begin{gather*}
\sum_{k=0}^p \E_{\scri^+_{^*t_0}} (\partial_R^k \hat\psi ) \lesssim  F_{p-1} \left( \tilde{\E}_{\mathcal{H}_1}(\hat\psi) , \tilde{\E}_{\mathcal{H}_1}(\partial_R \hat\psi) , ... , \tilde{\E}_{\mathcal{H}_1}(\partial_R^{p-1}\hat\psi) \right) \sum_{k=0}^p \E_{\mathcal{H}_1}(\partial_R^k \hat\psi) \, ,\\
\sum_{k=0}^p \E_{{\cal H}_1} (\partial_R^k \hat\psi ) + \E_{{\cal H}_1} (\partial_R^k \hat\psi ) \lesssim  F_{p-1} \left( \tilde{\E}_{\scri^+_{^*t_0}}(\hat\psi) + \tilde{\E}_{{\cal S}_{^*t_0}}(\hat\psi) , \tilde{\E}_{\scri^+_{^*t_0}}(\partial_R \hat\psi) + \tilde{\E}_{{\cal S}_{^*t_0}}(\partial_R \hat\psi) , ... \right. \\
\left. ..., \tilde{\E}_{\scri^+_{^*t_0}}(\partial_R^{p-1}\hat\psi) + \tilde{\E}_{{\cal S}_{^*t_0}}(\partial_R^{p-1} \hat\psi) \right) \sum_{k=0}^p \left( \E_{\scri^+_{^*t_0}}(\partial_R^k \hat\psi) + \E_{{\cal S}_{^*t_0}}(\partial_R^k \hat\psi) \right) \, ,
\end{gather*}
also
\begin{gather*}
\sum_{k=0}^p \E_{\scri^+_{^*t_0}} ({\cal L}_{\cal B}^k \hat\psi ) \lesssim  F_{p-1} \left( \tilde{\E}_{\mathcal{H}_1}(\hat\psi) , \tilde{\E}_{\mathcal{H}_1}({\cal L}_{\cal B} \hat\psi) , ... , \tilde{\E}_{\mathcal{H}_1}({\cal L}_{\cal B}^{p-1}\hat\psi) \right) \sum_{k=0}^p \E_{\mathcal{H}_1}({\cal L}_{\cal B}^k \hat\psi) \, ,\\
\sum_{k=0}^p \E_{{\cal H}_1} ({\cal L}_{\cal B}^k \hat\psi ) + \E_{{\cal H}_1} ({\cal L}_{\cal B}^k \hat\psi ) \lesssim  F_{p-1} \left( \tilde{\E}_{\scri^+_{^*t_0}}(\hat\psi) + \tilde{\E}_{{\cal S}_{^*t_0}}(\hat\psi) , \tilde{\E}_{\scri^+_{^*t_0}}({\cal L}_{\cal B} \hat\psi) + \tilde{\E}_{{\cal S}_{^*t_0}}({\cal L}_{\cal B} \hat\psi) , ... \right. \\
\left. ..., \tilde{\E}_{\scri^+_{^*t_0}}({\cal L}_{\cal B}^{p-1}\hat\psi) + \tilde{\E}_{{\cal S}_{^*t_0}}({\cal L}_{\cal B}^{p-1} \hat\psi) \right) \sum_{k=0}^p \left( \E_{\scri^+_{^*t_0}}({\cal L}_{\cal B}^k \hat\psi) + \E_{{\cal S}_{^*t_0}}({\cal L}_{\cal B}^k \hat\psi) \right)
\end{gather*}
and finally
\begin{gather*}
\sum_{k=0}^p \E_{\scri^+_{^*t_0}} ({\cal L}_{\cal A}^k \hat\psi ) \lesssim  F_{p-1} \left( \tilde{\E}_{\mathcal{H}_1}(\hat\psi) , \tilde{\E}_{\mathcal{H}_1}({\cal L}_{\cal A} \hat\psi) , ... , \tilde{\E}_{\mathcal{H}_1}({\cal L}_{\cal A}^{p-1}\hat\psi) \right) \sum_{k=0}^p \E_{\mathcal{H}_1}({\cal L}_{\cal A}^k \hat\psi) \, ,\\
\sum_{k=0}^p \E_{{\cal H}_1} ({\cal L}_{\cal A}^k \hat\psi ) + \E_{{\cal H}_1} ({\cal L}_{\cal A}^k \hat\psi ) \lesssim  F_{p-1} \left( \tilde{\E}_{\scri^+_{^*t_0}}(\hat\psi) + \tilde{\E}_{{\cal S}_{^*t_0}}(\hat\psi) , \tilde{\E}_{\scri^+_{^*t_0}}({\cal L}_{\cal A} \hat\psi) + \tilde{\E}_{{\cal S}_{^*t_0}}({\cal L}_{\cal A} \hat\psi) , ... \right. \\
\left. ..., \tilde{\E}_{\scri^+_{^*t_0}}({\cal L}_{\cal A}^{p-1}\hat\psi) + \tilde{\E}_{{\cal S}_{^*t_0}}({\cal L}_{\cal A}^{p-1} \hat\psi) \right) \sum_{k=0}^p \left( \E_{\scri^+_{^*t_0}}({\cal L}_{\cal A}^k \hat\psi) + \E_{{\cal S}_{^*t_0}}({\cal L}_{\cal A}^k \hat\psi) \right)
\end{gather*}
\end{theorem}

\subsubsection{Peeling}

The estimates of Proposition \ref{BasicNLlin} and Theorem \ref{NLEstimates} are a priori valid only for solutions of \eqref{RescNLW} associated with smooth compactly supported data on ${\cal H}_1$. Since the equation is nonlinear, in order to extend the estimates by density to less regular frameworks, we need some estimates on the difference of two solutions. These can be proved in the same manner as we obtained estimates on a given solution, using the Sobolev embeddings of Proposition \ref{EstimNLLinEn}. For example, as an analogue of Proposition \ref{BasicNLlin}, we have for two solutions $\hat\psi$ and $\hat\phi$ associated with smooth compactly supported data on ${\cal H}_1$ and for $0\leq s_1 \leq s_2 \leq 1$~:
\[ {\E}_{{\cal H}_{s_1}} (\hat\psi - \hat\phi) \leq  \Phi \left( \E_{{\cal H}_{s_2}}(\hat\psi) , \E_{{\cal H}_{s_2}}(\hat\phi) \right)  {\E}_{{\cal H}_{s_2}}(\hat\psi - \hat\phi) \, . \]
Similar estimates can be established at any order, using the same ingredients as for the proof of Theorem \ref{NLEstimates}. The estimates of Theorem \ref{NLEstimates} therefore extend to solutions arising from data in Sobolev spaces on ${\cal H}_1$ of sufficiently high order so that the derivatives considered have finite energy on the initial data slice.

In the linear case, we had expressed the finiteness of the energy on ${\cal H}_1$ as the finiteness of a norm that was intrinsic to ${\cal H}_1$ since it involved only derivatives tangent to ${\cal H}_1$. This was done using the equation to transform the transverse derivatives into combinations of tangential ones, the operators $L$ and $N$ of \eqref{OpL} and \eqref{OpN}. We can still do the same in the nonlinear case, but the operators $L$ and $N$ will now become nonlinear and the quantities appearing in Theorem \ref{LinearPeeling} no longer define norms.


\begin{thebibliography}{100}

\bibitem{Au} T. Aubin, {\em Nonlinear analysis on manifolds. Monge-Ampère equations}, Grundlehren der Mathematischen Wissenschaften, 252, Springer-Verlag, New York, 1982.

\bibitem{CaCho} F. Cagnac, Y. Choquet-Bruhat, {\em Solution globale d'une équation non linéaire sur une variété hyperbolique}, J. Math. Pures Appl. {\bf 63} (1984), 9,  377--390.

\bibitem{ChriKla90} D. Christodoulou, S. Klainerman, {\em Asymptotic properties of linear field equations in Minkowski space}, Comm. Pure Appl. Math. {\bf 43}
(1990), 137--199.

\bibitem{ChriKla} D. Christodoulou \& S. Klainerman, {\em The global nonlinear stability of the Minkowski space}, Princeton Mathematical Series 41, Princeton University Press 1993.

\bibitem{ChruDe2002} P. Chru\'sciel \& E. Delay, {\em Existence of non trivial, asymptotically vacuum, asymptotically simple space-times}, Class. Quantum Grav. {\bf 19} (2002), L71-L79, erratum Class. Quantum Grav. {\bf 19} (2002), 3389.

\bibitem{ChruDe2003} P. Chru\'sciel \& E. Delay, {\em On mapping properties of the general relativistic constraints operator in weighted function spaces, with applications}, preprint Tours Univervity, 2003.

\bibitem{Co2000} J. Corvino, {\em Scalar curvature deformation and a gluing construction for the Einstein constraint equations}, Comm. Math. Phys. {\bf 214} (2000), 137-189.

\bibitem{CoScho2003} J. Corvino \& R.M. Schoen, {\em On the asymptotics for the vacuum Einstein constraint equations}, gr-qc 0301071, 2003.

\bibitem{DaRoLN} M. Dafermos, I. Rodnianski, {\em Lectures on black holes and linear waves}, Evolution equations, 97--205, Clay Math. Proc., 17, Amer. Math. Soc., Providence, RI, 2013, arXiv:0811.0354, 2008.

\bibitem{DaRo2009} M. Dafermos, I. Rodnianski, {\em The red-shift effect and radiation decay on black hole spacetimes}, Comm. Pure Appl. Math. {\bf 62} (2009), 7, 859--919.

\bibitem{FleLu2003} S. J. Flechter, A.W.C. Lun, {\em The Kerr spacetime in generalized Bondi-Sachs coordinates}, Class. Quantum Grav. {\bf 20} (2003), 4153--4167.

\bibitem{HFri2004} H. Friedrich, {\em Smoothness at null infinity and the structure of initial data}, in The Einstein equations and the large scale behavior of gravitational fields, p. 121-203, Ed. P. Chru\'sciel and H. Friedrich, Birkha\"user, Basel, 2004.

\bibitem{GS} J. N. Goldberg and R. K. Sachs, {\em A theorem on Petrov types}, Acta Phys. Polon., {\bf 22} (1962), 13--23, Suppl.

\bibitem{Ha2009} D. H\"afner, {\em Creation of fermions by rotating charged black holes}, Memoirs of the SMF {\bf 117} (2009), 158 pp,  arXiv:math/0612501.

\bibitem{InglNi} W. Inglese, F. Nicolò, {\em Asymptotic properties of the electromagnetic field in the external Schwarzschild spacetime}, Ann. Henri Poincaré {\bf 1} (2000), 5, 895-944.

\bibitem{Jo2012} J. Joudioux, {\em Conformal scattering for a nonlinear wave equation}, J. Hyperbolic Differ. Equ. {\bf 9} (2012), 1, 1--65.

\bibitem{KlaNi} S. Klainerman \& F. Nicol{\`o}, {\em On local and global aspects of the Cauchy problem in general relativity}, Class. Quantum Grav. {\bf 16} (1999), p. R73-R157.

\bibitem{KlaNi2002} S. Klainerman \& F. Nicol{\`o}, The Evolution Problem in General Relativity, Progress in Mathematical Physics Vol. 25, Birkha\"user, 2002.

\bibitem{KlaNi2003} S. Klainerman \& F. Nicol{\`o}, {\em Peeling properties of asymptotically flat solutions to the Einstein vacuum equations}, Class. Quantum Grav. {\bf 20} (2003), p. 3215-3257.

\bibitem{Le1953} J. Leray, {\em Hyperbolic differential equations}, lecture notes, Princeton Institute for Advanced Studies, 1953.

\bibitem{Ma} L.J. Mason, {\em On Ward's integral formula for the wave equation in plane-wave spacetimes}, Twistor Newsletter {\bf 28} (1989), 17--19.

\bibitem{MaNi2004} L.J. Mason, J.-P. Nicolas, {\em Conformal scattering and the Goursat problem}, J. Hyperbolic Differ. Equ. {\bf 1} (2004), 2, 197--233.

\bibitem{MaNi2009} L.J. Mason, J.-P. Nicolas, {\em Regularity an space-like and null infinity}, J. Inst. Math. Jussieu {\bf 8} (2009), 1, 179--208.

\bibitem{MaNi2012} L.J. Mason, J.-P. Nicolas, {\em Peeling of Dirac and Maxwell fields on a Schwarzschild background}, arXiv:1101.4333, J. Geom. Phys. {\bf 62} (2012), 4, 867--889.

\bibitem{MeTaTo}  J. Metcalfe, D. Tataru, M. Tohaneanu, {\em Pointwise decay for the Maxwell field on black hole space-times}, Adv. Math. {\bf 316} (2017), 53--93.

\bibitem{Mo1961} C. S. Morawetz, {\em The decay of solutions of the exterior initial-boundary value problem for the wave equation}, Commun. Pure Appl. Math. {\bf 14} (1961), 561--568.

\bibitem{NP} E. T. Newman, R. Penrose, {\em An approach to gravitational radiation by a method of spin coefficients}, J. Math. Phys. {\bf 3} (1962), 566--768.

\bibitem{Ni2002} J.-P. Nicolas, {\em A nonlinear Klein-Gordon equation on Kerr metrics}, J. Math Pures et Appliqu\'ees, {\bf 81} (2002), 9, 885--914.

\bibitem{O95} B. O'Neill,  {\em The geometry of Kerr black holes}, A.K. Peters, Wellesley, 1995.

\bibitem{Pe1963} R. Penrose, {\em Asymptotic properties of fields and spacetime}, Phys. Rev. Lett. {\bf 10} (1963), 66--68.

\bibitem{Pe1964} R. Penrose, {\em Conformal approach to infinity}, in Relativity, groups and topology, Les Houches 1963, ed. B.S. De Witt and C.M. De Witt, Gordon and Breach, New-York, 1964.

\bibitem{Pe1965} R. Penrose, {\em Zero rest-mass fields including gravitation~: asymptotic behaviour}, Proc. Roy. Soc. London {\bf A284} (1965), 159--203.

\bibitem{PeRi} R. Penrose, W. Rindler, {\em Spinors and space-time}, Vol. I \& II, Cambridge monographs on mathematical physics, Cambridge University Press 1984 \& 1986.

\bibitem{Petrov} A.Z. Petrov, {\em The classification of spaces defining gravitational fields}, Scientific Proceedings of Kazan State University (named after V.I. Ulyanov-Lenin), Jubilee (1804-1954) Collection {\bf 114} (1954),8 ,55-69, translation by J. Jezierski and M.A.H. MacCallum, with introduction, by M.A.H. MacCallum, Gen. Rel. Grav. 32 (2000), 1661--1685.
  
\bibitem{Sa61} R. Sachs, {\em Gravitational waves in general relativity VI, the outgoing radiation condition}, Proc. Roy. Soc. London {\bf A264} (1961), 309-338.

\bibitem{Sa62} R. Sachs, {\em Gravitational waves in general relativity VIII, waves in asymptotically flat space-time}, Proc. Roy. Soc. London {\bf A270} (1962), 103--126.

\end{thebibliography}
\end{document}